\documentclass[11pt]{article}

\usepackage{graphicx,amssymb,amsmath}
\usepackage{amsfonts}
\usepackage{booktabs}
\usepackage{cite}
\usepackage{color}
\usepackage{enumerate}
\usepackage{float}
\usepackage{hyperref}
\usepackage{mathrsfs}
\usepackage{subfig}
\usepackage{tabularx}
\usepackage{authblk}

\usepackage[in]{fullpage}
\usepackage[top=0.8in, bottom=0.9in, left=0.9in, right=0.9in]{geometry}

\newcommand{\vd}{\mbox{$V\!D$}}

\def\calC{\mathcal{C}}
\def\calS{\mathcal{S}}
\def\calA{\mathcal{A}}
\def\scrD{\mathscr{D}}
\def\calD{\mathcal{D}}
\def\calP{\mathcal{P}}
\newcommand{\poly}{\text{poly}}
\newtheorem{observation}{Observation}

\newtheorem{lemma}{Lemma}
\newtheorem{theorem}{Theorem}
\newtheorem{corollary}{Corollary}
\newenvironment{proof}{\noindent {\textbf{Proof:}}\rm}{\hfill $\Box$ \rm\bigskip}

\title{Unit-Disk Range Searching and Applications\thanks{This research was supported in part by NSF under Grant CCF-2005323. A preliminary version of this paper will appear
in Proceedings of the 18th Scandinavian Symposium and Workshops on Algorithm Theory (SWAT 2022).}}

\author{
Haitao Wang
}
\affil{Department of Computer Science \\
Utah State University, Logan, UT 84322, USA
\\ {\tt haitao.wang@usu.edu}}

\begin{document}

\pagestyle{plain}
\pagenumbering{arabic}
\setcounter{page}{1}
\date{}

\thispagestyle{empty}
\maketitle

\vspace{-0.45in}

\begin{abstract}
Given a set $P$ of $n$ points in the plane,
we consider the problem of computing the number of points of $P$ in a query unit disk (i.e., all query disks have the same radius).
We show that the main techniques for simplex range searching in the plane can be adapted to this problem. For example, by adapting Matou\v{s}ek's results, we can build a data structure of $O(n)$ space
so that each query can be answered in $O(\sqrt{n})$ time; alternatively, we can build a data structure of $O(n^2/\log^2 n)$ space with $O(\log n)$ query time.
Our techniques lead to improvements for several other classical problems in computational geometry.
\begin{enumerate}
  \item Given a set of $n$ unit disks and a set of $n$ points in the plane, the batched unit-disk range counting problem is to compute for each disk the number of points in it. Previous work [Katz and Sharir, 1997] solved the problem in $O(n^{4/3}\log n)$ time.
      We give a new algorithm of $O(n^{4/3})$ time, which is optimal as it matches an $\Omega(n^{4/3})$-time lower bound. For small $\chi$, where $\chi$ is the number of pairs of unit disks that intersect, we further improve the algorithm to $O(n^{2/3}\chi^{1/3}+n^{1+\delta})$ time, for any $\delta>0$.

\item
The above result immediately leads to an $O(n^{4/3})$ time optimal algorithm for counting the intersecting pairs of circles for a set of $n$ unit circles in the plane. The previous best algorithms solve the problem in $O(n^{4/3}\log n)$ deterministic time [Katz and Sharir, 1997] or in $O(n^{4/3}\log^{2/3} n)$ expected time by a randomized algorithm [Agarwal, Pellegrini, and Sharir, 1993].

    \item
    Given a set $P$ of $n$ points in the plane and an integer $k$, the distance selection problem is to find the $k$-th smallest distance among all pairwise distances of $P$. The problem can be solved in $O(n^{4/3}\log^2 n)$ deterministic time [Katz and Sharir, 1997] or in $O(n\log n+n^{2/3}k^{1/3}\log^{5/3}n)$ expected time by a randomized algorithm [Chan, 2001]. Our new randomized algorithm runs in $O(n\log n +n^{2/3}k^{1/3}\log n)$ expected time.

  \item
  Given a set $P$ of $n$ points in the plane, the discrete $2$-center problem is to compute two smallest congruent disks whose centers are in $P$ and whose union covers $P$. An $O(n^{4/3}\log^5 n)$-time algorithm was known [Agarwal, Sharir, and Welzl, 1998]. Our techniques yield a deterministic algorithm of $O(n^{4/3}\log^{10/3} n\cdot (\log\log n)^{O(1)})$ time and a randomized algorithm of $O(n^{4/3}\log^3 n\cdot (\log\log n)^{1/3})$ expected time.
\end{enumerate}
\end{abstract}



\section{Introduction}
\label{sec:intro}

We consider the range counting queries for unit disks. Given a set $P$ of $n$ points in the plane, the problem is to build a data structure so that the number of points of $P$ in $D$ can be computed efficiently for any  query unit disk $D$ (i.e., all query disks have the same known radius).

Our problem is a special case of the general disk range searching problem in which each query disk may have an arbitrary radius. Although we are not aware of any previous work particulary for our special case, the general problem has been studied before~\cite{ref:AgarwalOn94,ref:AgarwalOn13,ref:ChazelleQu89,ref:YaoA83,ref:MatousekMu15}. First of all, it is well-known that the lifting method can reduce the disk range searching in the $d$-dimensional space to half-space range searching in the $(d+1)$-dimensional space; see, e.g.,~\cite{ref:EdelsbrunnerAl87,ref:YaoA83}. For example, using Matou\v{s}ek's results in 3D~\cite{ref:MatousekRa93}, with $O(n)$ space, each disk query in the plane can be answered in $O(n^{2/3})$ time. Using the randomized results for general semialgebraic range searching~\cite{ref:AgarwalOn13,ref:MatousekMu15}, one can build a data structure of $O(n)$ space in $O(n^{1+\delta})$ expected time that can answer each disk query in $O(\sqrt{n}\log^{O(1)}n)$ time, where (and throughout the paper) $\delta$ denotes any small positive constant. For deterministic results, Agarwal and Matou\v{s}ek's techniques~\cite{ref:AgarwalOn94} can build a data structure of $O(n)$ space in $O(n\log n)$ time, and each query can be answered in $O(n^{1/2+\delta})$ time.

A related problem is to report all points of $P$ in a query disk. If all query disks are unit disks, the problem is known as {\em fixed-radius neighbor problem} in the literature~\cite{ref:BentleyA79,ref:ChazelleAn83,ref:ChazelleNe86,ref:ChazelleOp85}.
Chazelle and Edelsbrunner~\cite{ref:ChazelleOp85} gave an optimal solution (in terms of space and query time): they constructed a data structure of $O(n)$ space that can answer each query in $O(\log n+k)$ time, where $k$ is the output size; their data structure can be constructed in $O(n^2)$ time.
By a standard lifting transformation that reduces the problem to the halfspace range reporting queries in 3D, Chan and Tsakalidis~\cite{ref:ChanOp16} constructed a data structure of $O(n)$ space in $O(n\log n)$ time that can answer each query in $O(\log n+k)$ time; the result also applies to the general case where the query disks may have arbitrary radii.
Refer to~\cite{ref:AgarwalRa17,ref:AgarwalSi17,ref:MatousekGe94} for excellent surveys on range searching.

In this paper, we focus on unit-disk counting queries. By taking advantage of the property that all query disks have the same known radius, we manage to adapt the techniques for simplex range searching to our problem. We show that literally all main results for simplex range searching can be adapted to our problem with asymptotically the same performance. For example, by adapting Matou\v{s}ek's result in~\cite{ref:MatousekEf92}, we build a data structure of $O(n)$ space in $O(n\log n)$ time and each query can be answered in $O(\sqrt{n}\log^{O(1)}n)$ time. By adapting Matou\v{s}ek's result in~\cite{ref:MatousekRa93}, we build a data structure of $O(n)$ space in $O(n^{1+\delta})$ time and each query can be answered in $O(\sqrt{n})$ time. By adapting Chan's randomized result in~\cite{ref:ChanOp12}, we build a data structure of $O(n)$ space in $O(n\log n)$ expected time and each query can be answered in $O(\sqrt{n})$ time with high probability.

In addition, we obtain the following trade-off: After $O(nr)$ space and $O(nr(n/r)^{\delta})$ time preprocessing, each query can be answered in $O(\sqrt{n/r})$ time, for any $1\leq r\leq n/\log^2 n$.
In particularly, setting $r=n/\log^2 n$, we can achieve $O(\log n)$ query time, using $O(n^2/\log^2 n)$ space and $O(n^2/\log^{2-\delta}n)$ preprocessing time. To the best of our knowledge, the only previous work we are aware of with $O(\log n)$ time queries for the disk range searching problem is a result in~\cite{ref:KatzAn97},\footnote{See Theorem 3.1~\cite{ref:KatzAn97}. The authors noted in their paper that the result was due to Pankaj K. Agarwal.} which can answer each general disk query in $O(\log n)$ time with $O(n^2\log n)$ time and space preprocessing.


Probably more interestingly to some extent, our techniques can be used to derive improved algorithms for several classical problems, as follows. Our results are the first progress since the previous best algorithms for these problems were proposed over two decades ago.

\paragraph{Batched unit-disk range counting.} Let $P$ be a set of $n$ points and $\calD$ be a set of $m$ (possibly overlapping) congruent disks in the plane. The problem is to compute for all disks $D\in \calD$ the number of points of $P$ in $D$.
The algorithm of Katz and Sharir~\cite{ref:KatzAn97} solves the problem in
$O((m^{2/3}n^{2/3}+m+n)\log n)$ time.
By using our techniques for unit-disk range searching and adapting a recent result of Chan and Zheng~\cite{ref:ChanHo22}, we obtain a new algorithm of $O(n^{2/3}m^{2/3}+m\log n+n\log m)$ time.
We further improve the algorithm so that the complexities are sensitive to $\chi$, the number of pairs of disks of $\calD$ that intersect. The runtime of the algorithm is $O(n^{2/3}\chi^{1/3}+m^{1+\delta}+n\log n)$.

On the negative side, Erickson~\cite{ref:EricksonNe96} proved a lower bound of $\Omega(n^{2/3}m^{2/3}+m\log n+n\log m)$ time for the problem in a so-called {\em partition algorithm model}, even if
each disk is a half-plane (note that a half-plane can be considered as a special unit disk of infinite radius). Therefore, our algorithm is optimal under Erickson's model.

\paragraph{Counting intersections of congruent circles.}
As discussed in~\cite{ref:KatzAn97}, the following problem can be
immediately solved using batched unit-disk range counting:
Given a set of $n$ congruent circles of radius $r$ in the
plane, compute the number of intersecting pairs. To do so, define $P$ as the set of the centers of circles
and define $\calD$ as the set of congruent disks centered at points of
$P$ with radius $2r$. Then apply the batched unit-disk range counting
algorithm on $P$ and $\calD$. The algorithm runs in $O(n^{4/3})$ time, matching an $\Omega(n^{4/3})$-time lower bound~\cite{ref:EricksonNe96}.
To the best of our knowledge, the previous best results for this problem are a deterministic algorithm of $O(n^{4/3}\log n)$ time~\cite{ref:KatzAn97} and a randomized algorithm of $O(n^{4/3}\log^{2/3} n)$ expected time~\cite{ref:AgarwalSe93}. Agarwal, Pellegrini, and Sharir~\cite{ref:AgarwalCo93} also studied the problem for circles of different radii and gave an $O(n^{3/2+\delta})$ time deterministic algorithm.

\paragraph{Distance selection.}
Let $P$ be a set of $n$ points in the plane. Given an integer $k$ in the range $[1,n(n-1)/2]$, the
distance selection problem is to find the $k$-th smallest distance among all pairwise distances
of $P$; let $\lambda^*$ denote the $k$-th smallest distance. Given a value $\lambda$, the decision problem is to decide whether $\lambda\geq \lambda^*$.
We refer to the original problem as the optimization problem.

Chazelle~\cite{ref:ChazelleNe85} gave the first subquadratic algorithm of $O(n^{9/5}\log^{4/5}n)$ time. Agarwal, Aronov, Sharir, and Suri~\cite{ref:AgarwalSe93} presented randomized algorithms that solve the decision problem in $O(n^{4/3}\log^{2/3}n)$ expected time and the optimization problem in $O(n^{4/3}\log^{8/3}n)$ expected time, respectively. Goodrich~\cite{ref:GoodrichGe93} later gave a deterministic algorithm of $O(n^{4/3}\log^{8/3}n)$ time for the optimization problem. Katz and Sharir~\cite{ref:KatzAn97} proposed a deterministic algorithm of $O(n^{4/3}\log n)$ time for the decision problem and used it to solve the optimization problem in $O(n^{4/3}\log^2 n)$ deterministic time. Using the decision algorithm of \cite{ref:AgarwalSe93}, Chan's randomized technique~\cite{ref:ChanOn01} solved the optimization problem in $O(n\log n+n^{2/3}k^{1/3}\log^{5/3}n)$ expected time.

Our algorithm for the batched unit-disk range counting problem can be used to solve the decision problem in $O(n^{4/3})$ time.
Combining it with the randomized technique of Chan~\cite{ref:ChanOn01}, the optimization problem can now be solved in $O(n\log n + n^{2/3}k^{1/3}\log n)$ expected time.

\paragraph{Discrete $2$-center.} Let $P$ be a set of $n$ points in the
plane. The discrete $2$-center problem is to find two smallest congruent disks whose centers are in $P$ and whose union covers $P$. Agarwal, Sharir, and
Welzl~\cite{ref:AgarwalTh98} gave an $O(n^{4/3}\log^5 n)$-time
algorithm. Using our techniques for unit-disk range searching, we reduce the time of their algorithm to $O(n^{4/3}\log^{10/3} n(\log\log n)^{O(1)})$ deterministic time or to $O(n^{4/3}\log^3
n(\log\log n)^{1/3})$ expected time by a randomized algorithm.


\medskip
In the following, we present our algorithms for unit-disk range searching in Section~\ref{sec:diskrange}. The other problems are discussed in Section~\ref{sec:app}. Section~\ref{sec:con} concludes the paper.

\section{Unit-disk range searching}
\label{sec:diskrange}
In this section, we present our algorithms for unit-disk range searching problem. Our goal is to show that the main techniques for simplex range searching can be used to solve our problem. In particular, we show that, after overcoming many difficulties, the techniques of Matou\v{s}ek in \cite{ref:MatousekEf92} and \cite{ref:MatousekRa93} as well as the results of Chan~\cite{ref:ChanOp12} can be adapted to our problem with asymptotically the same performance.

We assume that the radius of unit disks is $1$. In the rest of this section, unless otherwise stated, a disk refers to a unit disk. We begin with an overview of our approach.

\paragraph{An overview.}
We roughly (but not precisely) discuss the main idea. We first implicitly build a grid $G$ of side length $1/\sqrt{2}$ such that any query disk $D$ only intersects $O(1)$ cells of $G$. This means that it suffices to build a data structure for the subset $P(C')$ of the points of $P$ in each individual cell $C'$ of $G$ with respect to query disks whose centers are in another cell $C$ that is close to $C'$.
A helpful property for processing $P(C')$ with respect to $C$ is that for any two disks with centers in $C$, their boundary portions in $C'$ cross each other at most once. More importantly, we can define a duality relationship between points in $C$ and disk arcs in $C'$ (and vice versa): a point $p$ in $C$ is dual to the arc of the boundary of $D_p$ in $C'$, where $D_p$ is the disk centered at $p$. This duality helps to obtain a {\em Test Set Lemma} that is crucial to the algorithms in~\cite{ref:ChanOp12,ref:MatousekEf92,ref:MatousekRa93}. With these properties and some additional observations, we show that the algorithm for computing cuttings for hyperplanes~\cite{ref:ChazelleCu93} can be adapted to the disk arcs in $C'$. With the cutting algorithms and the Test Set Lemma, we show that the techniques in~\cite{ref:ChanOp12,ref:MatousekEf92,ref:MatousekRa93} can be adapted to unit-disk range searching for the points of $P(C')$ with respect to the query disks centered in $C$.

The rest of this section is organized as follows. In Section~\ref{sec:reduction}, we reduce the problem to problems with respect to pairs of cells $(C,C')$. Section~\ref{sec:basic} introduces some basic concepts and observations that are fundamental to our approach. We adapt the cutting algorithm of Chazelle~\cite{ref:ChazelleCu93} to our problem in Section~\ref{sec:cutting}. Section~\ref{sec:testset} proves the Test Set Lemma.
In the subsequent subsections, we adapt the algorithms of~\cite{ref:ChanOp12,ref:MatousekEf92,ref:MatousekRa93}, whose query times are all $\Omega(\sqrt{n})$ with $O(n)$ space. Section~\ref{sec:tradeoff} presents the trade-offs between the preprocessing and the query time. Section~\ref{sec:summary} finally summarizes all results.

\subsection{Reducing the problem to pairs of grid cells}
\label{sec:reduction}


For each point $p$ in the plane, we use $x(p)$ and $y(p)$ to denote its $x$- and $y$-coordinates, respectively, and we use $D_p$ to denote the disk centered at $p$. For any region $A$ in the plane, we use $P(A)$ to denote the subset of points of $P$ in $A$, i.e., $P(A)=P\cap A$.

We will compute a set $\calC$ of $O(n)$ pairwise-disjoint square cells in the plane with the following properties. (1) Each cell has side length $1/\sqrt{2}$. (2) Every two cells are separated by an axis-parallel line. (3) For a disk $D_p$ with center $p$, if $p$ is not in any cell of $\calC$, then $D_p\cap P=\emptyset$. (4) Each cell $C$ of $\calC$ is associated with a subset $N(C)$ of $O(1)$ cells of $\calC$, such that for any disk $D$ with center in $C$, every point of $P\cap D$ is in one of the cells of $N(C)$. (5) Each cell $C'$ of $\calC$ is in $N(C)$ for a constant number of cells $C\in \calC$.

The following is a key lemma for reducing the problem to pairs of square cells.

\begin{lemma}\label{lem:10}
\begin{enumerate}
\item
The set $\calC$ with the above properties, along with the subsets $P(C)$ and $N(C)$ for all cells $C\in \calC$, can be computed in $O(n\log n)$ time and $O(n)$ space.
\item
With $O(n\log n)$ time and $O(n)$ space preprocessing, given any query disk $D_p$ with center $p$, we can determine whether $p$ is in a cell $C$ of $\calC$, and if yes, return the set $N(C)$ in $O(\log n)$ time.
\end{enumerate}
\end{lemma}
\begin{proof}
We first compute $O(n)$ disjoint vertical trips in the plane, each bounded by
two vertical lines, as follows.
We sort all points of $P$ from left to right as $p_1,p_2,\ldots,p_n$.
Starting from $p_1$, we sweep the plane by a
vertical line $\ell$. The algorithm maintains an invariant that $\ell$
is in the current vertical strip whose left bounding line is known (and to the left of
$\ell$) and whose
right bounding line is to be determine and to the right of $\ell$.
Initially, we put a vertical line at $x(p_1)-1$ as the left bounding
line of the first strip. Suppose $\ell$ is at a point $p_i$. If
$i<n$ and $x(p_{i+1})-x(p_i)\leq 3$, then we move $\ell$ to
$p_{i+1}$. Otherwise, we put a vertical line at $x(p_i)+1+\xi$
as the right bounding line of the current strip, where $\xi$ is the
smallest non-negative value such that $x(p_i)+1+\xi-x'$ is a multiple
of $1/\sqrt{2}$ with $x'$ as the $x$-coordinate of the left bounding
line of the current strip. Next, if $i=n$, then
we halt the algorithm; otherwise, we put a vertical line at
$x(p_{i+1})-1$ as the left bounding line of the next strip and move
$\ell$ to $p_{i+1}$.

After the algorithm, we have at most $n$ vertical strips that are
pairwise-disjoint. According to our algorithm, if the center of a disk $D$
is outside those strips, then $P(D)=\emptyset$.
Also, if a strip contains $m$ points of $P$, then the width of the
strip is $O(m)$. This also means that the sum of the widths of
all strips is $O(n)$. In addition, the width of each strip is
a multiple of $1/\sqrt{2}$.

Next, for each vertical strip $A$, by sweeping the points of $P(A)$ from top to bottom in a similar way as above, we compute $O(|P(A)|)$ disjoint horizontal strips, each of which becomes a rectangle with the two bounding lines of $A$. Similar to the above, if the center of a disk $D$ is in $A$ but outside those rectangles, then $P(D)=\emptyset$. The height of each rectangle is a multiple of $1/\sqrt{2}$. Also, the height of $R$ is $O(|P(R)|)$. This implies that the sum of the heights of all rectangles in $A$ is $O(|P(A)|)$. As such, the sum of the heights of all rectangles in all vertical strips is $O(n)$.

In this way, we compute a set of $O(n)$ pairwise-disjoint rectangles in $O(n)$ vertical strips with the following property. (1) If a disk $D$ whose center is outside those rectangles, then $D(P)=\emptyset$. (2) Each rectangle contains at least one point of $P$. (3) The sum of the widths of all vertical strips is $O(n)$. (4) The sum of heights of all rectangles in all vertical strips is $O(n)$. (5) The height (resp., width) of each rectangle is a multiple of $1/\sqrt{2}$.

In the following, due to the above property (5), we partition each rectangle into a grid of square cells of side length $1/\sqrt{2}$. Consider a vertical strip $A$.
We use a set $V_A$ of vertical lines to further partition $A$ into vertical sub-strips of width exactly $1/\sqrt{2}$ each.
Since the width of $A$ is $O(|P(A)|)$, $|V_A|=O(|P(A)|)$. Consider a rectangle $R$ of $A$.
We use a set $H_R$ of horizontal lines to partition $R$ into smaller rectangles of height exactly $1/\sqrt{2}$ each. 
Since the height of $R$ is $O(|P(R)|)$, $|H_R|=O(|P(R)|)$. The lines of $V_A\cup H_R$ together partition $R$ into square cells of side length $1/\sqrt{2}$, which form a grid $G_R$. We process the points of $P(R)$ using the grid $G_R$, as follows. Processing all rectangles in this way will prove the lemma.

\begin{figure}[t]
\begin{minipage}[t]{\textwidth}
\begin{center}
\includegraphics[height=2.0in]{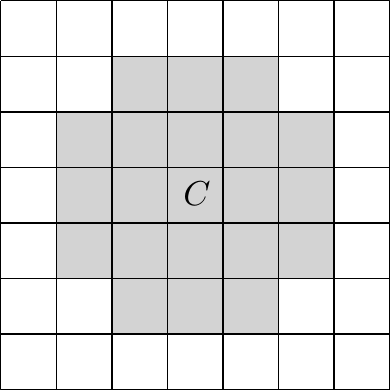}
\caption{\footnotesize The grey cells are all neighbor cells of $C$.}
\label{fig:grid}
\end{center}
\end{minipage}
\vspace{-0.15in}
\end{figure}

For each cell $C$ of $G_R$, a cell $C'$ is a {\em neighbor} of $C$ if the minimum distance between $C$ and $C'$ is at most $1$. Let $N'(C)$ denote the set of all neighbors of $C$ (e.g., see Fig.~\ref{fig:grid}). Clearly, $|N'(C)|=O(1)$.

Each cell $C$ of the grid $G_R$ has an index $(i,j)$ if $C$ is in the $i$-th low of $G_R$ from the top and in the $j$-th column from the left. For each point $p\in P(R)$, by doing binary search on the lines of $V_R$ and the lines of $H_R$, we can determine the cell $C_p$ (along with its index) that contains $p$. In this way, we can find all ``non-empty'' cells of $G_R$ that contain at least one point of $P$; further, for each non-empty cell $C$, the points of $P$ in $C$ are also computed and we associate them with $C$. Clearly, the number of non-empty cells is at most $|P(R)|$.

We next define a set $\calC_R$ of cells in the grid $G_R$ (the union of $\calC_R$ for all rectangles $R$ is $\calC$). For each non-empty cell $C$, we can find its neighbor set $N'(C)$ in $O(1)$ time using its index and we put all cells of $N'(C)$ in $\calC_R$. As the number of non-empty cells is at most $|P(R)|$ and $|N'(C)|=O(1)$ for each cell $C$, $|\calC_R|=O(|P(R)|)$. Note that it is possible that $\calC_R$ is a multi-set. To remove the repetitions, we can first sort all cells of $\calC_R$ using their indices (i.e., an index $(i,j)$ is ``smaller than'' $(i',j')$ if $i<i'$, or $i=i'$ and $j<j'$) and then scan the list to remove repetitions. We store this sorted list in our data structure.
Finally, for each cell $C$ of $\calC_R$, we define $N(C)$ as the set
of non-empty cells $C'$ with $C\in N'(C')$. We can compute $N(C)$ for
all $C\in \calC_R$ by scanning $N'(C')$ for all non-empty cells $C'$.
Note that $N(C)$ is a subset of $N'(C)$ and thus $|N(C)|=O(1)$ and $N(C)\subseteq \calC_R$. We
store $N(C)$ in our data structure.

Consider a disk $D_p$ whose center $p$ is in $R$. Let $C$ be a cell of $R$ containing $p$ (if $p$ lies on the common boundary of more than one cell, then let $C$ be an arbitrary one). If $C\not\in \calC$, then $D_p$ does not contain any point of $P(R)$ and thus does not contain any point of $P$ by the definition of $R$. Otherwise, by the definition of $N(C)$, the points of $P(R)\cap D_p$ are contained in the union of the cells of $N(C)$; further, by the definition of $R$, $P(D_p)=P(R)\cap D_p$.

In summary, given $V_A$, the above processing of $P(R)$ computes the line set $H_R$, the cells of $\calC_R$, and the set $N(C)$ and the points of $P(C)$ for each cell $C\in \calC_R$. Other than computing $V_A$, the total time is $O(|P(R)|\log n)$ and the space is $O(|P(R)|)$.

We process each rectangle $R$ of $A$ as above. Since the sum of $|P(R)|$ for all rectangles $R$ of $A$ is $|P(A)|$, the total processing time for all rectangles of $A$ (including computing the line set $V_A$) is $O(|P(A)|\log n)$ and the space is $O(|P(A)|)$.

We processing all vertical strips as above for $A$. Define $\calC$ stated in the lemma as the union of $\calC_R$ for all rectangles $R$ in all vertical strips. Because the total sum of $|P(A)|$ of all vertical strips is $O(n)$, the total processing time is $O(n\log n)$ and the total space is $O(n)$. This proves the first lemma statement.

For the second part of the lemma, consider a query disk $D_p$ with center $p$. We first do binary search on the bounding lines of all vertical strips and check whether $p$ is in any vertical strip. If not, then $p$ is not in any cell of $\calC$ and $D_p\cap P=\emptyset$. Otherwise, assume that $p$ is in a vertical strip $A$. Then, we do binary search on the horizontal bounding lines of the rectangles of $A$ and check whether $p$ is in any such rectangle. If not, then $p$ is not in any cell of $\calC$ and $D_p\cap P=\emptyset$. Otherwise, assume that $p$ is in a rectangle $R$. Next, by doing binary search on the vertical lines of $V_A$ and then on the horizontal lines of $H_R$, we determine the index of the cell $C$ of the grid $G_R$ that contains $p$. To determine whether $C$ is in $\calC_R$, we do binary search on the sorted list of $\calC_R$ using the index of $C$. If $C\not\in \calC_R$, then $C\not\in \calC$ and $D_p\cap P=\emptyset$. Otherwise, we return $N(C)$, which is stored in the data structure. Clearly, the running time for the algorithm is $O(\log n)$. This proves the lemma.
\end{proof}

With Lemma~\ref{lem:10} in hand, to solve the unit disk range searching problem,
for each cell $C\in \calC$ and each cell $C'\in N(C)$, we will preprocess the points of $P(C')$ with respect to the query disks whose
centers are in $C$. Suppose the preprocessing
time (resp. space) for each such pair $(C,C')$ is $f(m)=\Omega(m)$, where
$m=|P(C')|$. Then, by the property (5) of $\calC$, the total
preprocessing time (resp., space) for all such pairs $(C,C')$ is
$f(n)$ (more precisely, this holds for all functions $f(\cdot)$ used in our paper). In the following, we will describe our preprocessing
algorithm for $(C,C')$. Since $N(C)\subset \calC$ and the points of $P$ in each cell of $\calC$ are already known by Lemma~\ref{lem:10}, $P(C')$ is available to us.
To simplify the notation and also due to the above discussion, we assume that all points of $P$ are in
$C'$, i.e., $P(C')=P$. Note that if $C=C'$, then the problem is trivial because any disk centered in $C'$ covers the entire cell. We thus assume $C\neq C'$. Due to the property (2) of $\calC$, without loss of generality, in the following we assume that $C$ and $C'$ are separated by a horizontal line such that $C$ is below the line.

\subsection{Basic concepts and observations}
\label{sec:basic}

For any two points $a$ and $b$, we use $\overline{ab}$ to denote the line segment connecting them.
For any compact region $A$ in the plane, let $\partial A$ denote the boundary of $A$, e.g., if $A$ is a disk, then $\partial A$ is a circle.

\begin{figure}[t]
\begin{minipage}[t]{0.47\textwidth}
\begin{center}
\includegraphics[height=1.3in]{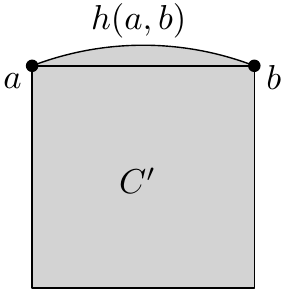}
\caption{\footnotesize Illustrating $\overline{C'}$, which is the grey region.}
\label{fig:enlarge}
\end{center}
\end{minipage}
\hspace{0.02in}
\begin{minipage}[t]{0.52\textwidth}
\begin{center}
\includegraphics[height=1.3in]{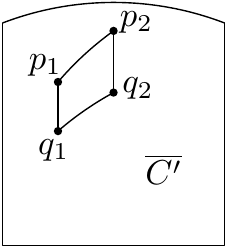}
\caption{\footnotesize Illustrating an upper arc pseudo-trapezoid in $\overline{C'}$.}
\label{fig:trapezoid}
\end{center}
\end{minipage}
\vspace{-0.15in}
\end{figure}

%

Consider a disk $D$ whose center is in $C$. As the side length of $C'$ is $1/\sqrt{2}$, $\partial D\cap C'$ may contain up to two arcs of the circle $\partial D$. For this reason, we enlarge $C'$ to a region $\overline{C'}$ so that $\partial D\cap \overline{C'}$ contains at most one arc. The region $\overline{C'}$ is defined as follows (e.g., see Fig.~\ref{fig:enlarge}).

Let $a$ and $b$ be the two vertices of $C'$ on its top edge. Let $D_{ab}$ be the disk whose center is below $\overline{ab}$ and whose boundary contains both $a$ and $b$.
Let $h(a,b)$ be the arc of $\partial D_{ab}$ above $\overline{ab}$ and connecting $a$ and $b$. Define $\overline{C'}$ to be the region bounded by $h(a,b)$, and the three edges of $C'$ other than $\overline{ab}$. As the side length of $C'$ is $1/\sqrt{2}$, for any disk $D$ whose center is in $C$, $\partial D\cap \overline{C'}$ is either $\emptyset$ or a single arc of $\partial D$ (which is on the upper half-circle of $\partial D$).
Let $e_b$ denote the bottom edge of $C'$.

Consider a disk $D$. An arc $h$ on the upper half-circle of
$\partial D$ (i.e., the half-circle above the horizontal line through
its center) is called an {\em upper disk arc} (or {\em upper arc} for
short); {\em lower arcs} are defined symmetrically. Note that an
upper arc is $x$-monotone, i.e., each vertical line intersects it at a
single point if not empty. If $h$ is an arc of a disk $D$, then we say that $D$ is the {\em underlying disk} of $h$ and the center of $D$ is also called the {\em center} of $h$.
An arc $h$ in $\overline{C'}$ is called a {\em spanning arc} if both endpoints of $h$ are on $\partial\overline{C'}$.
As we are mainly dealing with upper arcs of $\overline{C'}$ whose centers are in $C$, in the following, unless otherwise stated, an upper arc always refers to one whose center is in $C$.

The following is an easy but crucial observation that makes it possible to adapt many techniques for dealing with lines in the plane to spanning upper arcs of $\overline{C'}$. In the following discussion, we will use this observation without explicitly referring to it.

\begin{observation}\label{obser:10}
Suppose $h$ is an upper arc in $\overline{C'}$, and $e$ is a vertical line segment or an upper arc in $\overline{C'}$. Then, $h$ and $e$ can intersect each other at most once.
\end{observation}
\begin{proof}
If $e$ is a vertical segment, since $h$ is $x$-monotone, $h$ and $e$ can intersect each other at most once. If $e$ is an upper arc, since both $e$ and $h$ are upper arcs of disks whose centers are in $C$ and they are both in $\overline{C'}$, they can intersect each other at most once.
\end{proof}



\paragraph{Pseudo-trapezoids.}
Let $h(p_1,p_2)$ be an upper arc with $p_1$ and $p_2$ as its left and right endpoints, respectively. Define $h(q_1,q_2)$ similarly, such that $x(p_1)=x(q_1)$ and $x(p_2)=x(q_2)$. Assume that $h(p_1,p_2)$ and $h(q_1,q_2)$ do not cross each other and $h(p_1,p_2)$ is above $h(q_1,q_2)$. The region $\sigma$ bounded by the two arcs and the two vertical lines
$\overline{p_1q_1}$ and $\overline{p_2q_2}$ is called an {\em upper-arc pseudo-trapezoid} (e.g., see~Fig.~\ref{fig:trapezoid}). We call $\overline{p_1q_1}$ and $\overline{p_2q_2}$ the two {\em vertical sides} of $\sigma$, and call $h(q_1,q_1)$ and $h(p_1,p_2)$ the {\em top arc} and {\em bottom arc} of $\sigma$, respectively. The region $\sigma$ is also considered as an upper-arc pseudo-trapezoid if the bottom arc $h(q_1,q_2)$ is replaced by a line segment $\overline{q_1q_2}$ on $e_b$ (for simplicity, we still refer to $\overline{q_1q_2}$ as the bottom-arc of $\sigma$). In this way, $\overline{C'}$ itself is an upper-arc pseudo-trapezoid.
Note that for any pseudo-trapezoid $\sigma$ in $\overline{C'}$ and a disk $D$ centered in $C$, $\partial D\cap \sigma$ is either empty or an upper arc.

\paragraph{The counterparts of $C$ (with respect to $C'$).} The above definitions in $C'$ (with respect to $C$) have counterparts in $C$ (with respect to $C'$) with similar properties. First, we define $\overline{C}$ in a symmetric way as $\overline{C'}$, i.e., a lower arc connecting the two bottom vertices of $C'$ is used to bound $\partial \overline{C}$; e.g., see Fig.~\ref{fig:enlargeC}. Also, we define {\em lower-arc pseudo-trapezoids} and {\em spanning lower arcs} similarly, and unless otherwise stated, a lower arc in $\overline{C}$ refer to one whose center is in $C'$. In the following discussion, unless otherwise stated, properties, algorithms, and observations for the concepts of $C'$ with respect to $C$ also hold for their counterparts of $C$ with respect to $C'$.

\begin{figure}[t]
\begin{minipage}[t]{\textwidth}
\begin{center}
\includegraphics[height=1.3in]{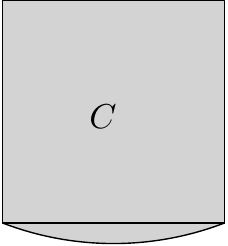}
\caption{\footnotesize Illustrating $\overline{C}$, which is the grey region.}
\label{fig:enlargeC}
\end{center}
\end{minipage}
\vspace{-0.15in}
\end{figure}

\paragraph{Duality.} We define a duality relationship between
upper arcs in $\overline{C'}$ and points in $C$. For an upper arc $h$
in $\overline{C'}$, we consider its center as its {\em dual point} in
$C$. For a point $q\in C$, we consider the upper arc $\partial D_q\cap
\overline{C'}$ as its {\em dual arc} in $\overline{C'}$ if it is not
empty. Similarly, we define duality relationship between lower arcs in
$\overline{C}$ and points in $C'$.
Note that if the boundary of a disk centered at a point $p\in P$ does not intersect $\overline{C}$,
then the point $p$ can be ignored from $P$ in our preprocessing because among all disks centered in $C$ one disk contains $p$ if and only if all other disks contain $p$.
Henceforth, without loss of generality, we assume that $\partial D_p$ intersects
$\overline{C}$ for all points $p\in P$, implying that every point of
$P$ is dual to a lower arc in $\overline{C}$.
Note that our duality is similar in spirit to the duality introduced by Agarwal and Sharir~\cite{ref:AgarwalPs05} between points and pseudo-lines.

\subsection{Computing hierarchical cuttings for disk arcs}
\label{sec:cutting}

Let $H$ be a set of $n$ spanning upper arcs in $\overline{C'}$. For a compact region $A$ of $\overline{C'}$, we use $H_A$ to denote the set of arcs of $H$ that intersect the relative interior of $A$.
By adapting its definition for hyperplanes, e.g.,~\cite{ref:ChazelleCu93,ref:MatousekRa93}, a {\em cutting} for $H$ is a collection $\Xi$ of closed cells (each of which is an upper-arc pseudo-trapezoid) with disjoint interiors, which together cover the entire $\overline{C'}$. The {\em size} of $\Xi$ is the number of cells in $\Xi$.
For a parameter $1\leq r\leq n$, a {\em $(1/r)$-cutting} for $H$ is a cutting $\Xi$ satisfying $|H_{\sigma}|\leq n/r$ for every cell $\sigma\in \Xi$.

We will adapt the algorithm of Chazelle~\cite{ref:ChazelleCu93} to computing a $(1/r)$-cutting of size $O(r^2)$ for $H$. It is actually a sequence of {\em hierarchical cuttings}. Specifically, we say that a cutting $\Xi'$ {\em $c$-refines} a cutting $\Xi$ if every cell of $\Xi'$ is contained in a single cell of $\Xi$ and every cell of $\Xi$ contains at most $c$ cells of $\Xi'$. Let $\Xi_0,\Xi_1,\ldots,\Xi_k$ be a sequence of cuttings such that $\Xi_0$ consists of the single cell $\overline{C'}$ (recall that $\overline{C'}$ itself is an upper arc pseudo-trapezoid), and every $\Xi_i$ is a $(1/\rho^i)$-cutting of size $O(\rho^{2i})$ which $c$-refines $\Xi_{i-1}$, for two constants $\rho$ and $c$. In order to make $\Xi_k$ a $(1/r)$-cutting, we set $k=\lceil\log_{\rho} r\rceil$.
The above sequence of cuttings is called a {\em hierarchical $(1/r)$-cutting} of $H$. If a cell $\sigma\in \Xi_{j-1}$ contains a cell $\sigma'\in \Xi_j$, we say that $\sigma$ is the {\em parent} of $\sigma'$ and $\sigma'$ is a {\em child} of $\sigma$. Hence, one could view $\Xi$ as a tree structure with $\Xi_0$ as the root.

Let $\chi$ denote the number of intersections of the arcs of $H$.
We have the following theorem.

\begin{theorem}\label{theo:cutting}{\em (The Cutting Theorem)}
Let $\chi$ denote the number of intersections of the arcs of $H$.
For any $r\leq n$, a hierarchical $(1/r)$-cutting of size $O(r^2)$ for $H$ (together with the sets $H_{\sigma}$ for every cell $\sigma$ of $\Xi_i$ for all $0\leq i\leq k$) can be computed in $O(nr)$ time; more specifically, the size of the cutting is bounded by $O(r^{1+\delta}+\chi\cdot r^2/n^2)$ and the running time of the algorithm is bounded by $O(nr^{\delta}+\chi\cdot r/n)$, for any small $\delta>0$.
\end{theorem}

\paragraph{Remark.}
Correspondingly, for a set of lower arcs in $\overline{C}$, we can define cuttings similarly with lower-arc pseudo-trapezoids as cells; the same result as Theorem~\ref{theo:cutting} also holds for computing lower-arc cuttings. Also note that the algorithm is optimal if the subsets $H_{\sigma}$'s need to be computed. Further, similar result on cuttings for other more general curves in the plane (e.g., circles or circular arcs of different radii, pseudo-lines, line segments, etc.) can also be obtained.
\medskip

To prove Theorem~\ref{theo:cutting}, we adapt Chazelle's algorithm for computing cuttings for hyperplanes~\cite{ref:ChazelleCu93}. It was stated in \cite{ref:AgarwalPs05} that Chazelle's algorithm can be extended to compute such a cutting of size $O(r^{1+\delta}+\chi\cdot r^2/n^2)$ in $O(n^{1+\delta}+\chi\cdot r/n)$ time. However, no details were provided in~\cite{ref:AgarwalPs05}.
For completeness and also for helping the reader to better understand our cutting, we present the algorithm details in the appendix, where we actually give a more general algorithm that also works for other curves in the plane (e.g., circles or circular arcs of different radii, pseudo-lines, line segments, etc.).
Further, our result also reduces the factor $n^{1+\delta}$ in the above time complexity of~\cite{ref:AgarwalPs05} to $nr^{\delta}$.


\paragraph{The weighted case.} To adapt the simplex range searching algorithms
in~\cite{ref:ChanOp12,ref:MatousekEf92,ref:MatousekRa93}, we will need to compute cuttings for a weighted set $H$ of spanning upper arcs in $\overline{C'}$, where each arc $h\in H$ has a nonnegative weight $w(h)$. The hierarchical $(1/r)$-cutting can be naturally generalized to the weighted case (i.e., the interior of each pseudo-trapezoid in a $(1/r)$-cutting can be intersected by upper arcs of $H$ of total weight at most $w(H)/r$, where $w(H)$ is the total weight of all arcs of $H$). By a method in~\cite{ref:MatousekCu91}, any algorithm computing a hierarchical $(1/r)$-cutting for a set of hyperplanes can be converted to the weighted case with only a constant factor overhead. We can use the same technique to extend any algorithm computing a hierarchical $(1/r)$-cutting for a set of upper arcs to the weighted case.

\subsection{Test Set Lemma}
\label{sec:testset}

A critical component in all simplex range searching algorithms
in~\cite{ref:ChanOp12,ref:MatousekEf92,ref:MatousekRa93} is a Test Set
Lemma. We prove a similar result for our problem, by using the duality.
For any pseudo-trapezoid $\sigma$ in $\overline{C'}$, we say that an upper arc $h$ {\em crosses} $\sigma$ if $h$ intersects the interior of $\sigma$.

\begin{lemma}\label{lem:testset}{\em (Test Set Lemma)}
For any parameter $r\leq n$,
there exists a set $Q$ of at most $r$ spanning upper arcs in $\overline{C'}$, such that for any
collection $\Pi$ of interior-disjoint upper-arc pseudo-trapezoids in $\overline{C'}$
satisfying that each pseudo-trapezoid contains at least $n/(c\cdot r)$ points of $P$ for some constant $c>0$, the following holds: if $\kappa$ is the
maximum number of pseudo-trapezoids of $\Pi$ crossed by any
upper arc of $Q$, then the maximum number of pseudo-trapezoids of
$\Pi$ crossed by any upper arc in $\overline{C'}$ is at most
$O(\kappa+\sqrt{r})$.
\end{lemma}
\begin{proof}
We adapt the proof of Lemma~3.3~\cite{ref:MatousekEf92} by using our duality.
Let $H$ be the set of lower arcs in $\overline{C}$ dual to the points of $P$. By Theorem~\ref{theo:cutting}, we can choose a $(1/t)$-cutting $\Xi$ for $H$ whose lower-arc pseudo-trapezoids have at most $r$ vertices in total, with $t=\Theta(\sqrt{r})$. Let $V$ be the set of all vertices of the pseudo-trapezoids of $\Xi$. Let $Q$ be the set of upper arcs in $\overline{C'}$ dual to the points of $V$. Below we argue that $Q$ has the property stated in the lemma.

Consider an upper arc $h$ in $\overline{C'}$. Let $\sigma$ be the pseudo-trapezoid of $\Xi$ that contains the center of $h$ (recall that by our convention any upper arc of $\overline{C'}$ has its center in $C$). Let $G$ be the set of at most four upper arcs in $\overline{C'}$ dual to the vertices of $\sigma$.
By the hypothesis of the lemma, each arc of $G$ crosses at most $\kappa$ cells of $\Pi$. It remains to bound the number of cells of $\Pi$ crossed by $h$ but by no arc of $G$. Such cells must be completely contained in the zone $Z(h)$ of $h$ in the arrangement $\calA(G)$ of the arcs of $G$ in $\overline{C'}$. It can be verified that any point of $P$ lying in the interior of $Z(h)$ must be dual to a lower arc in $H$ that crosses $\sigma$; there are at most $n/t=O(n/\sqrt{r})$ such lower arcs in $H$. Therefore, the zone $Z(h)$ contains $O(n/\sqrt{r})$ points of $P$. Since each cell of $\Pi$ has at least $n/(cr)$ points of $P$, the number of cells of $\Pi$ completely contained in $Z(h)$ is $O(\sqrt{r})$. This proves the lemma.
\end{proof}

With our Cutting Theorem (i.e., Theorem~\ref{theo:cutting}) and the Test Set Lemma, we proceed to adapt the simplex range searching algorithms in~\cite{ref:ChanOp12,ref:MatousekEf92,ref:MatousekRa93} to our problem in the following subsections.

\subsection{A data structure based on pseudo-trapezoidal partitions}
\label{sec:partition}

We first extend the simplicial partition for hyperplanes in~\cite{ref:MatousekEf92} to our problem, which we rename {\em pseudo-trapezoidal partition}.
A {\em pseudo-trapezoidal partition} for $P$ is a
collection $\Pi=\{(P_1,\sigma_1),\ldots,(P_m,\sigma_m)\}$, where the $P_i$'s
are pairwise disjoint subsets forming a partition of $P$, and each
$\sigma_i$ is a relatively open upper-arc pseudo-trapezoid in $\overline{C'}$ containing
all points of $P_i$.
The pseudo-trapezoidal partition we will compute has the following additional property:
$\max_{1\leq i\leq m}|P_i|<2\cdot \min_{1\leq i\leq m}|P_i|$,
i.e., all subsets have roughly the same size. Note that the trapezoids
$\sigma_i$'s may overlap. The subsets $P_i$'s are called {\em classes} of $\Pi$.

For any upper arc $h$ in $\overline{C'}$, we define its {\em
crossing number} with respect to $\Pi$ as the number of pseudo-trapezoids of
$\Pi$ crossed by $h$. The crossing number of $\Pi$ is defined as the maximum crossing numbers of all upper arcs $h$ in $\overline{C'}$.
The following partition theorem corresponds to Theorem~3.1~\cite{ref:MatousekEf92}. Its proof is similar to Theorem~3.1 in~\cite{ref:MatousekEf92}, with our Test Set Lemma and our Cutting Theorem.

\begin{theorem}{\em (Partition Theorem)}
Let $s$ be an integer with $2\leq s<n$ and $r=n/s$. There exists a pseudo-trapezoidal partition $\Pi$ for $P$, whose classes $P_i$ satisfy $s\leq |P_i|<2s$, and whose crossing number is $O(\sqrt{r})$.
\end{theorem}
\begin{proof}
Note that similar result as the theorem is already known for pseudo-lines with respect to points~\cite{ref:AgarwalPs05}. Here for completeness we sketch the proof and refer the reader to~\cite{ref:MatousekEf92} for detailed analysis.

We first apply the Test Set Lemma on $P$ to obtain a set $Q$ of at most $r$ upper arcs in $\overline{C'}$. 
The algorithm proceeds with $r$ iterations. In the $i$-th iteration, $1\leq i\leq r$, we will compute the set $P_i$ and the upper-arc pseudo-trapezoid $\sigma_i$. Suppose $P_1,\ldots,P_i$ and $\sigma_1,\ldots,\sigma_i$ have already been computed. Let $P_i'=P\setminus \bigcup_{k=1}^i P_k$ and $n_i=|P_i'|$. The algorithm maintains an invariant that $n_i\geq s$. We describe the $(i+1)$-the iteration of the algorithm below.

If $n_i<2s$, then we set $P_{i+1}=P_i'$, $\sigma_{i+1}=\overline{C'}$, $m=i+1$, and $\Pi=\{(P_1,\sigma_1),\ldots, (P_m,\sigma_m)\}$; this finishes the construction of $\Pi$. In what follows, we assume $n_i\geq 2s$.

For each arc $h\in Q$, let $k_i(h)$ denote the number of pseudo-trapezoids among $\sigma_1,\ldots,\sigma_i$ crossed by $h$. We define a weighted set $(Q,w_i)$ by setting $w_i(h)=2^{k_i(h)}$ for every $h\in Q$.
By the Cutting Theorem, we can choose a parameter $t_i\geq c\cdot \sqrt{n_i/s}$ for some constant $c>0$, such that there exists a $(1/t_i)$-cutting $\Xi_i$ for $(Q,w_i)$ that has at most $n_i/s$ cells.
Hence, some cell of $\Xi_i$ contains at least $s$ points of $P_i'$. Let $\sigma_{i+1}$ be such a cell.
Among the at least $s$ points of $P_i'$ contained in $\sigma_{i+1}$, we arbitrarily pick $s$ points to form $P_{i+1}$.

This finishes the description of the construction of $\Pi$. Using the Test Set Lemma and following the same analysis as that in Lemma~3.2~\cite{ref:MatousekEf92}, we can show that the crossing number of $\Pi$ is bounded by $O(\sqrt{r})$.
\end{proof}

\begin{lemma}\label{lem:60}
Given an integer $s$ with $2\leq s<n$ and $r=n/s$, a pseudo-trapezoidal partition for $P$
satisfying $2\leq |P_i|<2s$ for every class $P_i$ and with crossing
number $O(r^{1/2+\delta})$ can be constructed in $O(n\log r)$ time.
\end{lemma}
\begin{proof}
First of all, if $r=O(1)$, a pseudo-trapezoidal partition
as in the Partition Theorem can be constructed in $O(n)$ time. To see this, it
suffices to go through the steps of the proof of the Partition
Theorem and verify that the total time is bounded by $O(n)$. For the Test Set Lemma, we need to compute a $(1/t)$-cutting for a set of $n$ arcs with $t=O(1)$, which can be done in $O(n)$ time by our Cutting Theorem. The rest of the time analysis follows the same as
Lemma~3.4~\cite{ref:MatousekEf92}.

To prove the lemma, we build the partition recursively. Specifically,
we first build a partition $\Pi$ with class sizes between $s_1$ and
$2s_1$ and with crossing number at most $c\cdot \sqrt{r_1}$, where
$r_1=n/s_1$ and $c$ is a constant depending on the proof analysis of the
Partition Theorem. Then, for every class $P_i$ of this partition, we
construct a partition $\Pi_i$ with class sizes between $s_2$ and
$2s_2$ and with crossing number $c^2\cdot \sqrt{r_2}$, where
$r_2=s_1/s_2$. All these secondary partitions form a
pseudo-trapezoidal partition with class sizes between $s_2$ and $2s_2$
and with crossing number $c^2\sqrt{r}$, where $r=n/s_2$. By repeating
this process for a certain number of times and choosing a sufficient
large constant $r_0$ for parameter $s_1=n/r_0$, $s_2=s_1/r_0$, etc, we can
achieve the lemma. Note that we lose a constant factor in the crossing number at every iteration.
Refer to Corollary~3.5~\cite{ref:MatousekEf92} for
details.
\end{proof}

Let $H$ be a set of $n$ spanning upper arcs in $\overline{C'}$.
Next, using Lemma~\ref{lem:60} we give a faster algorithm (than the one in our Cutting Theorem) to compute cuttings for $H$ in the following lemma, which corresponds to Proposition~4.4~\cite{ref:MatousekEf92}.

\begin{lemma}\label{lem:70}
For any $r\leq n^{1/2-\delta}$, a $(1/r)$-cutting of size $O(r^2)$ for $H$ can be computed in time $O(n\log r+r^{3+\delta})$. In particular, the running time is $O(n\log r)$ when $r\leq n^{1/3+\delta}$.
\end{lemma}
\begin{proof}
We first define the {\em $\epsilon$-approximations}. Let $R$ be a subset of $H$ and is equipped with a weight function $w(\cdot)$, i.e., for each arc $h\in R$, $h$ has a weight $w(h)$. For any subset $R'$ of $R$, we use $w(R')$ to denote the total sum of the weights of all arcs of $R'$. The weighted set $(R,w)$ is an {\em $\epsilon$-approximation} if $|w(R_e)/w(R)-|H_e|/|H||<\epsilon$, for $e$ as any sub-segment of $e_b$, any vertical line segment, or any upper arc in $\overline{C'}$.
Next we prove the lemma.

Let $\scrD(H)$ be the set of points in $C$ dual to the arcs of $H$. Let $\Pi$ be a lower-arc trapezoidal partition of $\scrD(H)$ with crossing number $\kappa$ and each class has no more than $s$ points. For each class $P_i$ of $\Pi$, we pick an arbitrary point $p_i\in P_i$ and let $h_i$ be the arc of $H$ dual to $p_i$; we set $w(h_i)=|P_i|$. In this way, we obtain a subset $R$ of weighted arcs of $H$.

We claim that that $(R,w)$ is a $(2\kappa s/n)$-approximation of $H$. Indeed, let $e$ be sub-segment of $e_b$, a vertical segment, or an upper arc in $\overline{C'}$. Our goal is to show that $|w(R_e)/w(R)-|H_e|/|H||<2\kappa s/n$. By definition, $w(R)=|H|$. Hence, it suffices to prove $|w(R_e)-|H_e||<2\kappa s$. Let $D_1$ and $D_2$ be the disks centered at the two endpoints of $e$, respectively. Then, one can verify that if an arc of $H$ crosses the interior of $e$, then the center of the arc must be in one and only one disk of $D_1$ and $D_2$, i.e., in the symmetric difference of the two disks. This implies that $|w(R_e)-|H_e||$ is no more than the total number of points in the classes of $\Pi$ whose pseudo-trapezoids are crossed by the two lower arcs of $\overline{C}$ due to the two endpoints of $e$. As the crossing number of $\Pi$ is $\kappa$ and each class of $\Pi$ has at most $s$ points, we obtain that $|w(R_e)-|H_e||\leq 2\kappa s$. The claim is thus proved.

Due to the claim, we can compute a $(1/t)$-approximation $(R,w)$ of size $O(t^{2+\delta})$ for $H$ in $O(n\log t)$ time by Lemma~\ref{lem:60}.

By adapting an observation made by Matou\v{s}ek~\cite{ref:MatousekCu91} for hyperplanes (see also Lemma~4.3~\cite{ref:MatousekEf92}), we can show that if $(R,w)$ is an $\epsilon$-approximation for $H$ and $\Xi$ is an $\epsilon'$-cutting for $(R,w)$, then $\Xi$ is a $4(\epsilon+\epsilon')$-cutting for $H$. As such, we can compute a $(1/r)$-cutting for $H$, as follows.
First, we compute a $(1/8r)$-approximation $(R,w)$ of size $O(r^{2+\delta})$ for $H$ in $O(n\log r)$ time using by Lemma~\ref{lem:60}, as discussed above. Second, we compute a $(1/8r)$-cutting $\Xi$ for $(R,w)$ in $O(r^{3+\delta})$ time by our Cutting Theorem. According to the above observation, $\Xi$ is a $(1/r)$-cutting for $H$. The total time of the algorithm is $O(n\log r+r^{3+\delta})$.
\end{proof}

\paragraph{Remark.} The algorithm in Lemma~\ref{lem:70} does
not compute the subsets $H_{\sigma}$ for all cells $\sigma$
of the cutting, since otherwise the algorithm of the Cutting Theorem is optimal.
\medskip

The following lemma, which corresponds to Lemma~4.5~\cite{ref:MatousekEf92}, will be used in the algorithm for Lemma~\ref{lem:90}.
\begin{lemma}\label{lem:80}
Given any constant $c>0$, let $r\leq n^{\alpha}$ be a parameter, where $\alpha$ is a constant depending on $c$. We can build in $O(n\log r)$ time a data structure of $O(n)$ space for $P$ such that the number of points of $P$ in a query upper-arc pseudo-trapezoid in $\overline{C'}$ can be computed in $O(n/r^c)$ time and deleting a point can be handled in $O(\log r)$ time (the value of $n$ in the query time refers to the original size of $P$ before any deletion happens).
\end{lemma}
\begin{proof}
In the preprocessing, we compute in $O(n\log t)$ time an upper-arc pseudo-trapezoidal partition $\Pi$ for $P$ by Lemma~\ref{lem:60}, with at most $t$ classes of sizes between $n/t$ and $2n/t$, and with crossing number $O(t^{1/2+\delta})$, where $t$ is a sufficiently large (constant) power of $r$. For every $(P_i,\sigma_i)\in \Pi$, we store the pseudo-trapezoid $\sigma_i$, the size $|P_i|$, and the list of points of $P_i$. We further partially sort the points of each $P_i$ into at most $2t$ subsets of size at most $n/t^2$ each, such that all points in the $i$-th subset are to the left of all points of the $(i+1)$-th subset (but points in each subset are not sorted). As $|P_i|$ is between $n/t$ and $2n/t$, the above partial sorting can be done in $O(n/t\cdot \log t)$ time using the selection algorithm. Next, we build a tree $T_i$, whose leaves correspond to the above subsets of $P_i$ from left to right. For each node $v$ of $T_i$, we store the number of points in the subsets of the leaves in the subtree rooted at $v$. It takes $O(|P_i|)$ time to build the tree.
Hence, the total preprocessing time for all classes of $\Pi$ is $O(n\log t)$, which is $O(n\log r)$ as $t$ is a constant power of $r$. The space is $O(n)$.

Given a query upper-arc pseudo-trapezoid $\sigma$, we compute the number of the points of $P$ in $\sigma$ as follows. For each trapezoid $\sigma_i\in \Pi$, we check whether $\sigma_i$ is contained in $\sigma$. If yes, we add $|P_i|$ to the total count. The remaining points of $P$ in $\sigma$ that are not counted are those contained in classes $P_i$ whose pseudo-trapezoids $\sigma_i$ are crossed by the boundary of $\sigma$. We partition those pseudo-trapezoids into two subsets. Let $\Sigma_1$ denote the subset of those pseudo-trapezoids that are crossed by either the top arc or the bottom arc of $\sigma$. Let $\Sigma_2$ denote the subset of the rest pseudo-trapezoids of $\Pi$ crossed by the boundary of $\sigma$; hence, each pseudo-trapezoid of $\Sigma_2$ is crossed by either the left or the right side of $\sigma$ but not crossed by either the top or the bottom arc of $\sigma$. The two subsets $\Sigma_1$ and $\Sigma_2$ can be found by checking every pseudo-trapezoid of $\Pi$. As the crossing number of $\Pi$ is $O(t^{1/2+\delta})$, we have $|\Sigma_1|=O(t^{1/2+\delta})$.
However, $|\Sigma_2|$ may be as large as $t$.

For each pseudo-trapezoid $\sigma_i\in \Sigma_1$, we check each point of $P_i$ to see whether it lies in $\sigma$. For each pseudo-trapezoid $\sigma_i\in \Sigma_2$, suppose $\sigma_i$ intersects the left side of $\sigma$ but does not intersect the right side. Let $\ell$ be the vertical line containing the left side of $\sigma$. Because $\sigma_i$ does not intersect the top arc, the bottom arc, or the right side of $\sigma$, points of $P_i$ in $\sigma$ are exactly those to the right of $\ell$. Based on this observation, we find the number of such points using the tree $T_i$, as follows.
First, we search $T_i$ to find the leaf $v$ whose subset spans $\ell$ (i.e., $\ell$ is between the leftmost and the rightmost points of the subset). This search can also compute the total number of points of $P_i$ in the leaves to the right of $v$. Next, for the subset $P_v$ stored at the leaf $v$, we check every point of $P_v$ to determine whether it is in $\sigma$. As $|P_v|\leq n/t^2$, the time for searching $P_i$ is $O(\log t + n/t^2)$. If $\sigma_i$ intersects the right side of $\sigma$ but does not intersect the left side, then we can use a similar algorithm. If $\sigma_i$ intersects both the right side and the left side of $\sigma$, then the points of $P_i$ in $\sigma$ are exactly those points between the supporting lines of the left and right sides of $\sigma$. Hence, we can still find the number by searching  $T_i$ but following two search paths. The search time is still $O(\log t + n/t^2)$.

In this way, the query time is bounded by $O(t+(n/t)\cdot t^{1/2+\delta}+t\cdot (\log t + n/t^2))$, which is $O(t\log t+n/t^{1/2-\delta})$. If $t$ is a large enough power of $r$ and $\alpha$ is small enough, the time is bounded by $O(n/r^c)$.

Finally, to delete a point, we simply mark the point as deleted in the appropriate tree $T_i$, and update the point counts in the affected nodes of $T_i$. This can be done in $O(\log t)$ time, which is $O(\log r)$ time.
\end{proof}

The following lemma, which corresponds to Lemma~4.6~\cite{ref:MatousekEf92}, will be used in the algorithm for Lemma~\ref{lem:100}.

\begin{lemma}\label{lem:90}
There exists a small constant $\alpha>0$ such that
a pseudo-trapezoidal partition as in the Partition Theorem can be
constructed in $O(n\log r)$ time for any $r\leq n^{\alpha}$.
\end{lemma}
\begin{proof}
We go through the proof of the Partition Theorem. The first step is to compute a $(1/t)$-cutting using the Test Set Lemma with $t=\Theta(\sqrt{r})$, which can be done in $O(n\log r)$ time by Lemma~\ref{lem:70}.

Most of the remaining steps can be performed in time polynomial in $r$, not depending on $n$. The only exception is when we select a pseudo-trapezoid of the cutting $\Xi_i$ containing at least $s$ points of $P_i'$. To do so, we need to find the number of points in these faces as well as report the points inside the selected face. This requires $O(r^2)$ pseudo-trapezoidal range counting (and $O(r)$ reporting) queries on the original point set $P$, whose points may be deleted (after reporting). By Lemma~\ref{lem:80}, the queries together take $O(n\log r)$ time including the preprocessing if we set $c=2$, which can be achieved when $\alpha$ is small enough. Note that Lemma~\ref{lem:80} does not mention the range reporting query but it can be done by modifying the range counting algorithm with query time bounded by $O(n/r^c+k)$, where $k$ is the number of reported points.
\end{proof}

The following lemma, which corresponds to Theorem~4.7(i)~\cite{ref:MatousekEf92}, will be used in the algorithm for Theorem~\ref{theo:30}.

\begin{lemma}\label{lem:100}
For any fixed $\delta>0$, if $s\geq n^{\delta}$, then a pseudo-trapezoidal partition as in the Partition Theorem (whose classes $|P_i|$ satisfy $s\leq |P_i|<2s$ and whose crossing number is $O(\sqrt{r})$) can be constructed in $O(n\log r)$ time, where $r=n/s$.
\end{lemma}
\begin{proof}
We apply the recursive algorithm in Lemma~\ref{lem:60}, but use Lemma~\ref{lem:90} to construct the partition at each iteration. For a current point set of size $m$, we set the parameter $s$ to $m^{1-\alpha}$ for the next iteration, where $\alpha$ refers to the parameter in Lemma~\ref{lem:90}. After the $i$-th iteration, the size of the classes of the current partition are roughly $n^{(1-\alpha)^i}$. Hence, it suffices to iterate $O(1)$ times before $(1-\alpha)^i$ drops below $\delta$. Therefore, we only lost a constant factor in the crossing number. The lemma thus follows.
\end{proof}

Using the above lemma, we can obtain the following result for the disk range searching problem.

\begin{theorem}\label{theo:30}
We can build an $O(n)$ space data structure for $P$ in $O(n\log n)$ time, such that given any disk $D$ centered in $C$, the number of points of $P$ in $D$ can be computed in $O(\sqrt{n}(\log n)^{O(1)})$ time.
\end{theorem}
\begin{proof}
We build a partition tree $T$ using the algorithm of Lemma~\ref{lem:100} recursively, until we obtain a  partition of $P$ into subsets of constant sizes, which form the leaves of $T$. Each inner node $v$ of $T$ corresponds to a subset $P_v$ of $P$ as well as a pseudo-trapezoidal partition $\Pi_v$ of $P_v$, which form the children of $v$. At each child $u$ of $v$, we store the pseudo-trapezoid $\sigma_u$ of $\Pi_v$ containing $P_u$ and also store the size $|P_u|$. We construct the partition $\Pi_v$ using Lemma~\ref{lem:100} with parameter $s=\sqrt{|P_v|}$. Note that if $u$ is the root, then $\sigma_u=\overline{C'}$ and $P_u=P$.  Hence, the height of $T$ is $O(\log\log n)$.

Given a query disk $D$ whose center is in $C$, starting from the root of $T$, for each node $v$, we check whether $D$ contains the pseudo-trapezoid $\sigma_v$ stored at $v$. If yes, then we add $|P_v|$ to the total count. Otherwise, if $D$ crosses $\sigma_v$, then we proceed to the children of $v$.

The complexities are as stated in the theorem. The analysis is the same as Theorem~5.1~\cite{ref:MatousekEf92}.
\end{proof}

\paragraph{Remark.} It is straightforward to modify the algorithm to answer the {\em outside-disk queries}: compute the number of points of $P$ {\em outside} any query disk, with asymptotically the same complexities. This is also the case for other data structures given later, e.g., Theorems~\ref{theo:40}, \ref{theo:randomize}, \ref{theo:tradeoff}.

\subsection{A data structure based on hierarchical cuttings}
\label{sec:hierarchical}

In this section, by using our Cutting Theorem and Test Set Lemma, we adapt the techniques of Matou\v{s}ek~\cite{ref:MatousekRa93} to our problem. Our goal is to prove the following theorem.

\begin{theorem}\label{theo:40}
We can build an $O(n)$ space data structure for $P$ in $O(n^{1+\delta})$ time for any small constant $\delta>0$, such that given any disk $D$ whose center is in $C$, the number of points of $P$ in $D$ can be computed in $O(\sqrt{n})$ time.
\end{theorem}

We first construct a data structure for a subset $P'$ of at least half
points of $P$. To build a data structure for the whole $P$, the same
construction is performed for $P$, then for $P\setminus P'$, etc., and
thus a logarithmic number of data structures with geometrically
decreasing sizes will be obtained. Because the preprocessing time and space of
the data structure for $P'$ is $\Omega(n)$,
constructing all data structures for $P$ takes asymptotically the same time and space
as those for $P'$ only. To answer a disk query on $P$, each of these data
structures will be called. Since the query time for $P'$ is
$\Omega(\sqrt{n})$, the total query time for $P$ is asymptotically the
same as that for $P'$. Below we describe the data
structure for $P'$.

The data structure consists of a set of (not necessarily disjoint) upper-arc pseudo-trapezoids in $\overline{C'}$,
$\Psi_0=\{\sigma_1,\ldots,\sigma_t\}$ with $t=\sqrt{n}\log n$.
For each $1\leq i\leq t$, we have a subset $P_i\subseteq P$ of $n/(2t)$ points that are contained in $\sigma_i$. The subsets $P_i$'s form a disjoint partition of
$P'$.
For each $i$, there is a rooted tree $T_i$ whose nodes correspond to
pseudo-trapezoids, with $\sigma_i$ as the root. Each internal node of $T_i$
has $O(1)$ children whose pseudo-trapezoids are
interior-disjoint and together cover their parent pseudo-trapezoid. For each
pseudo-trapezoid $\sigma$ of $T_i$, let $P_{\sigma}=P_i\cap \sigma$.
If $\sigma$ is a leaf, then the points of $P_{\sigma}$ are
explicitly stored at $\sigma$; otherwise only the size $|P_{\sigma}|$ is stored there. Each point
of $P_i$ is stored in exactly one leaf pseudo-trapezoid of $T_i$.
The depth of $T_i$ is $q=O(\log n)$. Hence, the data structure is a forest of $t$
trees. Let $\Psi_j$ denote the set of all pseudo-trapezoids of all trees
$T_i$'s that lie at distance $j$ from the root. For any upper arc $h$ in $\overline{C'}$, let $K_j(h)$ be
the set of pseudo-trapezoids of $\Psi_j$ crossed by $h$; let $L_j(h)$ be set of the
leaf pseudo-trapezoids of $K_j(h)$. Define $K(h)=\bigcup_{j=0}^q K_j(h)$ and
$L(h)=\bigcup_{j=0}^q L_j(h)$.
The data structure guarantees the following for any upper arc $h$ in $\overline{C'}$:
\begin{equation}\label{equ:100}
\sum_{j=0}^q|\Psi_j|=O(n),
\end{equation}
\begin{equation}\label{equ:200}
|K(h)|=O(\sqrt{n}), \sum_{\sigma\in L(h)}|P_{\sigma}|=O(\sqrt{n}).
\end{equation}



We next discuss the algorithm for constructing the data structure.
The first step is to compute a test set $H$ (called a {\em guarding set} in~\cite{ref:MatousekRa93}) of $n$ spanning upper arcs in $\overline{C'}$. This can be done in time polynomial in $n$ by our Test Set Lemma. After that, the algorithm proceeds in $t$ iterations; in the $i$-th iteration, $T_i$, $\sigma_i$, and $P_i$ will be produced.

Suppose $T_j$, $\sigma_j$, and $P_j$ for all $j=1,2\ldots,i$ have been constructed.
Define $P_i'=P\setminus(P_1\cup\cdots \cup P_i)$. If $|P_i'|<n/2$, then we stop the construction. Otherwise, we proceed with the $(i+1)$-th iteration as follows. Let $\Psi_0^{(i)},\ldots,\Psi_q^{(i)}$ denote the already constructed parts of $\Psi_0,\ldots,\Psi_q$. Define $K_j^{(i)}(h)$ and $L_j^{(i)}(h)$ similarly as $K_j(h)$ and $L_j(h)$. We define a weighted arc set $(H,w_i)$.
For each arc $h\in H$, define the weight
\begin{equation*}\label{equ:weight}
w_i(h)= \exp\bigg(\frac{\log n}{\sqrt{n}}\cdot\bigg[\sum_{j=0}^{q}4^{q-j}\cdot |K_j^{(i)}(l)|+\sum_{\sigma\in K_q^{(i)}(l)}{|P_{\sigma}|}\bigg]\bigg).
\end{equation*}

Next, by our Cutting Theorem, we compute a hierarchical
$(1/r)$-cutting for $(H,w_i)$ with $r=\sqrt{n}$, which consists of a
sequence of cuttings $\Xi_0,\Xi_1,\ldots,\Xi_k$ with $\rho>4$ (note that by our algorithm for the Cutting Theorem, we can make $\rho$ larger than any given constant).


Suppose $p$ is the largest index such that the size of $\Xi_p$ is at most $t$. As the size of $\Xi_j$ is
$O(\rho^{2j})$, $\rho^{2p}=\Theta(t)$ and $\Xi_p$ is a
$(1/r_p)$-cutting of $(H,w_i)$ with $r_p=\rho^p=\Theta(\sqrt{t})$.
Let $q=k-p$. Note that $\rho^{q}=O(r/\sqrt{t})=O(\sqrt{n/t})$. Since $|P_i'|\geq n/2$ and $\Xi_p$ has at most $t$ pseudo-trapezoids, $\Xi_p$ has a pseudo-trapezoid, denoted by $\sigma_{i+1}$, containing at least $n/(2t)$ points of $P_i'$. We arbitrarily select $n/(2t)$ points of $P_i'\cap \sigma_{i+1}$ to form the set $P_{i+1}$. Further, all pseudo-trapezoids in $\Xi_p,\Xi_{p+1},\ldots,\Xi_k$ contained in $\sigma_{i+1}$ form the tree $T_{i+1}$, whose root is $\sigma_{i+1}$. Next, we eliminate some nodes from $T_{i+1}$ as follows.
Starting from the root, we perform a depth-first-search (DFS). Let $\sigma$ be the pseudo-trapezoid of the current node the DFS is visiting. Suppose $\sigma$ belongs to $\Xi_{p+j}$ for some $0\leq j\leq q$. If $\sigma$ contains at least $2^{q-j}$ points of $P_{i+1}$ ($\sigma$ is said to be {\em fat} in~\cite{ref:MatousekRa93}), then we proceed on the children of $\sigma$; otherwise, we make $\sigma$ a leaf  and return to its parent (and continue DFS). In other words, a pseudo-trapezoid of $T_{i+1}$ is kept if and only all its ancestors are fat. This finishes the construction of the $(i+1)$-th iteration.

The running time of the construction algorithm is polynomial in $n$. Using our Test Set Lemma and following the same analysis as in~\cite{ref:MatousekRa93}, we can show that Equations \eqref{equ:100} and \eqref{equ:200} hold. Note that \eqref{equ:100} implies that the total space of the data structure is $O(n)$. With \eqref{equ:200},  we show below that each disk range query can be answered in $O(\sqrt{n})$ time.

Given a query disk $D$ whose center is in $C$, the points of $P'$ in $D$ can be computed as follows.
First, compute the total number of points in the pseudo-trapezoids $\sigma_i$ of $\Psi_0$ that are contained in $\sigma$. Second, find the set $\Sigma$ of pseudo-trapezoids $\sigma_i$ of $\Psi_0$ that are crossed by the boundary of $\sigma$. Third, repeat the following steps until $\Sigma$ becomes empty. We remove one pseudo-trapezoid $\sigma$ from $\Sigma$. If it is a leaf, we check whether each point of $P_{\sigma}$ is in $\sigma$. Otherwise we check each of its children. We handle those completely contained in $\sigma$ directly and add those crossed by the boundary of $\sigma$ to $\Sigma$.

Equation~\eqref{equ:200} guarantees that the time spent in the third step is $O(\sqrt{n})$. The first two steps, however, take $O(t)=O(\sqrt{n}\log n)$ time if the pseudo-trapezoids of $\Psi_0$ are checked one by one. The following two lemmas respectively reduce the time of these two steps to $O(\sqrt{n})$ with additional preprocessing on the pseudo-trapezoids of $\Psi_0$. Note that our intention is to implement the two steps in $O(\sqrt{n})$ time with $O(n)$ time and space preprocessing. Hence, the results of the two lemmas may not be the best possible, but are sufficient for our purpose.

\begin{lemma}\label{lem:110}
With $O(t(\log t)^{O(1)})$ time and space preprocessing, the first step can be executed in $O(\sqrt{t}\cdot \exp(c\cdot \sqrt{\log t}))$ time, for a constant $c$.
\end{lemma}
\begin{proof}
Consider a pseudo-trapezoid $\sigma\in \Psi_0$. By the definition of upper-arc pseudo-trapezoids, $\sigma$ is completely contained in the query disk $D$ if and only if all four vertices of $\sigma$ are in $D$. We consider the four vertices of $\sigma$ as a $4$-tuple with a weight equal to $|P_{\sigma}|$. Hence, the problem becomes the following: preprocessing the set $A$ of $t$ $4$-tuples of the pseudo-trapezoids of $\Psi_0$ such that the total weight of all $4$-tuples contained in a query disk $D$ can be computed efficiently.

We use the algorithm for Lemma~6.2~\cite{ref:MatousekEf92} to build a multi-level data structure. We proceed by induction on $k$ with $1\leq k\leq 4$, i.e., solving the $k$-tuple problem by constructing a data structure $S_k(A)$. For $k=1$, we apply Theorem~\ref{theo:30} to obtain $S_1(A)$. For $k>1$, let $F$ be the set of first elements of all $k$-tuples of $A$.

To construct $S_k(A)$, we build a partition tree as Theorem~\ref{theo:30} on $F$, by setting $r=m/s=\exp(\sqrt{\log m})$ in a node $v$ whose subset $P_v$ has $m$ points. For every class $Q_i$ of the pseudo-trapezoidal partition $\Pi_v=\{(Q_1,\sigma_1),(Q_2,\sigma_2),\ldots\}$ for $P_v$, we let $A_i\subseteq A$  be the set of $k$-tuples whose first elements are in $Q_i$, and let $A_i'$ be the set of $(k-1)$-tuples arising by removing the first element from the $k$-tuples of $A_i$. We compute the data structure $S_{k-1}(A_i')$ and store it in the node $v$.

To answer a query for a disk $D$, we start from the root of the partition tree. For each current node $v$, we find the pseudo-trapezoids of the partition $\Pi_v$ contained in $D$, and for each such trapezoid $\sigma_i$, we use the data structure $S_{k-1}(A_i')$ to find the $k$-tuples of $A_i$ contained in $D$. We also find the pseudo-trapezoids of $\Pi_v$ crossed by the boundary of $\sigma$, and visit the corresponding subtrees of $v$ recursively.

The complexities are as stated in the lemma, which can be proved by the same analysis as in the proof of Lemma~6.2~\cite{ref:MatousekEf92}.
\end{proof}

Note that $\exp(c\sqrt{\log t})=O(t^{\delta})$ for any small $\delta>0$. Since $t=\sqrt{n}\log n$, the preprocessing time and space of the above lemma is bounded by $O(n)$ and the query time is bounded by $O(\sqrt{n})$.

\begin{lemma}
With $O(t(\log t)^{O(1)})$ time and space preprocessing, the second
step can be executed in $O(\sqrt{t}\cdot(\log n)^{O(1)}+k)$ time, where $k$ is the output size.
\end{lemma}
\begin{proof}
The second step of the query algorithm is to find all pseudo-trapezoids
of $\Psi_0$ that are crossed by the boundary of a query disk $D$. Recall that
$\partial D\cap \overline{C'}$ is a spanning upper arc $h$. For each edge $e$ of a
pseudo-trapezoid $\sigma_i\in \Psi_0$,
$h$ crosses $e$ if and only if one of the following two conditions holds: (1) the two endpoints of $e$ are in the
two regions of $\overline{C'}$ separated by $h$; (2) $e$ is a sub-segment of $e_b$, both endpoints $a$ and $b$ of $h$ are on $e_b$, and $\overline{ab}\subseteq e$. We say that $e$ is a {\em type-1 target edge} (resp., {\em type-2 target edge}) if $e$ satisfies the first (resp., second) condition. In the following, we discuss how to compute each type of target edges with complexities as stated in the lemma.

\paragraph{Computing type-1 target edges.}
Let $E$ be the set of the edges of all pseudo-trapezoids of $\Psi_0$. Note
that $|E|\leq 4t$. Given a query disk $D$, the problem is to find all type-1 target edges of $E$. We adapt the algorithm
for Lemma~6.3~\cite{ref:MatousekRa93} for reporting the segments crossed by a query hyperplane. Let $V_1$ denote the set of all left vertices of the edges of $E$ (if an edge is a vertical segment, then we take the bottom vertex); let $V_2$ be the set of the right vertices.

Similarly to Lemma~\ref{lem:110}, we build a 2-level partition tree $T$ on $V_1$. For each node $v$, we build a disk range reporting data structure on $P_v$, i.e., given a query disk $D$, report all points $P_v\cap D$. By the lifting method, the problem can be reduced to half-space range reporting in 3D~\cite{ref:AggarwalSo90,ref:AfshaniOp09,ref:ChanOp16,ref:ChazelleHa86,ref:MatousekRe92}. For example, using the result of~\cite{ref:ChanOp16}, for $m$ points in the plane, a data structure of $O(m)$ space can be built in $O(m\log m)$ time such that each disk range reporting can be answered in $O(\log m + k)$ time.
At each node $v$ of $T$, we set the parameter $s$ to $|P_v|^{2/3}$ when building the partition $\Pi_v$. Let $P_v'$ be the set of right vertices whose corresponding left vertices are in $P_v$. We build a disk range reporting data structure for $P_v'$ at $v$.

For each query disk $D$, we find all nodes $v$ of $T$ whose pseudo-trapezoids are {\em outside} $D$. For each such node $v$, using the disk range reporting data structure at $v$, we report all points of $P_v'$ {\em inside} $D$; all reported points correspond to the type-1 target edges of $E$.

The complexities are as stated in the lemma, which can be proved by the same analysis as in the proof of Lemma~6.3~\cite{ref:MatousekEf92}.

\paragraph{Computing type-2 target edges.} This case is fairly easy to handle. Let $E'$ be the set of the bottom edges of pseudo-trapezoids of $\Psi_0$ that are on $e_b$. Note that $|E'|\leq t$. For a query disk $D$, if the two endpoints $a$ and $b$ of $h$ are not both on $e_b$, then no type-2 target edges exist. Assume that both $a$ and $b$ are on $e_b$. Then, the problem is to report the segments of $E'$ that contain $\overline{ab}$. This problem can be solved in $O(\log |E'| + k)$ time after $O(|E'|)$ space and $O(|E'|\log |E'|)$ time preprocessing (e.g., by reducing the problem to 2D range reporting queries and then using priority search trees; see Exercise 10.10 in~\cite{ref:deBergCo08}).
\end{proof}

Since $t=\sqrt{n}\log n$, the preprocessing time and space of the above lemma is bounded by $O(n)$ and the query time is bounded by $O(\sqrt{n})$, for $k=O(\sqrt{n})$ by Equation~\eqref{equ:200}.

In summary, the above constructs our disk range searching data structure for Theorem~\ref{theo:40} in $O(n)$ space and the query time is $O(\sqrt{n})$. The preprocessing time is polynomial in $n$. To reduce it to $O(n^{1+\delta})$, we can can use the following approach.

\paragraph{Proof of Theorem~\ref{theo:40}.} We apply the same algorithm as in the proof of Theorem~\ref{theo:30}, but stop the algorithm when the size of $P_v$ is roughly equal to $n^{\delta'}$ for a suitable small value $\delta'>0$. Then, the height of $T$ is $O(1)$.
For each leaf node $v$ of $T$, we build the data structure $\scrD_v$ discussed above on $P_v$, which takes time polynomial in $|P_v|$. We make $\delta'$ small enough so that the total time we spend on processing the leaves of $T$ is $O(n^{1+\delta})$. This finishes the preprocessing, which takes $O(n^{1+\delta})$ time and $O(n)$ space.

Given a query disk $D$, we first follow the tree $T$ in the same way as before. Eventually we will reach a set $V$ of leaves $v$ whose pseudo-trapezoids $\sigma_v$ are crossed by $\partial D$. Since the height of $T$ of $O(1)$, the size of $V$ is $O(\sqrt{r})$, where $r=n^{1-\delta'}$ is the number of leaves of $T$. Again since the height of $T$ is $O(1)$, the time we spend on searching $T$ is $O(\sqrt{n})$ (see the detailed analysis of Theorem~5.1~\cite{ref:MatousekEf92}). Finally, for each leaf node $v\in V$, we use the data structure $\scrD_v$ to find the number of points of $P_v\cap D$, in $O(\sqrt{|P_v|})$ time. As $|P_v|=n^{\delta'}=n/r$, the total time for searching all leaf nodes of $V$ is $O(\sqrt{n})$. This proves Theorem~\ref{theo:40}.


\subsection{A randomized result}
\label{sec:randomize}

In this section, we show that the randomized result of Chan~\cite{ref:ChanOp12} can also be adapted for our problem, with the following result.

\begin{theorem}\label{theo:randomize}
We can build an $O(n)$ space data structure for $P$ in $O(n\log n)$ expected time by a randomized algorithm, such that given any disk $D$ whose center is in $C$, the number of points of $P$ in $D$ can be computed in $O(\sqrt{n})$ time with high probability.
\end{theorem}

The data structure is a partition tree, denoted by $T$, obtained by recursively subdividing $\overline{C'}$ into cells each of which is an upper-arc pseudo-trapezoid. Each node $v$ of $T$ corresponds to a cell, denoted by $\sigma_v$. If $v$ is the root, then $\sigma_v$ is $\overline{C'}$. If $v$ is not a leaf, then $v$ has $O(1)$ children whose cells form a disjoint partition of $\sigma_v$. Define $P_{v}=P\cap \sigma_v$. The set $P_{v}$ is not explicitly stored at $v$ unless $v$ is a leaf, in which case $|P_{v}|=O(1)$. The cardinality $|P_v|$ is stored at $v$. The height of $T$ is $O(\log n)$. If $\kappa$ is the maximum number of pseudo-trapezoids of $T$ that are crossed by any upper arc in $\overline{C'}$, then $\kappa=O(\sqrt{n})$ holds with high probability. The partition tree $T$ can be built by a randomized algorithm of $O(n\log n)$ expected time. The space of $T$ is $O(n)$.

We follow the algorithm scheme of Chan~\cite{ref:ChanOp12} but instead use our Cutting Algorithm, Test Set Lemma, and the duality relationship, except that two data structures in the algorithm need to be provided. Both data structures are for the same subproblem but with different performances, as follows. Let $H$ be a set of $m$ spanning upper arcs in $\overline{C'}$. Given a query upper-arc pseudo-trapezoid $\sigma$ in $\overline{C'}$, the problem is to report $H_{\sigma}$, where $H_{\sigma}$ is the set of all arcs of $H$ crossing $\sigma$. We show below that we can achieve the same performances as needed in Chan's algorithm.

The first data structure\footnote{It is used for computing $H_{\Delta_i}$ in Step~3(a) of the algorithm of Theorem~5.2~\cite{ref:ChanOp12}} requires the performance in the following lemma.

\begin{lemma}\label{lem:130}
With $O(m\log m)$ time preprocessing, each query can be answered in $O(\sqrt{m}(\log m)^{O(1)}+|H_{\sigma}|)$ time.
\end{lemma}
\begin{proof}
Consider a query pseudo-trapezoid $\sigma$. An arc $h\in H$ crosses $\sigma$ if and only if it crosses an edge $e$ of $\sigma$. Note that $e$ can be a sub-segment of the bottom side $e_b$ of $C'$, a vertical segment, or an upper arc in $\overline{C'}$. To answer the query, it suffices to find $H_e$ for all edges $e$ of $\sigma$, where $H_e$ is the set of arcs of $H$ crossing $e$. Observe that an arc $h\in H$ crosses $e$ if and only if one of the following two conditions hold: (1) the two endpoints of $e$ are in the
two regions of $\overline{C'}$ separated by $h$; (2) $e$ is a sub-segment of $e_b$, both endpoints $a$ and $b$ of $h$ are on $e_b$, and $\overline{ab}\subseteq e$. We say that $h$ is a {\em type-1 target arc} (resp., {\em type-2 target arc}) if $h$ satisfies the first (resp., second) condition. In the following, we discuss how to compute each type of target arcs with complexities as stated in the lemma.

\paragraph{Computing type-1 target arcs.}
For computing type-1 target arcs, we consider the problem in the dual setting as follows.
Let $H^*$ be the set of points in $C$ dual to the arcs of $H$. Let $D_1$ and $D_2$ be the disks centered at the two endpoints of $e$, respectively. Observe that an arc $h\in H$ is a type-1 target arc if and only if its dual point $h^*$ is in the intersection of $\overline{C}$ and $D'$, where $D'$ is the symmetric distance of $D_1$ and $D_2$. Hence, $|H_e|$ is equal to $|H^*\cap D'|$. Note that $D'$ is bounded by the two lower arcs of $D_1$ and $D_2$ in $\overline{C}$ as well as the boundary of $\overline{C}$. Hence, if we build the data structure of Theorem~\ref{theo:30} on $H^*$, by following a similar query algorithm, we can compute $|H^*\cap D'|$. The difference is that now for each node $v$ of $T$, we check whether $D'$ contains the pseudo-trapezoid $\sigma_v$ at $P_v$. As each of the two lower arcs of $D_1$ and $D_2$ in $\overline{C}$ crosses $O(\sqrt{r})$ trapezoids of $\Pi_v$, where $r$ is the number of pseudo-trapezoids of $\Pi_v$, the query time is still bounded by $O(\sqrt{m}\cdot (\log m)^{O(1)})$. It is straightforward to verify that if we need to report all points in $|H^*\cap D'|$, then the query time is $O(\sqrt{m}\cdot (\log m)^{O(1)}+|H^*\cap D'|)$; indeed, if $D'$ contains $\sigma_v$, then we simply follow the subtree at $v$ and report all points in all leaves of the subtree. By Theorem~\ref{theo:30}, the preprocessing time is $O(m\log m)$.


\paragraph{Computing type-2 target arcs.} This case is fairly easy to handle. Let $H'$ be the set of arcs of $H$ whose endpoints both are on $e_b$. Define $E'=\{\overline{ab}\ |\ \text{$a$ and $b$ are the two endpoints of $h$}, h\in H'\}$. Note that $|E'|\leq m$. For a query edge $e$, if $e\not\subseteq e_b$, then no type-2 target arcs exist. Assume that $e\subseteq e_b$. Then, an arc $h\in H$ is a type-2 target arc if and only if $\overline{ab}\subseteq e$, where $a$ and $b$ are the two endpoints of $h$. Hence, the problem is to report the segments of $E'$ that are contained in $e$. This problem can be solved in $O(\log |E'| + k)$ time after $O(|E'|\log |E'|)$ time preprocessing (e.g., by reducing the problem to 2D range reporting queries and then using range trees; see Exercise 10.9 in~\cite{ref:deBergCo08}), where $k$ is the output size.
\end{proof}

The second data structure\footnote{It is used for computing $\widehat{R}^{(q)}_{\Delta_i}$ in Step~3(b) of the algorithm of Theorem~5.2~\cite{ref:ChanOp12}} requires the performance in the following lemma.

\begin{lemma}\label{lem:140}
With $O(m^2(\log m)^{O(1)})$ time preprocessing, each query can be answered in $O((\log m)^{O(1)}+|H_{\sigma}|)$ time.
\end{lemma}
\begin{proof}
As discussed in the proof of Lemma~\ref{lem:130}, it suffices to compute $H_e$ for all
edges $e$ of the query pseudo-trapezoid $\sigma$. We still define {\em type-1 target arcs} (resp., {\em type-2 target arcs}) of $H$ in the same way.
To handle type-2 target arcs, as discussed in the proof of Lemma~\ref{lem:130}, for each edge $e$, all type-2 target arcs can be computed in $O(\log m + k)$ time after $O(m\log m)$ time preprocessing.
In the following, we focus on computing the type-1 target arcs.

We will build a two-level data structure. To this end, we first consider a sub-problem: Given a query point $q$ in $\overline{C'}$, compute the subset $H_q$ of arcs $h$ of $H$ whose underlying disks contain $q$.

We adapt the approach of Theorem~5.1~\cite{ref:MatousekRa93}. By our Cutting Theorem, we compute a hierarchical $(1/r)$-cutting $\Xi_0,\ldots,\Xi_k$ for $H$ with $r=c\cdot m$ for a constant $c\leq 1/8$. Consider a cell $\sigma\in \Xi_i$ for $i<k$. For each child cell $\sigma'$ of $\sigma$ in $\Xi_{i+1}$, let $H_{\sigma\setminus\sigma'}$ denote the subset of the arcs of $H$ crossing $\sigma$ but not crossing $\sigma'$. We partition $H_{\sigma\setminus\sigma'}$ into two subsets: one, denoted by $H_1(\sigma')$, consists of the arcs of $H_{\sigma\setminus\sigma'}$ whose underlying disks contain $\sigma'$ and the other, denoted by $H_2(\sigma')$, consists of the remaining arcs of $H_{\sigma\setminus\sigma'}$ (hence, the underlying disk of each arc of $H_2(\sigma')$ does not contain any point of $\sigma'$). We call $H_1(\sigma')$ and $H_2(\sigma')$ the {\em canonical subsets} of $\sigma'$. We store both subsets explicitly at $\sigma'$.
For each cell $\sigma$ of $\Xi_k$, we store  at $\sigma$ the set $H_{\sigma}$ of arcs of $H$ crossing $\sigma$. Note that $|H_{\sigma}|\leq m/r = c$. This finishes the preprocessing.

The total preprocessing time is $O(m^2)$. To see this, computing the hierarchical cutting takes $O(m^2)$ time by our Cutting Theorem. For each cell $\sigma\in \Xi_i$ with $i<k$, the number of arcs of $H$ crossing $\sigma$ is at most $m/\rho^i$. Hence, for each child cell $\sigma'$ of $\sigma$, we can compute its two canonical sets in $O(m/\rho^i)$ time. As there are $O(\rho^{2(i+1)})$ such cells $\sigma'$ in $\Xi_{i+1}$, the time we spend on computing these canonical sets for all cells of $\Xi_{i+1}$ is $O(m\rho^{i+2})$. Hence, the total time for computing all canonical sets in the preprocessing is $\sum_{i=0}^{k-1}m\rho^{i+2}=O(mr)=O(m^2)$.

Given a query point $q$, using the hierarchical cutting, we locate the cell $\sigma_i$ containing $q$ in each $\Xi_i$ for $0\leq i\leq k$. For each $\sigma_i$, we report the canonical subset $H_1(\sigma')$. Further, for $\sigma_k$, we check each arc in $H_{\sigma_i}$ and report it if its underlying disk contains $q$. As such, the query time is $O(\log m+|H_q|)$.

Note that the above approach can also be used to solve the following sub-problem: Given a query point $q$ in $\overline{C'}$,
compute the set $H_q$ of arcs $h$ of $H$ whose underlying disks
do not contain $q$. Indeed, instead of reporting $H_1(\sigma')$, we report $H_2(\sigma')$.

We now consider our original problem for computing all type-1 target arcs for each edge $e$. Observe that an arc $h$ of $H$ is a type-1 target arc if and only if one endpoint of $e$ is in $D$ while the other one is outside $D$, where $D$ is the underlying disk of $h$. We build a two-level data structure. In the first level, we build the same data structure as above. In the second level, for each canonical set $H_1(\sigma')$, we build the same data structure as above on $H_1(\sigma')$, denoted by $\scrD(\sigma')$.
The preprocessing time is $O(m^2\log m)$. To see this, for each cell $\sigma'$ in $\Xi_{i+1}$, the size of each of its two canonical subsets $H_1(\sigma')$ and $H_2(\sigma')$ is $O(m/\rho^{i})$, and thus the time for computing the secondary data structure $\scrD_j(\sigma')$ for $\sigma'$ is $m^2/\rho^{2i}$. As there are $O(\rho^{2(i+1)})$ cells $\sigma'$ in $\Xi_{i+1}$, the total time for constructing the secondary data structure for all cells in $\Xi_{i+1}$ is $O(m^2\rho^2)$. Therefore, the total preprocessing time is bounded by $O(m^2\log m)$.

Let $p_1$ and $p_2$ be the two endpoints of $e$, respectively. We first report all arcs of $H$ whose underlying disks contain $p_1$ but not $p_2$, and we then report all arcs of $H$ whose underlying disks contain $p_2$ but not $p_1$. By the above observation, all these arcs constitute $H_e$. We show below how to find the arcs in the former case; the algorithm for the latter case is similar. Using the hierarchical cutting, we locate the cell $\sigma_i$ containing $q$ in each $\Xi_i$ for $0\leq i\leq k$. For each $\sigma_i$, using the data structure $\scrD(\sigma_i)$, we report the arcs of $H_1(\sigma')$ whose underlying disks do not contain $p_2$. For $\sigma_k$, we also check each arc in $H_{\sigma_k}$ and report it if its underlying disk contain $p_1$ but not $p_2$. The total query time is $O(\log^2 m+|H_{e}|)$.
\end{proof}

\subsection{Trade-offs}
\label{sec:tradeoff}

Using cuttings and the results of Theorems~\ref{theo:40} and \ref{theo:randomize},
trade-offs between preprocessing and query time can be derived by standard techniques~\cite{ref:AgarwalSi17,ref:MatousekRa93}, as follows.

Consider a query disk $D$ whose center $q$ is in $C$. An easy observation is that a point $p\in P$ is contained in $D$ if and only if $q$ is contained in the disk $D_p$ centered at $p$. As such, we consider the problem in the dual setting.

Let $H$ be the set of lower arcs in $\overline{C}$ dual to the points of $P$.
The problem is equivalent to finding the arcs of $H$ whose underlying disks contain $q$.
In the preprocessing, by our Cutting Theorem, we compute a hierarchical $(1/r)$-cutting $\Xi_0,\ldots,\Xi_k$ for $H$.
Consider a cell $\sigma\in \Xi_i$ for $i<k$. For each child cell $\sigma'$ of $\sigma$ in $\Xi_{i+1}$, let $H_{\sigma\setminus\sigma'}$ denote the subset of the arcs of $H$ crossing $\sigma$ but not crossing $\sigma'$. We partition $H_{\sigma\setminus\sigma'}$ into two subsets: one, denoted by $H_1(\sigma')$, consists of the arcs of $H_{\sigma\setminus\sigma'}$ whose underlying disks contain $\sigma'$ and the other, denoted by $H_2(\sigma')$, consists of the remaining arcs of $H_{\sigma\setminus\sigma'}$ (hence, the underlying disk of each arc of $H_2(\sigma')$ does not contain any point of $\sigma'$). We call $H_1(\sigma')$ and $H_2(\sigma')$ the {\em canonical subsets} of $\sigma'$. We store their cardinalities at $\sigma'$. For each cell $\sigma$ of $\Xi_k$, we store at $\sigma$ the set $H_{\sigma}$ of arcs of $H$ crossing $\sigma$. Note that $|H_{\sigma}|\leq n/r$. Let $P_{\sigma}$ be the set of the dual points of $H_{\sigma}$. We build a unit-disk range counting data structure (e.g., Theorem~\ref{theo:40}) on $P_{\sigma}$, denoted by $\scrD_{\sigma}$, with  complexity $(T(|P_{\sigma}|),S(|P_{\sigma}|),Q(|P_{\sigma}|))$ for $P^*(\sigma)$, where $T(\cdot)$, $S(\cdot)$, and $Q(\cdot)$ are the preprocessing time, space, and query time, respectively.
 We refer to $\scrD_{\sigma}$ as the {\em secondary data structure}.
 This finishes the preprocessing.

For the preprocessing time, constructing the hierarchical cutting takes $O(nr)$ time. Constructing the secondary data structure $\scrD_{\sigma}$ for all cells $\sigma$ of $\Xi_k$ takes $O(r^2\cdot T(n/r))$ time. Hence, the total preprocessing time is $O(nr+r^2\cdot T(n/r))$. Following similar analysis, the space is $O(nr+r^2\cdot S(n/r))$.

Given a query disk $D$ with center $q$ in $C$, using the hierarchical cutting, we locate the cell $\sigma_i$ containing $q$ in each $\Xi_i$ for $0\leq i\leq k$. For each $\sigma_i$, we add $|H_1(\sigma_i)|$ to the total count. In addition, for $\sigma_k$, we use the secondary data structure $\scrD_{\sigma_k}$ to compute the number of points of $P_{\sigma_k}$ contained in $D$. Clearly, the query time is $O(\log r + Q(n/r))$.\footnote{The subsets $H_2(\sigma')$ computed in the preprocessing is ``reserved'' for answering {\em outside-disk} queries.}

As such, we obtain a data structure of $O(nr+r^2\cdot T(n/r))$ preprocessing time, $O(nr+r^2\cdot S(n/r))$ space, and $O(\log r + Q(n/r))$ query time.  Using the results of Theorems~\ref{theo:40} and ~\ref{theo:randomize} to build the secondary data structure for each $P_{\sigma}$, respectively, we can obtain the following trade-offs.

\begin{theorem}\label{theo:tradeoff}
\begin{enumerate}
\item
We can build an $O(nr)$ space data structure for $P$ in $O(nr(n/r)^{\delta})$ time, such that given any query disk $D$ whose center is in $C$, the number of points of $P$ in $D$ can be computed in $O(\sqrt{n/r})$ time, for any $1\leq r\leq n/\log^2 n$.
\item
We can build an $O(nr)$ space data structure for $P$ in $O(nr\log (n/r))$ expected time, such that given any query disk $D$ whose center is in $C$, the number of points of $P$ in $D$ can be computed in $O(\sqrt{n/r})$ time with high probability, for any $1\leq r\leq n/\log^2 n$.
\end{enumerate}
\end{theorem}

In particular, for the large space case, i.e., $r=n/\log^2n$, we can obtain the following corollary by Theorem~\ref{theo:tradeoff}(1) (a randomized result with slightly better preprocessing time can also be obtained by Theorem~\ref{theo:tradeoff}(2)).

\begin{corollary}\label{coro:largespace}
We can build an $O(n^2/\log^2 n)$ space data structure for $P$ in $O(n^2/\log^{2-\delta}n)$ time, such that given any query disk $D$ whose center is in $C$, the number of points of $P$ in $D$ can be computed in $O(\log n)$ time.
\end{corollary}

\subsection{Wrapping things up}
\label{sec:summary}

All above results on $P$ are for a pair of cells $(C,C')$ such that all points of $P$ are in $C'$ and centers of query disks are in $C$. Combining the above results with Lemma~\ref{lem:10}, we can obtain our results for the general case where points of $P$ and query disk centers can be anywhere in the plane.

In the preprocessing, we compute the information and data structure in Lemma~\ref{lem:10}, which takes $O(n\log n)$ time and $O(n)$ space. For each pair of cells $(C,C')$ with $C\in \calC$ and $C'\in N(C)$, we construct the data structure on $P(C')$, i.e., $P\cap C'$, with respect to query disks centered in $C$, e.g., those in Theorems~\ref{theo:30}, \ref{theo:40}, \ref{theo:randomize}, and \ref{theo:tradeoff}. As discussed before, due to property~(5) of $\calC$, the total preprocessing time and space is the same as those in the above theorems. Given a query disk $D$ with center $q$, by Lemma~\ref{lem:10}(2), we determine whether $q$ is in a cell $C$ of $\calC$ in $O(\log n)$ time. If no, then $D\cap P=\emptyset$ and thus we simply return $0$. Otherwise, the data structure returns $N(C)$. Then, for each $C'\in N(C)$, we use the data structure constructed for $(C,C')$ to compute $|P(C')\cap D|$. We return $|P\cap D|=\sum_{C'\in N(C)}|P(C')\cap D|$. As $|N(C)|=O(1)$, the total query time is as stated in the above theorems. We summarize these results below.

\begin{corollary}\label{coro:10}
Let $P$ be a set of $n$ points in the plane. Given a query unit disk $D$, the unit-disk range counting problem is to find the number of points of $P$ in $D$. We have the following results.
\begin{enumerate}
\item
An $O(n)$ space data structure can be built in $O(n\log n)$ time, with $O(\sqrt{n}(\log n)^{O(1)})$ query time.
\item
An $O(n)$ space data structure can be built in $O(n^{1+\delta})$ time for any small constant $\delta>0$, with $O(\sqrt{n})$ query time.
\item
An $O(n)$ space data structure can be built in $O(n\log n)$ expected time by a randomized algorithm, with $O(\sqrt{n})$ query time with high probability.
\item
An $O(n^2/\log^2 n)$ space data structure can be built in $O(n^2/\log^{2-\delta}n)$ time, with $O(\log n)$ query time.
\item
An $O(nr)$ space data structure can be built in $O(nr(n/r)^{\delta})$ time, with $O(\sqrt{n/r})$ query time, for any $1\leq r\leq n/\log^2 n$.
\item
An $O(nr)$ space data structure can be built in $O(nr\log (n/r))$ expected time by a randomized algorithm, with $O(\sqrt{n/r})$ query time with high probability, for any $1\leq r\leq n/\log^2 n$.
\end{enumerate}
\end{corollary}

\paragraph{Remark.} As the simplex range searching~\cite{ref:MatousekEf92,ref:MatousekRa93,ref:ChanOp12},
all results in Corollary~\ref{coro:10} can be easily extended to the weighted case (or the more general semigroup model) where each point of $P$ has a weight, i.e., each query asks for the total weight of all points in a query unit disk $D$.

%
%
%
%

\section{Applications}
\label{sec:app}

In this section, we demonstrate that our techniques for the disk range searching problem can be used to solve several other problems. More specifically, our techniques yield improved results for three classical problems: batched unit-disk range counting, distance selection, and discrete $2$-center.

\subsection{Batched unit-disk range counting}

Let $P$ be a set of $n$ points and $\calD$ be a set of $m$ (possibly overlapping) unit disks
in the plane. The {\em batched unit-disk range counting} problem (also referred to as {\em offline range searching} in the literature)
is to compute for each disk $D\in\calD$ the number of points of $P$ in $D$.

Let $Q$ denote the set of centers of the disks of $\calD$. For each point $q\in Q$, we use $D_q$ to denote the unit disk centered at $q$.

We first apply Lemma~\ref{lem:10} on $P$. For each point $q\in Q$, by Lemma~\ref{lem:10}(2), we first determine whether $q$ is in a cell $C$ of $\calC$. If no, then $D_q$ does not contain any point of $P$ and thus it can be ignored for the problem; without loss of generality, we assume that this case does not happen to any disk of $\calD$. Otherwise, let $C$ be the cell of $\calC$ that contains $q$. By Lemma~\ref{lem:10}(2), we further find the set $N(C)$ of $C$. In this way, in $O((n+m)\log n)$ time, we can compute $Q(C)$ for each cell $C$ of $\calC$, where $Q(C)$ is the subset of points of $Q$ in $C$. Define $\calD(C)$ as the set of disks of $\calD$ whose centers are in $Q(C)$. Let $P(C)=P\cap C$.

In what follows, we will consider the problem for $P(C')$ and $\calD(C)$ for each pair $(C,C')$ of cells with $C\in \calC$ and $C'\in N(C)$. Combining the results for all such pairs leads to the result for $P$ and $\calD$ (the details on this will be discussed later). To simplify the notation, we assume that $P(C')=P$ and $\calD(C)=\calD$ (thus $Q(C)=Q$). Hence, our goal is to compute $|P\cap D|$ for all disks $D\in \calD$.

If $C=C'$, then all points of $P$ are in $D$ for each disk $D\in \calD$ and thus the problem is trivial.
Below we assume $C\neq C'$. Without loss of generality, we assume that $C'$ and $C$ are separated by a horizontal line and $C'$ is above the line.
We assume that each point of $P$ defines a lower arc in $\overline{C}$ since otherwise the point can be ignored. We also assume that the boundary of each disk of $\calD$ intersects $\overline{C}'$, i.e., each point $q$ of $Q$ is dual to an upper arc $h_q$ in $\overline{C'}$, since otherwise the disk can be ignored.
Observe that a point $p$ is in $D_q$ if and only if $p$ is below the upper arc $h_q$ (we say that $p$ is below $h_q$ if $p$ is below the upper half boundary of $D_q$), for any $p\in P$ and $q\in Q$. Hence, the problem is equivalent to computing the number of points of $P$ below each upper arc of $H$, where $H= \{h_q \ |\ q\in Q\}$.

Given a set of $n$ points and a set of $m$ lines in the plane, Chan and Zheng~\cite{ref:ChanHo22} recently gave an $O(m^{2/3}n^{2/3}+n\log m + m\log n)$ time algorithm to compute the number of points below each line (alternatively, compute the number of points inside the lower half-plane bounded by each line). We can easily adapt their algorithm to solve our problem. Indeed, the main techniques of Chan and Zheng's algorithm we need to adapt to our problem are the hierarchical cuttings and duality. Using our Cutting Theorem and our definition of duality, we can apply the same technique and solve our problem in $O(m^{2/3}n^{2/3}+n\log m + m\log n)$ time, with $n=|P|$ and $m=|H|=|\calD|$. We thus have the following theorem.
The proof follows the framework of the algorithm in \cite{ref:ChanHo22};
to make our paper self-contained, we sketch the algorithm in the appendix.

\begin{theorem}\label{theo:70}
We can compute, for all disks $D\in\calD$, the number of points of $P$ in $D$ in $O(m^{2/3}n^{2/3}+n\log m + m\log n)$, with $n=|P|$ and $m=|\calD|$.
\end{theorem}

Let $\chi$ denote the number of intersections of the arcs of $H$, and thus $\chi=O(m^2)$.
Using our Cutting Theorem and Theorem~\ref{theo:70}, we further improve the algorithm for relatively small $\chi$.

\begin{theorem}\label{theo:75}
We can compute, for all disks $D\in\calD$, the number of points of $P$ in $D$ in $O(n^{2/3}\chi^{1/3}+m^{1+\delta}+n\log n)$ time, with $n=|P|$ and $m=|\calD|$.
\end{theorem}
\begin{proof}
We compute a hierarchical $(1/r)$-cutting $\Xi_0,\ldots,\Xi_k$ for $H$, where $r=\min\{m/8,(m^2/\chi)^{1/(1-\delta)}\}$ and $\delta$ refers to the parameter in the Cutting Theorem.
By our Cutting Theorem, the size of the cutting, denoted by $K$, is bounded by $O(r^{\delta}+\chi\cdot r^2/m^2)$ and the time for computing the cutting is $O(mr^{\delta}+\chi\cdot r/m)$. Since the parameter $r$ depends on $\chi$, which is not available to us, we can overcome the issue by using the standard trick of doubling. More specifically, initially we set $\chi$ to a constant. Then we run the algorithm until it exceeds the running time specified based on the guessed value of $\chi$. Next, we double the value $\chi$ and run the algorithm again. We repeat this process until when the algorithm finishes before it reaches the specified running time for a certain value of $\chi$. In this way, we run the cutting construction algorithm at most $O(\log \chi)$ time. Therefore, the total time for constructing the desired cutting is $O((mr^{\delta}+\chi\cdot r/m)\log \chi)$.

Next, we reduce the problem into $O(K)$ subproblems and then solve each subproblem by Theorem~\ref{theo:70}, which will lead to the theorem.

For each point $p\in P$, we find the cell $\sigma$ of $\Xi_i$ that contains $p$ and we
store $p$ in a {\em canonical subset} $P(\sigma)$ of $P$ (which is initially $\emptyset$), for all $0\leq i\leq k$, i.e., $P(\sigma)=P\cap \sigma$; in fact, we only need to store the cardinality of $P(\sigma)$. For ease of exposition, we assume that no point of $P$ lies on the boundary of any cell of $\Xi_i$ for any $i$.

For each disk $D\in \calD$, our goal is to compute the number of points of $P$ in $D$, denoted by $n_D$. We process $D$ as follows. We initialize $n_D=0$. Let $h$ be the upper arc of $H$ defined by $D$, i.e., $h=\partial D\cap \overline{C'}$. Starting from $\Xi_0=\overline{C'}$. Suppose $\sigma$ is a cell of $\Xi_i$
crossed by $h$ (initially, $i=0$ and $\sigma$ is $\overline{C'}$) and $i<k$. For each child cell $\sigma'$ of $\sigma$ in $\Xi_{i+1}$, if $\sigma'$ is contained in $D$, then we increase $n_D$ by $|P(\sigma')|$ because all points of $P(\sigma')$ are contained in $D$. Otherwise, if $h$ crosses $\sigma'$, then we proceed on $\sigma'$. In this way, the points of $P\cap D$ not counted in $n_D$ are those contained in cells $\sigma\in \Xi_k$ that are crossed by $h$. To count those points, we perform further processing as follows.

For each cell $\sigma$ in $\Xi_k$, if $|P_{\sigma}|>n/K$, then we arbitrarily partition $P(\sigma)$ into subsets of size between $n/(2K)$ and $n/K$, called {\em standard subsets} of $P(\sigma)$. As $\Xi_k$ has $O(K)$ cells and $|P|=n$, the number of standard subsets of all cells of $\Xi_k$ is $O(K)$.
Denote by $\calD_{\sigma}$ the subset of disks of $\calD$ whose boundaries cross $\sigma$. Our problem is to compute for all disks $D\in \calD_{\sigma}$ the number of points of $P(\sigma)$ contained in $D$, for all cells $\sigma\in \Xi_k$. To this end, for each cell $\sigma$ of $\Xi_k$, for each standard subset $P'(\sigma)$ of $P(\sigma)$, we solve the batched unit-disk range counting problem on the point set $P'(\sigma)$ and the disk set $\calD_{\sigma}$ by Theorem~\ref{theo:70}. Note that $|\calD_{\sigma}|\leq m/r$. As $\Xi_k$ has $O(K)$ cells, we obtain $O(K)$ subproblems of size $(n/K,m/r)$ each. As discussed above, solving these subproblems also solves our original problem. It remains to analyze the time complexity of the algorithm.


\paragraph{Time analysis.}
We use ``reduction algorithm'' to refer to the algorithm excluding the procedure for solving all subproblems by Theorem~\ref{theo:70}. For the time of the reduction algorithm, as discussed above, constructing the hierarchical cutting takes $O((nr^{\delta}+\chi\cdot r/n)\log \chi)$ time. Locating the cells of the cuttings containing the points of $P$ and thus computing $|P(\sigma)|$ for all cells $\sigma$ in the cuttings can be done in $O(n\log r)$ time, which is $O(n\log m)$. As the cutting algorithm also computes the sets $H_{\sigma}$ for all cells $\sigma\in \Xi_i$ for all $0\leq i\leq k$ and each cell of $\sigma\in \Xi_i$ has $O(1)$ child cells in $\Xi_{i+1}$, processing all disks $D$ (i.e., computing $n_D$ without counting those in the subproblems) takes $O(nr^{\delta}+\chi\cdot r/n)$ time.
Hence, the total time of the reduction algorithm is $O((nr^{\delta}+\chi\cdot r/n)\log \chi+n\log m)$.

As we have $O(K)$ subproblems of size $(n/K,m/r)$ each, by Theorem~\ref{theo:70}, the total time solving these subproblems is proportional to
\begin{equation*}
\begin{split}
& K\cdot \left\{\left(\frac{n}{K}\right)^{2/3}\left(\frac{m}{r}\right)^{2/3}+\frac{m}{r}\log \frac{n}{K}+\frac{n}{K}\log \frac{m}{r}\right\}\\
& =K^{1/3}n^{2/3}\left(\frac{m}{r}\right)^{2/3}+K\frac{m}{r}\log \frac{n}{K}+n\log \frac{m}{r}.
\end{split}
\end{equation*}
Therefore, the total time of the overall algorithm is proportional to
\begin{align}\label{equ:time}
(mr^{\delta}+\chi\cdot r/m)\log \chi+K^{1/3}n^{2/3}\left(\frac{m}{r}\right)^{2/3}+K\frac{m}{r}\log \frac{n}{K}+n\log m
\end{align}

Recall that $r=\min\{m/8,(m^2/\chi)^{1/(1-\delta)}\}$.
Depending on whether $r=m/8$ or $r=(m^2/\chi)^{1/(1-\delta)}$, there are two cases. We analyze the running time in each case below.

\paragraph{The case $r=m/8$.}
In this case, $m/8\leq (m^2/\chi)^{1/(1-\delta)}$, and thus $\chi=O(m^{1+\delta})$ and $K=O(m^{1+\delta})$.

By plugging these values into \eqref{equ:time}, we obtain that the total time is bounded by $O(m^{1+\delta}(\log m+\log \frac{n}{K})+n\log m+n^{2/3}m^{(1+\delta)/3})$.
Observe that $n^{2/3}m^{(1+\delta)/3}=O(n+m^{1+\delta})$. Hence, the total time is bounded by $O(m^{1+\delta}(\log m+\log \frac{n}{K})+n\log m)$.

Let $T=m^{1+\delta}(\log m+\log \frac{n}{K})+n\log m$.
We claim that $T=O(m^{1+\delta}\log m+n\log n)$. Indeed, if $n\leq K$, then $\log \frac{n}{K} = O(1)$ and $n=O(m^{1+\delta})$, and thus $T=O(m^{1+\delta}\log m)$. On the other hand, assume $n>K$. Notice that $K\geq r$ since $K$ is the size of a $(1/r)$-cutting. Hence, $n>r = m/8$ and thus $\log m=O(\log n)$. If $n<m^{1+\delta}$, then $\log \frac{n}{K} = O(\log m)$, and thus $T=O(m^{1+\delta}\log m+n\log n)$; otherwise, $T=O(n\log n)$. The claim thus follows.

Therefore, in the case $r=m/8$, the total time of the algorithm is $O(m^{1+\delta}\log m+n\log n)$.

\paragraph{The case $r=(m^2/\chi)^{1/(1-\delta)}$.} In this case, $m/8\geq (m^2/\chi)^{1/(1-\delta)}$, and thus $r=O(m)$ and $K=O(\chi\cdot r^2/m^2)$. Since $\chi=O(m^2)$, $\log \chi=O(\log m)$.

By plugging these values into \eqref{equ:time}, we obtain that the total time is bounded by $O(m^{1+\delta}\log m+K\cdot m/r\log(n/K)+n\log m+n^{2/3}\chi^{1/3})$.
Notice that
\begin{equation*}
\begin{split}
K\cdot m/r=\chi\cdot r/m=\frac{\chi}{m^2} \cdot r\cdot m = \frac{1}{r^{1-\delta}} \cdot r\cdot m=r^{\delta}\cdot m=O(m^{\delta})\cdot m=O(m^{1+\delta}).
\end{split}
\end{equation*}
Hence, the total time of the algorithm is $O(m^{1+\delta}(\log m+\log (n/K))+n\log m+n^{2/3}\chi^{1/3})$.
We have proved above that $m^{1+\delta}(\log m+\log (n/K))+n\log m= O(m^{1+\delta}\log m+n\log n)$. Therefore, the total time of the algorithm is bounded by $O(m^{1+\delta}\log m+n\log n+n^{2/3}\chi^{1/3})$.

\paragraph{Summary.}
Combining the above two cases, the total time of the algorithm is $O(m^{1+\delta}\log m+n\log n+n^{2/3}\chi^{1/3})$. Note that the factor $\log m$ of $m^{1+\delta}\log m$ in the running time is absorbed by $\delta$ in the theorem statement.
\end{proof}

\paragraph{The general problem.}
The above results are for the case where points of $P$ are in the square cell $C'$ while centers of $\calD$ are all in $C$. For solving the general problem where both $P$ and $\calD$ can be anywhere in the plane, as discussed before, we reduce the problem to the above case by Lemma~\ref{lem:10}. The properties of the set $\calC$ guarantee that the complexities for the general problem are asymptotically the same as those in Theorem~\ref{theo:70}. To see this, we consider all pairs $(C,C')$ with $C\in \calC$ and $C'\in N(C)$. For the $i$-th pair $(C,C')$, let $n_i=|P(C')|$ and $m_i=|\calD(C)|$. Then, solving the problem for the $i$-th pair $(C,C')$ takes $O(n_i^{2/3}m_i^{2/3}+m_i\log n_i+n_i\log m_i)$ time by Theorem~\ref{theo:70}. Due to the properties (4) and (5) of $\calC$, $\sum_i n_i=O(n)$ and $\sum_i m_i=O(m)$. Therefore, by H\"older's Inequality, $\sum_i n_i^{2/3}m_i^{2/3}\leq n^{1/3}\cdot \sum_i n_i^{1/3}m_i^{2/3}\leq n^{2/3}m^{2/3}$, and thus the total time for solving the problem for all pairs of cells is $O(n^{2/3}m^{2/3}+m\log n+n\log m)$.
Similarly, the complexity of Theorem~\ref{theo:75} also holds for the general problem, with $\chi$ as the number of pairs of disks of $\calD$ that intersect.

\paragraph{Computing incidences between points and circles.} It is easy to modify the algorithm to solve the following problem: Given $n$ points and $m$ unit circles in the plane, computing (either counting or reporting) the incidences between points and unit circles. The runtime is $O(n^{2/3}m^{2/3}+m\log n + n\log m)$ or $O(n^{2/3}\chi^{1/3}+m^{1+\delta}+n\log n)$, where $\chi$ is the number of intersecting pairs of the unit circles.
Although the details were not given, Agarwal and Sharir~\cite{ref:AgarwalPs05} already mentioned that an $n^{2/3}m^{2/3}2^{O(\log^*(n+m))}+O((m+n)\log (m+n))$ time algorithm can be obtained by adapting Matou\v{s}ek's technique~\cite{ref:MatousekRa93}. (The same problem for circles of arbitrary radii is considered in~\cite{ref:AgarwalPs05}. Refer to~\cite{ref:SharirCo17} for many other incidence problems.)
Our result further leads to an $O(n^{4/3})$-time algorithm for the {\em unit-distance detection} problem: Given $n$ points in the plane, is there a pair of points at unit distance? Erickson~\cite{ref:EricksonNe96} gave a lower bound of $\Omega(n^{4/3})$ time for the problem in his partition algorithm model.

\subsection{The distance selection problem}

Given a set $P$ of $n$ points in the plane and an integer $k$ in the
range $[0,n(n-1)/2]$, the distance selection problem is to compute the $k$-th smallest
distance among the distances of all pairs of points of $P$. Let
$\lambda^*$ denote the $k$-th smallest distance to be computed. Given a value
$\lambda$, the {\em decision problem} is to decide whether
$\lambda\geq \lambda^*$. Using our batched unit-disk range counting algorithm, we can easily obtain the following lemma.

\begin{lemma}\label{lem:160}
Given a value $\lambda$, whether $\lambda\geq \lambda^*$ can be decided in $O(n^{4/3})$ time.
\end{lemma}
\begin{proof}
We can use our algorithm for the offine unit-disk range counting problem.
Indeed, let $\calD$ be the set of congruent disks centered at the points of $P$ with radius $\lambda$.
By Theorem~\ref{theo:70}, we can compute in $O(n^{4/3})$ time the cardinality
$|\Pi|$, where $\Pi$ is the set of all disk-point incidences $(D,p)$, where $D\in \calD$, $p\in P$, and $D$ contains $p$. Observe
that for each pair of points $(p_i,p_j)$ of $P$ whose distance is at
most $\lambda$, it introduces two pairs in $\Pi$. Also, each point
$p_i$ introduces one pair in $\Pi$ because $p_i$ is contained in the
disk of $\calD$ centered at $p_i$. Hence, the number of pairs of points of $P$
whose distances are at most $\lambda$ is equal to $(|\Pi|-n)/2$.
Clearly, $\lambda\geq \lambda^*$ if and only if $(|\Pi|-n)/2\geq k$.
\end{proof}

Plugging Lemma~\ref{lem:160} into a randomized algorithm of
Chan~\cite{ref:ChanOn01} (i.e., Theorem~5~\cite{ref:ChanOn01}),
$\lambda^*$ can be computed in $O(n\log n+ n^{2/3}k^{1/3}\log
n)$ expected time.

\begin{theorem}
Given a set $P$ of $n$ points in the plane and an integer $k$ in the
range $[0,n(n-1)/2]$, the $k$-th smallest distance of $P$ can be
computed in $O(n\log n+ n^{2/3}k^{1/3}\log n)$ expected time by a randomized algorithm.
\end{theorem}

\subsection{The discrete $2$-center problem}

Let $P$ be a set of $n$ points in the plane. The discrete $2$-center problem is to find two smallest congruent disks whose centers are in $P$ and whose union covers $P$. Let $\lambda^*$ be the radius of the disks in an optimal solution. Given a value $\lambda$, the {\em decision problem} is to decide whether $\lambda\geq \lambda^*$.

Agarwal, Sharir, and Welzl~\cite{ref:AgarwalTh98} gave an $O(n^{4/3}\log^5 n)$ time algorithm by solving the decision problem first. A {\em key subproblem} in their decision algorithm~\cite{ref:AgarwalTh98} is the following. Preprocess $P$ to compute a collection $\calP$ of {\em canonical subsets} of $P$, $\{P_1,P_2,\ldots,\}$, so that given a query point $p$ in the plane, the set $P_p$ of points of $P$ {\em outside} the unit disk centered at $p$ can be represented as the union of a sub-collection $\calP_p$ of canonical subsets and $\calP_p$ can be found efficiently (it suffices to give the ``names'' of the canonical subsets of $\calP_p$). Note that here the radius of unit disks is $\lambda$.

Roughly speaking, suppose we can solve the above key subproblem with preprocessing time $T$ such that $\sum_{P_i\in \calP}|P_i|=M$ and $|\calP_p|$ for any query point $p$ is bounded by $O(\tau)$ (and $|\calP_p|$ can be found in $O(\tau)$ time); then the algorithm of Agarwal, Sharir, and Welzl~\cite{ref:AgarwalTh98} can solve the decision problem in  $O(T+M\log n+ \tau\cdot n\log^3 n)$ time.
With the decision algorithm, the optimal radius $\lambda^*$ can be found by doing binary search on all pairwise distances of the points of $P$ (in each iteration, find the $k$-th smallest distance using a distance selection algorithm); the total time is $O((T_1+T_2)\log n)$, where $T_1$ is the time of the distance selection algorithm and $T_2$ is the time of the decision algorithm.

Note that the logarithmic factor of $M\log n$ in the above running time of the decision algorithm of~\cite{ref:AgarwalTh98} is due to that for each canonical subset $P_i\in\calP$, we need to compute the common intersection of all unit disks centered at the points of $P_i$, which takes $O(|P_i|\log n)$ time~\cite{ref:HershbergerFi91}.
However, if all points of $P_i$ are sorted (e.g., by $x$-coordinate or $y$-coordinate), then the common intersection can be computed in $O(|P_i|)$ time~\cite{ref:WangOn20}. Therefore, if we can guarantee that all canonical subsets are sorted, then the runtime of the decision algorithm of~\cite{ref:AgarwalTh98} can be bounded by $O(T+M+ \tau\cdot n\log^3 n)$.

In the following, we present new solutions to the above key subproblem. We will show that after $T=O(n^{4/3}\log^2 n(\log\log n)^{1/3})$ expected time preprocessing by a randomized algorithm, we can compute
$M=O(n^{4/3}\log^2n/(\log\log n)^{2/3})$ sorted canonical subsets of $P$ so that $\tau=O(n^{1/3}(\log\log n)^{1/3}/\log n)$ holds with high probability. Consequently, the decision problem can be solved in $O(n^{4/3}\log^2 n(\log\log n)^{1/3})$ expected time, and thus $\lambda^*$ can be computed in $O(n^{4/3}\log^3 n(\log\log n)^{1/3})$ expected time if we use the $O(n^{4/3}\log^2 n)$ time distance selection algorithm in~\cite{ref:KatzAn97}.
We also have another slightly slower deterministic result. After $T=O(n^{4/3}\log^{7/3} n(\log\log n)^{1/3})$ time preprocessing algorithm, we can compute
$M=O(n^{4/3}\log^{7/3}n/(\log\log n)^{2/3})$ sorted canonical subsets of $P$ so that $\tau=O(n^{1/3}(\log\log n)^{O(1)}/\log^{2/3} n)$. Consequently, the decision problem can be solved in $O(n^{4/3}\log^{7/3} n(\log\log n)^{O(1)})$ time, and thus $\lambda^*$ can be computed in $O(n^{4/3}\log^{10/3} n(\log\log n)^{O(1)})$ time.

\paragraph{Remark.}
It is straightforward to modify our algorithms to achieve the same results for the following {\em inside-disk} problem: represent the subset of points of $P$ {\em inside} $D$ as a collection of pairwise-disjoint canonical sets for any query disk $D$.
\bigskip

In what follows, we present our solutions to the above subproblem.
We apply Lemma~\ref{lem:10} on the set $P$ to compute the set $\calC$ of square cells.
As before, we first reduce the problem to the same problem with respect to pairs of cells $(C,C')$ of $\calC$, by using Lemma~\ref{lem:10} as well as the following lemma (whose proof is based on a modification of the algorithm for Lemma~\ref{lem:10}); then we will solve the problem using our techniques for disk range searching.

\begin{lemma}\label{lem:170}
We can compute in $O(n\log n)$ time a collection of $O(n)$ sorted canonical subsets of $P$ whose total size is $O(n\log n)$, such that for any cell $C$ of $\calC$, there are $O(\log n)$ pairwise-disjoint canonical subsets whose union consists of the points of $P$ that are not in the cells of $N(C)$, and we can find those canonical subsets in $O(\log n)$ time.
\end{lemma}
\begin{proof}
Recall that in the algorithm for Lemma~\ref{lem:10} we have at most $n$ vertical strips that contain the points of $P$. Let $T$ be a complete binary tree whose leaves correspond to the strips from left to right. Each leaf has a canonical subset that is the set of points of $P$ in the strip. Each internal node also has a canonical subset that is the union of all canonical subsets in the leaves of the subtree rooted at the node.
In this way, $T$ has $O(n)$ canonical subsets whose total size is $O(n\log n)$.  The tree $T$ with all canonical subsets can be constructed in $O(n\log n)$ time. In addition, if we sort all points of $P$ by their $x$-coordinates at the outset, then all canonical subsets are sorted.

For each vertical strip $A$, there are multiple rectangles. Similarly as above, $T_A$ be a complete binary tree whose leaves correspond to these rectangles from top to bottom. We construct the canonical subsets for $T_A$ in a similar way as above (with canonical subsets sorted by $y$-coordinate). In this way, the trees for all strips together have $O(n)$ canonical subsets and their total size is $O(n\log n)$. All trees along with their sorted canonical subsets can be constructed in $O(n\log n)$ time.

For each rectangle $R$ in each strip, we have a list $L_R$ of all non-empty cells, sorted by their indices. We build a complete binary search tree $T_R$ whose leaves correspond to these cells in order. We construct canonical subsets for $T_R$ in a similar way as above (with each canonical subset sorted by $x$-coordinate). The tree $T_R$ can be constructed in $O(|P(R)|\log |P(R)|)$ time, where $P(R)=P\cap R$. Indeed, we can sort all points of $P(R)$ initially by $x$-coordinate. Observe that the canonical subsets of all nodes in the same level of $T(R)$ form a partition of $P(R)$. Since $P(R)$ are already sorted, sorting all canonical subsets in the same level of $T_R$ takes $O(|P(R)|)$ time. Therefore, the total time for constructing $T_R$ with all sorted canonical subsets is $O(|P(R)|\log |P(R)|)$.
In this way, the trees for all rectangles in all vertical strips together can be constructed in $O(n\log n)$ time; the trees have $O(n)$ canonical subsets in total and their total size is $O(n\log n)$.

Consider a cell $C\in \calC$. Suppose it is in a rectangle $R$ of a vertical strip $A$. Using the tree $T$, we can find in $O(\log n)$ time $O(\log n)$ disjoint canonical subsets whose union is exactly the set of points of $P$ in the vertical strips right (resp., left) of $A$. Similarly, using the tree $T_A$, we can find in $O(\log n)$ time $O(\log n)$ disjoint canonical subsets whose union is exactly the set of points of $P$ in the rectangles of $A$ above (resp., below) $R$. Finally, we find the canonical subsets whose union is exactly the set of points of $P(R)$ outside the cells of $N(C)$. Recall that $N(C)$ contains only cells in the five rows around $R$ and the cells in each row are consecutive (e.g., see Fig.~\ref{fig:grid}). Hence, removing $N(C)$ from the list $L_R$ divides $L_R$ into at most five sublists, which can be found in $O(\log n)$ time by doing binary search using cell indices. For each such sublist, by searching $T_R$, we can find in $O(\log n)$ time $O(\log n)$ disjoint canonical subsets whose union is exactly the set of points of $P(R)$ in the cells of the sublist.
In summary, we can find in $O(\log n)$ time $O(\log n)$ disjoint canonical subsets whose union is exactly the set of points of $P$ not in the cells of $N(C)$.
\end{proof}

Let $D_p$ be the unit disk centered at a point $p$ in the plane. If $p$ is not in any cell of $\calC$, then $D_p\cap P=\emptyset$ and thus we can return the entire set $P$ as a canonical subset. Henceforth, we only consider the case where $p$ is in a cell $C$ of $\calC$.  According to Lemma~\ref{lem:170}, it suffices to find canonical subsets to cover all points of $P\cap C'$ not in $D_p$ for all cells $C'\in N(C)$. As $|N(C)|=O(1)$, it suffices to consider one such cell $C'\in N(C)$. Hence, as before, the problem reduces to a pair of square cells $(C,C')$ of $\calC$ with $C'\in N(C)$. If $C'=C$, then we know that all points of $P\cap C'$ are in $D_p$. Hence, we assume that $C'\neq C$. Without loss of generality, we assume that $C'$ and $C$ are separated by a horizontal line and $C$ is below the line. The problem is to process all points of $P\cap C'$, such that given any query disk $D_p$ whose center $p$ is in $C$, we can find a collection of disjoint canonical subsets whose union is the set of points of $P\cap C'$ not in $D_p$. To simplify the notation, we assume that all $n$ points of $P$ are in $C'$.

Our data structure combines some techniques for the disk range searching problem.
As remarked before, all our results on disk range searching with respect to $(C,C')$ can be applied to find the number of points of $P$ outside any query disk $D$ whose center is in $C$ (indeed, the disk $D$ defines a spanning upper arc $h$ in $\overline{C'}$, and points in $D$ lie on one side of $h$ while points outside $D$ lie on the other side of $h$). Hence, our main idea is to examine our disk range searching data structures and define canonical subsets of $P$ in these data structures. For each query disk $D$, we apply the query algorithm on $D$, which will produce a collection of canonical subsects. The crux is to carefully design the disk range searching data structure (e.g., by setting parameters to some appropriate values) so that the following are as small as possible (tradeoffs are needed): the preprocessing time, the total size of all canonical subsets of the data structure, which is $M$, and the total number of canonical subsets for each query disk $D$, which is $\tau$. In the following, whenever we say ``apply our query algorithm on $D$'', we mean ``finding points outside $D$''. We will present two results, a randomized result based on Chan's partition trees~\cite{ref:ChanOp12} and a slightly slower deterministic result.

\subsubsection{The randomized result}

Our data structure has three levels. We will present them from the lowest level to the highest one.
We start with the lowest level, which relies on the partition tree $T$ built in Theorem~\ref{theo:randomize}.
For any disk $D$, we use $P\setminus D$ to refer to the subset of the points of $P$ not in $D$.

\begin{lemma}\label{lem:180}
We can compute in $O(n\log n)$ expected time a data structure with $O(n)$ sorted canonical subsets of $P$ whose total size is $O(n\log n)$, so that for any disk $D$ whose center is in $C$, we can find in $O(\kappa)$ time $O(\kappa)$ pairwise-disjoint canonical sets whose union is $P\setminus D$, where $\kappa=O(\sqrt{n})$ holds with high probability.
\end{lemma}
\begin{proof}
We build the partition tree $T$ in Theorem~\ref{theo:randomize}.
For each node $v$ of $T$, we define $P_v$ as a canonical subset. Hence, there are $O(n)$ canonical
subsets. Since the height of $T$ is $O(\log n)$, each point of $P$
is in $O(\log n)$ canonical subsets. Therefore, the total size of
all canonical subsets is $O(n\log n)$. To sort all canonical subsets, we sort all points of $P$ initially. Observe that all canonical subsets $P_v$ for the nodes $v$ at the same level of $T$ form a partition of $P$; hence, sorting all these canonical subsets takes $O(n)$ time based on the sorted list of $P$. As $T$ has $O(\log n)$ levels, sorting all canonical subsets of $T$ takes $O(n\log n)$ time.

For any disk $D$ whose center is in $C$, we apply the query algorithm on $D$ (looking for the points of $P$ outside $D$). The algorithm produces in $O(\kappa)$ time a collection of $O(\kappa)$ disjoint canonical subsets whose union is exactly $P\setminus D$, where $\kappa=O(\sqrt{n})$ holds with high probability.
\end{proof}

In the next lemma we add the second level to the data structure of Lemma~\ref{lem:180}. In fact, it is similar to the algorithm in Section~\ref{sec:tradeoff} except that we replace the secondary data structure $\scrD_{\sigma}$ for each cell $\sigma$ of $\Xi_k$ by Lemma~\ref{lem:180}.

\begin{lemma}\label{lem:190}
We can compute in $O(n^2\log\log n/\log^2 n)$ expected time a data structure with $O(n^2/\log^2 n)$ sorted canonical subsets of $P$ whose total size is $O(n^2\log \log n/\log^2 n)$, so that for any disk $D$ whose center is in $C$, we can find in $O(\kappa)$ time $O(\kappa)$ pairwise-disjoint canonical sets whose union is $P\setminus D$, where $\kappa=O(\log n)$ holds with high probability.
\end{lemma}
\begin{proof}
Let $H$ be the set of lower arcs in $\overline{C}$ dual to the points of $P$, which are in $C'$.
Consider a disk $D$ whose center $q$ is in $C$. The points of $P\setminus D$ are dual to the arcs of $H$ whose underlying disks do not contain $q$.
We compute a hierarchical $(1/r)$-cutting $\Xi_0,\ldots,\Xi_k$ for $H$ with $r=n/\log^2 n$. Consider a cell $\sigma\in \Xi_i$ for $i<k$.
For each child cell $\sigma'$ of $\sigma$ in $\Xi_{i+1}$, let $H_{\sigma\setminus\sigma'}$ denote the subset of the arcs of $H$ crossing $\sigma$ but not crossing $\sigma'$. We partition $H_{\sigma\setminus\sigma'}$ into two subsets: one, denoted by $H_1(\sigma')$, consists of the arcs of $H_{\sigma\setminus\sigma'}$ whose underlying disks contain $\sigma'$ and the other, denoted by $H_2(\sigma')$, consists of the remaining arcs of $H_{\sigma\setminus\sigma'}$ (hence, the underlying disk of each arc of $H_2(\sigma')$ does not contain any point of $\sigma'$). We call $H_1(\sigma')$ and $H_2(\sigma')$ the {\em canonical subsets} of $\sigma'$. We store $H_2(\sigma')$ explicitly at $\sigma'$.
For each cell $\sigma$ of $\Xi_k$, we store at $\sigma$ the set $H_{\sigma}$ of arcs of $H$ crossing $\sigma$. Note that $|H_{\sigma}|\leq n/r = \log^2n$. Let $P_{\sigma}$ be the set of points of $P$ dual to the arcs of $H_{\sigma}$. We build the data structure in Lemma~\ref{lem:180} on $P_{\sigma}$, denoted by $\scrD_{\sigma}$.
By Lemma~\ref{lem:180}, the canonical subsets in the secondary data structure have already been sorted. To sort other canonical subsets, we can first sort all points of $P$. With the help of the sorted list of $P$, all canonical subsets can be sorted in time linear in their total size.
This finishes the preprocessing.

For the preprocessing time, similar to the analysis of Lemma~\ref{lem:140}, constructing the hierarchical cutting, along with all canonical subsets, takes $O(nr)$ time. Constructing the secondary data structures $\scrD_{\sigma}$ for all cells $\sigma$ of $\Xi_k$ takes $O(r^2\log^{2}n\log\log n)$ expected time. Hence, the total expected time is $O(nr+r^2\log^{2}n\log\log n)=O(n^2\log\log n/\log^{2}n)$, excluding the time for sorting all canonical subsets not in the secondary data structure.
The total size of all canonical subsets excluding those in the secondary data structures is $O(nr)$. Hence, the total time for sorting these canonical subsets is $O(nr)$, after $P$ is sorted in $O(n\log n)$ time.
The total size of all canonical subsets of the secondary data structure $\scrD_{\sigma}$ is $O(r^2\cdot \log^2 n\log\log n)$ by Lemma~\ref{lem:180}. Therefore, the total size of all canonical subsets is $O(n^2\log\log n/\log^{2}n)$.
The number of all canonical subsets excluding those in the secondary data structure is $O(nr)$.
The number of all canonical subsets in the secondary data structure is $O(r^2\log^2 n)$ by Lemma~\ref{lem:180}. Hence, the total number of all canonical subsets is $O(n^2/\log^2n)$.

Given a disk $D_q$ with center $q\in C$, using the hierarchical cutting, we locate the cell $\sigma_i$ containing $q$ in each $\Xi_i$ for all $0\leq i\leq k$. For each $\sigma_i$, we report the canonical subset $H_2(\sigma_i)$ (note that we only need to report the ``name'' of this subset). There are $O(\log n)$ such canonical subsets, which can be found in $O(\log n)$ time. In addition, for $\sigma_k$, by Theorem~\ref{theo:randomize}, we use the secondary data structure $\scrD_{\sigma_k}$ to output the $O(\kappa)$ canonical subsets of $\scrD_{\sigma_k}$ in $O(\kappa)$ time, where $\kappa=O(\log n)$ holds with high probability. Hence, the total number of canonical subsets for $D_q$ is $O(\kappa+\log n)$. Clearly, these canonical subsets are disjoint and together form the set $P\setminus D$.
\end{proof}

We finally add the top-level data structure in the following lemma.

\begin{lemma}\label{lem:200}
For any $r<n/\log^{\omega(1)} n$, we can compute in $O(n\log n + n r\log\log r/\log^2 r)$ expected time a data structure with  $O(n r/\log^2 r)$ sorted canonical subsets of $P$ whose total size is $O(n\log(n/r)+n r\log\log r/\log^2 r)$, so that for any disk $D$ whose center is in $C$, we can find in $O(\kappa)$ time $O(\kappa)$ pairwise-disjoint canonical sets whose union is $P\setminus D$, where $\kappa=O(\sqrt{n/r}\log r)$ holds with high probability.
\end{lemma}
\begin{proof}
As in~\cite{ref:ChanOp12}, we build a {\em $r$-partial partition tree} $T$ for $P$ in the same way as the algorithm of Theorem~\ref{theo:randomize}, except that we stop the construction when a cell has less than $r$ points. Following the same algorithm of Theorem~4.2~\cite{ref:ChanOp12} and using our new algorithms for our problem as discussed in Section~\ref{sec:randomize}, for any $r<n/\log^{\omega(1)} n$, we can build a $r$-partition tree $T$ of size $O(n/r)$ in expected $O(n\log n)$ time with query time bounded by $O(\sqrt{n/r})$ with high probability (more specifically, for each disk $D$ whose center is in $C$, the number of leaf cells crossed by $h$ is $O(\kappa')$ and they can be found in $O(\kappa')$ time, where $\kappa'=O(\sqrt{n/r})$ holds with high probability).

As discussed in Section~\ref{sec:randomize}, each node of $T$ corresponds to a cell $\sigma$ (which is a pseudo-trapezoid) in $\overline{C'}$. Let $P(\sigma)$ be the set of points of $P$ in $\sigma$, and we define $P(\sigma)$ as the canonical subset of $\sigma$. Note that $|P(\sigma)|=O(r)$ if $\sigma$ is a leaf cell. As $T$ has $O(n/r)$ nodes, there are $O(n/r)$ canonical subsets in $T$. Since the height of $T$ is $O(\log (n/r))$, the total size of all these canonical subsets is $O(n\log (n/r))$ as each point of $P$ is in the canonical subset of only one cell in each level of $T$. If we sort $P$ initially, all these canonical subsets can be sorted in $O(n\log n)$ time.

Next, for each leaf cell $\sigma$ of $T$, we build the data structure $\scrD_{\sigma}$ in Lemma~\ref{lem:190} on $P(\sigma)$ as a secondary data structure. For all the secondary data structures, the total expected construction time is $O(n/r\cdot r^2\log\log r/\log^2 r)=O(n r\log\log r/\log^2 r)$, and the total number of canonical subsets is $O(n r/\log^2 r)$, whose total size is $O(n r\log\log r/\log^2 r)$.

This finishes our preprocessing. Based on the above discussions, the preprocessing takes $O(n\log n + n r\log\log r/\log^2 r)$ expected time. The number of canonical subsets, which are all sorted, is $O(n r/\log^2 r)$, and their total size is $O(n\log(n/r)+n r\log\log r/\log^2 r)$.

Let $D$ be a disk whose center is in $C$. Let $h=\partial D\cap \overline{C'}$, which is an upper arc of $\overline{C'}$ (note that the case $h=\emptyset$ can be easily handled). Using $T$, we can find $O(\kappa')$ canonical subsets in $O(\kappa')$ time as well as a set $\Sigma$ of $O(\kappa')$ leaf cells $\sigma$ crossed by $h$, where $\kappa'=\sqrt{n/r}$ holds with high probability; those canonical subsets are disjoint and their union is exactly the set of points of $P\setminus D$ not in the cells of $\Sigma$. Next, for each cell $\sigma\in \Sigma$, using the secondary data structure $\scrD_{\sigma}$, we compute in $O(\kappa'')$ time $O(\kappa'')$ canonical subsets whose union is exactly the points of $P\setminus D$ in $\sigma$, where $\kappa''=O(\log r)$ holds with high probability. Therefore, in total we can find in $O(\kappa)$ time $O(\kappa)$ pairwise-disjoint canonical subsets whose union is $P\setminus D$, where $\kappa=O(\sqrt{n/r}\log r)$ holds with high probability.
\end{proof}

By setting $r=n^{1/3}\log^4 n/(\log\log n)^{2/3}$ in the preceding lemma, we can obtain the following result.

\begin{corollary}
We can compute in $O(n^{4/3}\log^2 n(\log\log n)^{1/3})$ expected time by a randomized algorithm a data structure with $O(n^{4/3}\log^2 n/(\log\log n)^{2/3})$ sorted canonical subsets of $P$ whose total size is $O(n^{4/3}\log^2 n(\log\log n)^{1/3})$, so that for any disk $D$ whose center is in $C$, we can find in $O(\kappa)$ time $O(\kappa)$ pairwise-disjoint canonical sets whose union is $P\setminus D$, where $\kappa=O(n^{1/3}(\log\log n)^{1/3}/\log n)$ holds with high probability.
\end{corollary}

As discussed before, plugging our above result in the algorithm of~\cite{ref:AgarwalTh98}, we can solve the decision version of the discrete $2$-center problem in $O(n^{4/3}\log^2 n(\log\log n)^{1/3})$ expected time. Using the decision algorithm and the $O(n^{4/3}\log^2 n)$-time distance selection algorithm in~\cite{ref:KatzAn97}, the discrete $2$-center problem can be solved in $O(n^{4/3}\log^3 n(\log\log n)^{1/3})$ expected time.

\begin{theorem}
Given a set $P$ of $n$ points in the plane, the discrete $2$-center problem can be solved in $O(n^{4/3}\log^3 n(\log\log n)^{1/3})$ expected time by a randomized algorithm.
\end{theorem}

\subsubsection{The deterministic result}

The deterministic result also has three levels, which correspond to Lemmas~\ref{lem:180}, \ref{lem:190}, and \ref{lem:200}, respectively. The following lemma gives the lowest level, by using the partition tree of Theorem~\ref{theo:30}.

\begin{lemma}\label{lem:210}
We can compute in $O(n\log n)$ time a data structure with $O(n)$ sorted canonical subsets of $P$ whose total size is $O(n\log\log n)$, so that for any disk $D$ whose center is in $C$, we can find in $O(\sqrt{n}(\log n)^{O(1)})$ time $O(\sqrt{n}(\log n)^{O(1)})$ pairwise-disjoint canonical sets whose union is $P\setminus D$.
\end{lemma}
\begin{proof}
We build the partition tree $T$ in Theorem~\ref{theo:30}.
For each node $v$ of $T$, we define $P_v$ as a canonical subset. As $T$ has $O(n)$ nodes, there are $O(n)$ canonical subsets. Since the height of $T$ is $O(\log\log n)$, each point of $P$
is in $O(\log\log n)$ canonical subsets. Therefore, the total size of
all canonical subsets is $O(n\log\log n)$. These canonical subsets can be sorted in $O(n\log\log n)$ time after all points of $P$ are sorted initially.

For any disk $D$ whose center is in $C$, we apply the query algorithm on $D$. The algorithm produces in $O(\sqrt{n}(\log n)^{O(1)})$ time a collection of $O(\sqrt{n}(\log n)^{O(1)})$ pairwise-disjoint canonical subsets whose union is exactly $P\setminus D$.
\end{proof}

The following lemma follows the same algorithm as in Lemma~\ref{lem:190}, except that we use the data structure in Lemma~\ref{lem:210} as the secondary data structure.

\begin{lemma}\label{lem:220}
We can compute in $O(n^2\log\log n/\log^2 n)$ time a data structure with $O(n^2/\log^2 n)$ sorted canonical subsets of $P$ whose total size is $O(n^2\log \log\log n/\log^2 n)$, so that for any disk $D$ whose center is in $C$, we can find in $O(\log n(\log\log n)^{O(1)})$ time $O(\log n(\log\log n)^{O(1)})$ pairwise-disjoint canonical sets whose union is $P\setminus D$.
\end{lemma}
\begin{proof}
We follow the same algorithm as that for Lemma~\ref{lem:190}, except that for each cell $\sigma$ of $\Xi_k$, we build the data structure of Lemma~\ref{lem:210} on $P_{\sigma}$ as the secondary data structure. Following the same analysis will lead to the lemma.
\end{proof}

We finally add a partial half-space decomposition scheme of Theorem~5.2~\cite{ref:MatousekRa93} as the top level of our data structure.

\begin{lemma}\label{lem:230}
For any $r\leq n$, we can compute in $O(n\sqrt{r}+n\log n+r^2+(n^2/r)\log r\log\log(n/r)/\log^2(n/r))$ time a data structure with  $O(r\log r+ (n^2/r)\log r/\log^2(n/r))$ sorted canonical subsets of $P$ whose total size is $O(n\log^2 r+(n^2/r)\log r\log\log\log(n/r)/\log^2(n/r))$, so that for any disk $D$ whose center is in $C$, we can find in $O(\sqrt{r}\log(n/r)(\log\log(n/r))^{O(1)})$ time $O(\sqrt{r}\log(n/r)(\log\log(n/r))^{O(1)})$ pairwise-disjoint canonical sets whose union is $P\setminus D$.
\end{lemma}
\begin{proof}
We follow the algorithmic scheme of
Theorem~5.2~\cite{ref:MatousekRa93} and adapt it to our problem.

We first apply the Test Set Lemma to compute a test set $H$ of $r$ spanning upper arcs in $\overline{C'}$. Recall in the proof of the Test Set Lemma that this is done by computing a $(\Theta(\sqrt{r}))$-cutting $\Xi^*$ of size $O(r)$ for the lower-arcs of $\overline{C}$ dual to the points of $P$. For each cell $\sigma$ of $\Xi^*$, we define $S_{\sigma}$ as the set of upper arcs of $H$ dual to the at most four vertices of $\sigma$. Following the terminology of~\cite{ref:MatousekRa93}, we call $S_{\sigma}$ the {\em guarding set} of $\sigma$. 
Let $S$ be the collection of the guarding sets of all cells of $\Xi^*$.

The algorithm proceeds with $K=O(\log r)$ iterations. In each $i$-th iteration, $1\leq i\leq K$, a component of the data structure is computed, which has $O(r)$ canonical subsets of total size $O(n\log r)$ and is associated with a sub-collection $S_i'$ of guarding sets of $S$. We describe the $i$-th iteration as follows.
In the beginning we have a sub-collection $S_i$ of $S$; initially, $S_1=S$. Let $H_i$ be the union of all arcs of the sets of $S_i$. By our Cutting Theorem, we compute a hierarchical $(1/\sqrt{r})$-cutting $\Xi_0$, \ldots, $\Xi_k$ for $H_i$, whose total number of cells is $O(r)$. For each cell $\sigma$ of each cutting $\Xi_j$, let $P(\sigma)$ denote the subset of points of $P$ in $\sigma$ (for ease of discussion, we assume that each point of $P$ is in a single cell of $\Xi_j$). We say that a cell $\sigma\in \Xi_j$ is {\em fat} if $|P(\sigma)|\geq 2^{k-j}n/r$. We say that $\sigma$ is {\em active} if all its ancestor cells are fat, and $\sigma$ is a {\em leaf} cell if it is active and either not fat or in $\Xi_k$.

For each active cell $\sigma$ in any $\Xi_j$, we call $P(\sigma)$ a {\em canonical subset} of $P$. Hence, the total number of canonical subsets is $O(r)$, whose total size is $O(n\log r)$ as each point is in the canonical subset of a single cell at each $\Xi_i$.
For each leaf cell $\sigma$, if $P(\sigma)$ contains more than $n/r$ points, then we further partition it into subsets with sizes between $n/(2r)$ and $n/r$; we call them {\em standard subsets} ($P(\sigma)$ itself is also a standard subset if $|P(\sigma)|\leq n/r$). As the number of leaf cells is $O(r)$, the number of standard subsets is $O(r)$. For each standard subset, we build the data structure in Lemma~\ref{lem:220} as a secondary data structure.

Next, we define $S_i'$. For any upper arc $h$ in $\overline{C'}$, let $K_j(h)$ be the set of active cells of $\Xi_j$ crossed by $h$, and $K(h)=\bigcup_{j=1}^kK_j(h)$. We define the {\em weight} of $h$ as
$$w(h)=\sum_{j=1}^k4^{k-j}|K_j(h)|+\frac{1}{k}\sum_{\sigma\in K(h)}\sqrt{\frac{r}{n}|P(\sigma)|}\ \ +\sum_{\sigma \in K_k(h)}\frac{r}{n}|P(\sigma)|.$$
We define the weight of a guarding set $S_{\sigma}\in S_i$ as the sum of the weights of its arcs.
Let $W$ be the average weight of all guarding sets of $S_i$.
We define $S_i'$ as the set of the guarding sets whose weights are at most $2W$. Then, $|S_i'|\geq |S_i|/2$. We set $S_{i+1}=S_i\setminus S_i'$. The whole algorithm stops once $S_{i+1}=\emptyset$. Note that the algorithm has $O(\log r)$ steps.

This finishes the description of the preprocessing. Before analyzing the preprocessing time, we first show an observation. Let $D$ be a disk such that its center is in a cell $\sigma$ of the cutting $\Xi^*$ whose guarding set belongs to $S_i'$. Let $h=\partial D\cap \overline{C'}$, which is an upper arc in $\overline{C'}$. Then, by the same analysis as in Theorem~5.2~\cite{ref:MatousekRa93}, we can show the following {\em observation}: $h$ crosses at most $O(\sqrt{r})$ active cells of all cuttings $\Xi_j$, $1\leq j\leq k$, and the total number of standard subsets in all leaf cells crossed by $h$ is $O(\sqrt{r})$.

We now analyze the preprocessing time. Computing the test set $H$ takes $O(n\sqrt{r})$ time by the Cutting Theorem. Afterwards, the algorithm has $O(\log r)$ iterations. Consider the $i$-th iteration.
Computing the hierarchical cutting can be done in $O(|H_i|\sqrt{r})$ time by our Cutting Theorem. Computing the canonical subsets $P(\sigma)$ for all cells of all cuttings $\Xi_j$ takes $O(n\log r)$ time. For each arc $h\in H_i$, as the hierarchical cutting has $O(r)$ cells, computing the weight $w(h)$ can be done in $O(r)$ time. Hence, computing the weights for all arcs of $H_i$ takes $O(|H_i|r)$ time. As $|H_i|$ is geometrically decreasing, the total sum of $|H_i|$ in all iterations is $O(|H|)=O(r)$. Excluding those in the secondary data structures, each iteration computes $O(r)$ canonical subsets, whose total size is $O(n\log r)$; these subsets can be sorted in $O(n\log r)$ time if $P$ is initially sorted (before the first iteration).
Therefore, excluding the secondary data structures, the algorithm for all iterations runs in $O(n\log n + n\log^2 r+r^2)$ time and computes a total of $O(r\log r)$ sorted canonical subsets, whose total size is $O(n\log^2 r)$.

For the secondary data structures, because there are $O(r)$ standard subsets of $P$ generated in each iteration and the size of each subset is $O(n/r)$, by Lemma~\ref{lem:220}, the total time for building the secondary data structures is bounded by $O((n^2/r)\log r\log\log(n/r)/\log^2(n/r))$, and these data structures have a total of $O((n^2/r)\log r/\log^2(n/r))$ sorted canonical subsets, whose total size is $O((n^2/r)\log r\log\log\log(n/r)/\log^2(n/r))$.

In summary, our preprocessing runs in $O(n\sqrt{r}+n\log n+r^2+(n^2/r)\log r\log\log(n/r)/\log^2(n/r))$ time and computes a total of $O(r\log r+ (n^2/r)\log r/\log^2(n/r))$ sorted canonical subsets, whose total size is $O(n\log^2 r+(n^2/r)\log r\log\log\log(n/r)/\log^2(n/r))$.

Consider a disk $D$ whose center $q$ is in $C$. Let $h=\partial D\cap \overline{C'}$, which is an upper arc of $\overline{C'}$ (note that the case $h=\emptyset$ can be easily handled). Using the cutting $\Xi^*$, we find the cell $\sigma$ of $\Xi^*$ that contains $q$ in $O(\log r)$ time. Suppose the guarding set of $\sigma$ is in $S_i'$. Then, according to the observation discussed before, using the hierarchical cutting computed in the $i$-th iteration of the preprocessing algorithm, we can find in $O(\sqrt{r})$ time $O(\sqrt{r})$ canonical subsets as well as a set $\Sigma$ of leaf cells $\sigma$ crossed by $h$, such that those canonical subsets are pairwise disjoint and their union is exactly the set of points of $P\setminus D$ not in the cells of $\Sigma$. Next, for each cell $\sigma\in \Sigma$, for each standard subset $P'$ of $P(\sigma)$, using the secondary data structure built on $P'$, we compute in $O(\log(n/r)(\log\log(n/r))^{O(1)})$ time $O(\log(n/r)(\log\log(n/r))^{O(1)})$ canonical subsets whose union is the exactly the set of the points of $P\setminus D$ in $P'$. According to the above observation, there are $O(\sqrt{r})$ such standard subsets. Therefore, in total we can find in $O(\sqrt{r}\log(n/r)(\log\log(n/r))^{O(1)})$ time $O(\sqrt{r}\log(n/r)(\log\log(n/r))^{O(1)})$ pairwise-disjoint canonical subsets whose union is $P\setminus D$.
\end{proof}

By setting $r=n^{2/3}(\log\log n)^{2/3}/\log^{10/3}n$ in the preceding lemma, we obtain the following result.

\begin{corollary}
We can compute in $O(n^{4/3}\log^{7/3}n(\log\log n)^{1/3})$ time a data structure with a total of $O(n^{4/3}\log^{7/3}n/(\log\log n)^{2/3})$ sorted canonical subsets of $P$ whose total size is upper-bounded by $O(n^{4/3}\log^{7/3}n\log\log\log n/(\log\log n)^{2/3})$, so that for any disk $D$ whose center is in $C$, we can find in $O(n^{1/3}(\log\log n)^{O(1)}/\log^{2/3}n)$ time $O(n^{1/3}(\log\log n)^{O(1)}/\log^{2/3}n)$ pairwise-disjoint canonical sets whose union is $P\setminus D$.
\end{corollary}

According to our discussion before, plugging our above result in the algorithm of~\cite{ref:AgarwalTh98}, we can solve the decision version of the discrete $2$-center problem in $O(n^{4/3}\log^{7/3} n(\log\log n)^{O(1)})$ time. Using the decision algorithm and the $O(n^{4/3}\log^2 n)$-time distance selection algorithm in~\cite{ref:KatzAn97}, the discrete $2$-center problem can be solved in $O(n^{4/3}\log^{10/3} n(\log\log n)^{O(1)})$ time.

\begin{theorem}
Given a set $P$ of $n$ points in the plane, the discrete $2$-center problem can be solved in $O(n^{4/3}\log^{10/3} n(\log\log n)^{O(1)})$ time.
\end{theorem}

\section{Concluding remarks}
\label{sec:con}

In this paper, we present algorithms to adapt the techniques for simplex range searching to unit-disk range searching. We also show that our techniques can be used to derive improved algorithms for several classical problems.

Our techniques are likely to find other applications. Generally speaking, our techniques may be useful for solving problems involving a set of congruent disks in the plane. Our paper demonstrates that well-studied techniques for arrangements of lines may be adapted to solving problems involving arrangements of congruent disks. The general idea is to first reduce the problem to the same problem with respect to a pair of square cells using an algorithm like Lemma~\ref{lem:10}. Then, to tackle the problem on a pair of square cells $(C,C')$, we need to deal with an arrangement of spanning upper arcs in $\overline{C'}$ such that the centers of the underlying disks of these arcs are all in $C$. The properties of spanning upper arcs (e.g., Observation~\ref{obser:10}), along with the duality between the upper arcs in $\overline{C'}$ and the points in $C$, make an upper-arc arrangement ``resemble'' a line arrangement so that many algorithms and techniques on line arrangements may be easily adapted to the upper-arc arrangements.



\footnotesize
 \bibliographystyle{plain}

\begin{thebibliography}{10}

\bibitem{ref:AfshaniOp09}
Peyman Afshani and Timothy~M. Chan.
\newblock Optimal halfspace range reporting in three dimensions.
\newblock In {\em Proceedings of the 20th Annual ACM-SIAM Symposium on Discrete
  Algorithms (SODA)}, pages 180--186, 2009.

\bibitem{ref:AgarwalRa17}
Pankaj~K. Agarwal.
\newblock {\em {\em Range searching. In} Handbook of Discrete and Computational
  Geometry, {\em Csaba D. T\'{o}th, Joseph O'Rourke, and Jacob E. Goodman
  (eds.)}}, pages 1057--1092.
\newblock CRC Press, 3rd edition, 2017.

\bibitem{ref:AgarwalSi17}
Pankaj~K. Agarwal.
\newblock {\em {\em Simplex range searching and its variants: a review. In} A
  Journey Through Discrete Mathematics}, pages 1--30.
\newblock Springer, 2017.

\bibitem{ref:AgarwalSe93}
Pankaj~K. Agarwal, Boris Aronov, Micha Sharir, and Subhash Suri.
\newblock Selecting distances in the plane.
\newblock {\em Algorithmica}, 9:495--514, 1993.

\bibitem{ref:AgarwalOn94}
Pankaj~K. Agarwal and Ji\v{r}\'{i} Matou\v{s}ek.
\newblock On range searching with semialgebraic sets.
\newblock {\em Discrete and Computational Geometry}, 11:393--418, 1994.

\bibitem{ref:AgarwalOn13}
Pankaj~K. Agarwal, Ji\v{r}\'{i} Matou\v{s}ek, and Micha Sharir.
\newblock On range searching with semialgebraic sets. {II}.
\newblock {\em SIAM Journal on Computing}, 42:2039--2062, 2013.

\bibitem{ref:AgarwalCo93}
Pankaj~K. Agarwal, Marco Pellegrini, and Micha Sharir.
\newblock Counting circular arc intersections.
\newblock {\em SIAM Journal on Computing}, 22:778--793, 1993.

\bibitem{ref:AgarwalPs05}
Pankaj~K. Agarwal and Micha Sharir.
\newblock Pseudoline arrangements: Duality, algorithms, and applications.
\newblock {\em SIAM Journal on Computing}, 34:526--552, 2005.

\bibitem{ref:AgarwalTh98}
Pankaj~K. Agarwal, Micha Sharir, and Emo Welzl.
\newblock The discrete 2-center problem.
\newblock {\em Discrete and Computational Geometry}, 20:287--305, 1998.

\bibitem{ref:AggarwalSo90}
Alok Aggarwal, Mark Hansen, and Tom Leighton.
\newblock Solving query-retrieval problems by compacting {Voronoi} diagrams.
\newblock In {\em Proceedings of the 22nd Annual ACM Symposium on Theory of
  Computing (STOC)}, pages 331--340, 1990.

\bibitem{ref:BentleyA79}
Jon~L. Bentley and Hermann~A. Maurer.
\newblock A note on {Euclidean} near neighbor searching in the plane.
\newblock {\em Information Processing Letters}, 8:133--136, 1979.


\bibitem{ref:deBergCo08}
Mark de~Berg, Otfried Cheong, Marc~J. van Kreveld, and Mark~H. Overmars.
\newblock {\em Computational Geometry --- Algorithms and Applications}.
\newblock Springer-Verlag, Berlin, 3rd edition, 2008.

\bibitem{ref:ChanOn01}
Timothy~M. Chan.
\newblock On enumerating and selecting distances.
\newblock {\em International Journal of Computational Geometry and
  Application}, 11:291--304, 2001.

\bibitem{ref:ChanOp12}
Timothy~M. Chan.
\newblock Optimal partition trees.
\newblock {\em Discrete and Computational Geometry}, 47:661--690, 2012.

\bibitem{ref:ChanOp16}
Timothy~M. Chan and Konstantinos Tsakalidis.
\newblock Optimal deterministic algorithms for 2-d and 3-d shallow cuttings.
\newblock {\em Discrete and Computational Geometry}, 56:866--881, 2016.

\bibitem{ref:ChanHo22}
Timothy~M. Chan and Da~Wei Zheng.
\newblock Hopcroft’s problem, log-star shaving, {2D} fractional cascading,
  and decision trees.
\newblock In {\em Proceedings of the 33rd Annual ACM-SIAM Symposium on Discrete
  Algorithms (SODA)}, pages 190--210, 2022.

\bibitem{ref:ChazelleAn83}
Bernard Chazelle.
\newblock An improved algorithm for the fixed-radius neighbor problem.
\newblock {\em Information Processing Letters}, 16:193--198, 1983.

\bibitem{ref:ChazelleNe85}
Bernard Chazelle.
\newblock New techniques for computing order statistics in {Euclidean} space.
\newblock In {\em Proceedings of the 1st Annual Symposium on Computational
  Geometry (SoCG)}, pages 125--134, 1985.

\bibitem{ref:ChazelleCu93}
Bernard Chazelle.
\newblock Cutting hyperplanes for divide-and-conquer.
\newblock {\em Discrete and Computational Geometry}, 9(2):145--158, 1993.

\bibitem{ref:ChazelleNe86}
Bernard Chazelle, Richard Cole, Franco~P. Preparata, and Chee-Keng Yap.
\newblock New upper bounds for neighbor searching.
\newblock {\em Information and Control}, 68:105--124, 1986.

\bibitem{ref:ChazelleOp85}
Bernard Chazelle and Herbert Edelsbrunner.
\newblock Optimal solutions for a class of point retrieval problems.
\newblock {\em Journal of Symbolic Computation}, 1:47--56, 1985.

\bibitem{ref:ChazelleHa86}
Bernard Chazelle and Franco~P. Preparata.
\newblock Halfspace range search: An algorithmic application of {$k$}-sets.
\newblock {\em Discrete and Computational Geometry}, 1:83--93, 1986.

\bibitem{ref:ChazelleQu89}
Bernard Chazelle and Emo Welzl.
\newblock Quasi-optimal range searching in spaces of finite {VC}-dimension.
\newblock {\em Discrete and Computational Geometry}, 4(5):467--489, 1989.


\bibitem{ref:DyerLi84}
Martin~E. Dyer.
\newblock Linear time algorithms for two- and three-variable linear programs.
\newblock {\em SIAM Journal on Computing}, 13(1):31--45, 1984.

\bibitem{ref:EdelsbrunnerAl87}
Herbert Edelsbrunner.
\newblock {\em Algorithmis in Combinatorial Geometry}.
\newblock Heidelberg, 1987.

\bibitem{ref:EdelsbrunnerOp86}
Herbert Edelsbrunner, Leonidas~J. Guibas, and Jorge Stolfi.
\newblock Optimal point location in a monotone subdivision.
\newblock {\em SIAM Journal on Computing}, 15(2):317--340, 1986.

\bibitem{ref:EricksonNe96}
Jeff Erickson.
\newblock New lower bounds for {Hopcroft's} problem.
\newblock {\em Discrete and Computational Geometry}, 16:389--418, 1996.

\bibitem{ref:GoodrichGe93}
Michael~T. Goodrich.
\newblock Geometric partitioning made easier, even in parallel.
\newblock In {\em Proceedings of the 9th Annual Symposium on Computational
  Geometry (SoCG)}, pages 73--82, 1993.

\bibitem{ref:HershbergerFi91}
John Hershberger and Subhash Suri.
\newblock Finding tailored partitions.
\newblock {\em Journal of Algorithms}, 3:431--463, 1991.

\bibitem{ref:KatzAn97}
Matthew~J. Katz and Micha Sharir.
\newblock An expander-based approach to geometric optimization.
\newblock {\em SIAM Journal on Computing}, 26(5):1384--1408, 1997.

\bibitem{ref:KirkpatrickOp83}
David~G. Kirkpatrick.
\newblock Optimal search in planar subdivisions.
\newblock {\em SIAM Journal on Computing}, 12(1):28--35, 1983.

\bibitem{ref:MatousekCu91}
Ji\v{r}\'{i} Matou\v{s}ek.
\newblock Cutting hyperplane arrangement.
\newblock {\em Discrete and Computational Geometry}, 6:385--406, 1991.

\bibitem{ref:MatousekEf92}
Ji\v{r}\'{i} Matou\v{s}ek.
\newblock Efficient partition trees.
\newblock {\em Discrete and Computational Geometry}, 8(3):315--334, 1992.

\bibitem{ref:MatousekRe92}
Ji\v{r}\'{i} Matou\v{s}ek.
\newblock Reporting points in halfspaces.
\newblock {\em Computational Geometry: Theory and Applications}, 2:169--186,
  1992.

\bibitem{ref:MatousekRa93}
Ji\v{r}\'{i} Matou\v{s}ek.
\newblock Range searching with efficient hierarchical cuttings.
\newblock {\em Discrete and Computational Geometry}, 10(1):157--182, 1993.

\bibitem{ref:MatousekGe94}
Ji\v{r}\'{i} Matou\v{s}ek.
\newblock Geometric range searching.
\newblock {\em ACM Computing Survey}, 26:421--461, 1994.

\bibitem{ref:MatousekAp95}
Ji\v{r}\'{i} Matou\v{s}ek.
\newblock Approximations and optimal geometric divide-and-conquer.
\newblock {\em Journal of Computer and System Sciences}, 50(2):203--208, 1995.

\bibitem{ref:MatousekMu15}
Ji\v{r}\'{i} Matou\v{s}ek and Zuzana Pat\'akov\'a.
\newblock Multilevel polynomial partitions and simplified range searching.
\newblock {\em Discrete and Computational Geometry}, 54:22--41, 2015.

\bibitem{ref:MegiddoLi83}
Nimrod Megiddo.
\newblock Linear-time algorithms for linear programming in {$R^3$} and related
  problems.
\newblock {\em SIAM Journal on Computing}, 12(4):759--776, 1983.

\bibitem{ref:SharirCo17}
Micha Sharir.
\newblock Computational geometry column 65.
\newblock {\em SIGACT News}, 48:68--85, 2017.

\bibitem{ref:WangOn20}
Haitao Wang.
\newblock On the planar two-center problem and circular hulls.
\newblock In {\em Proceedings of the 36th International Symposium on
  Computational Geometry (SoCG)}, pages 68:1--68:14, 2020.

\bibitem{ref:YaoA83}
Frances~F. Yao.
\newblock A 3-space partition and its applications.
\newblock In {\em Proceedings of the 15th Annual ACM Symposium on Theory of
  Computing (STOC)}, pages 258--263, 1983.

\end{thebibliography}

\appendix

\normalsize
\section*{APPENDIX}

\section{Proof of Theorem~\ref{theo:cutting}}
\label{app:theocutting}

In this section, we provide the detailed algorithm for Theorem~\ref{theo:cutting}, by adapting Chazelle's algorithm~\cite{ref:ChazelleCu93}.
We actually present a more general algorithm that also works for other
curves in the plane (e.g., circles or circular arcs of different radii,
pseudo-lines, line segments, etc.).

Let $S$ be a set of algebraic arcs of constant complexity in the plane, i.e., each arc $s$ of $S$ is a connected portion (or the entire portion) of an algebraic curve defined by $O(1)$ real parameters (e.g., $s$ is an arc of a circle or the entire circle), such that any two arcs of $S$ intersect at most $O(1)$ times. For each arc $s$ of $S$, a point $p$ in the interior of $s$ is called a {\em break point} if $s$ has a vertical tangent at $p$ (i.e., $s$ is not $x$-monotone at $p$). Since $s$ is of constant complexity, $s$ has $O(1)$ break points, implying that $s$ can be partitioned into $O(1)$ $x$-monotone sub-arcs at these break points.

A region $\sigma$ in the plane is called {\em pseudo-trapezoid} if $\sigma$ is bounded from the left (resp., right) by a vertical line segment and bounded from the above (resp., below) by an $x$-monotone sub-arc of an arc of $S$.

We define a {\em hierarchical $(1/r)$-cutting} of $S$ in the same way as that for $H$ in Section~\ref{sec:cutting} except that each cell of the cutting is a pseudo-trapezoid defined above. We follow the similar notation to those in Section~\ref{sec:cutting} but with respect to $S$, e.g., $S_{\sigma}$, $\rho$, $\Xi_i$, $k$, $\chi$, etc. In particular, $\chi$ is the number of intersections of all arcs of $S$.
We will prove the following theorem, which immediately leads to Theorem~\ref{theo:cutting}.

\begin{theorem}\label{theo:cuttingnew}
For any $r\leq n$, a hierarchical $(1/r)$-cutting of size $O(r^2)$ for $S$ (together with the sets $S_{\sigma}$ for every cell $\sigma$ of $\Xi_i$ for all $0\leq i\leq k$) can be computed in $O(nr)$ time; more specifically, the size of the cutting is bounded by $O(r^{1+\delta}+\chi\cdot r^2/n^2)$ and the running time of the algorithm is bounded by $O(nr^{\delta}+\chi\cdot r/n)$, for any small $\delta>0$.
\end{theorem}

In the following, we first introduce some basic concepts (some of them are adapted from \cite{ref:ChazelleCu93}) and then describe the algorithm. For ease of exposition, we make a general position assumption that no three arcs of $S$ have a common intersection point.

\subsection{Basic concepts}
For any subset $R$ of $S$ and for any compact region $A$ in the plane, we use $R_A$ to denote the subset of arcs of $R$ that cross the interior of $A$ (note that if $A$ is a sub-arc of an arc $s$ of $R$, then $R_A$ does not contain $s$).
The {\em vertical pseudo-trapezoidal decomposition} of $R$, denoted by $\vd(R)$, is to extend a vertical line upwards (resp., downwards) until meeting an arc of $R$ or going to the infinity from each endpoint and each break point of each arc of $R$ as well as from each intersection of any two arcs of $R$. Each cell of $\vd(R)$ is a pseudo-trapezoid.
For any pseudo-trapezoid $\sigma$, we use $n_{\sigma}(R)$ to denote the number of intersections of the arcs of $R$ in the interior of $\sigma$.

We call $e$ a {\em canonical arc} if either $e$ is a vertical line segment or $e$ is a sub-arc of an arc of $S$. Note that each side of a pseudo-trapezoid is a canonical arc.

We adapt the concepts $\epsilon$-approximations and $\epsilon$-nets~\cite{ref:ChazelleCu93,ref:MatousekAp95} to our case. We say that a subset $R$ of $S$ is an {\em $\epsilon$-approximation} for $S$ if, for any canonical arc $e$, the following holds $$\Bigg|\frac{|S_e|}{|S|}-\frac{|R_e|}{|R|}\Bigg|<\epsilon.$$
A subset $R$ of $S$ is an {\em $\epsilon$-net} for $S$ if, for any canonical arc $e$, $|S_e|>\epsilon \cdot n$ implies that $R_e\neq \emptyset$.

The following lemma can be considered as a generalization of an $\epsilon$-approximation, which shows that an $\epsilon$-approximation of $S$ can be used to estimate $n_{\sigma}(S)$ for any canonical trapezoid $\sigma$. The proof of the lemma follows that of Lemma~2.1 in \cite{ref:ChazelleCu93}.

\begin{lemma}\label{lem:vertexcount}
Suppose $R$ is an $\epsilon$-approximation of $S$. Then, for any pseudo-trapezoid $\sigma$,
$$\Bigg|\frac{n_{\sigma}(S)}{|S|^2}-\frac{n_\sigma(R)}{|R|^2}\Bigg|<\epsilon.$$
\end{lemma}
\begin{proof}
Let $m=|R|$.
Let $s_1,s_2,\ldots,s_n$ denote the arcs of $S$ and let $r_1,r_2,\ldots,r_m$ denote the arcs of $R$.
For each $1\leq i\leq n$, let $s_i'$ denote the portion of $s_i$ inside $\sigma$; note that $s_i'$ may have multiple sub-arcs of $s_i$. For each $1\leq j\leq m$, let $r_i'$ denote the portion of $r_i$ inside $\sigma$.

We first observe that $\sum_{i=1}^{n}|S_{s_i'}|=2\cdot n_{\sigma}(S)$ and $\sum_{j=1}^{m}|R_{r_j'}|=2\cdot n_{\sigma}(R)$. Also, it holds that $\sum_{i=1}^{n}|R_{s_i'}|=\sum_{j=1}^{m}|S_{r_j'}|$.

As $R$ is an $\epsilon$-approximation for $S$, for $1\leq i\leq n$, it holds that
$$\Bigg|\frac{|S_{s_i'}|}{|S|}-\frac{|R_{s_i'}|}{|R|}\Bigg|<\epsilon.$$
Taking the sum for $1\leq i\leq n$ leads to
$$\Bigg|\frac{\sum_{i=1}^{n}|S_{s_i'}|}{|S|}-\frac{\sum_{i=1}^{n}|R_{s_i'}|}{|R|}\Bigg|<|S|\cdot \epsilon.$$
Since  $\sum_{i=1}^{n}|S_{s_i'}|=2\cdot n_{\sigma}(S)$, we obtain
\begin{equation}\label{equ:10}
\Bigg|\frac{2\cdot n_{\sigma}(S)}{|S|^2}-\frac{\sum_{i=1}^{n}|R_{s_i'}|}{|R|\cdot |S|}\Bigg|< \epsilon.
\end{equation}

On the other hand, as $R$ is an $\epsilon$-approximation for $S$, for $1\leq j\leq m$, it holds that
$$\Bigg|\frac{|S_{r_j'}|}{|S|}-\frac{|R_{r_j'}|}{|R|}\Bigg|<\epsilon.$$
Taking the sum for all $1\leq j\leq m$ leads to
$$\Bigg|\frac{\sum_{j=1}^{m}|S_{r_j'}|}{|S|}-\frac{\sum_{j=1}^{m}|R_{r_j'}|}{|R|}\Bigg|<|R|\cdot \epsilon.$$
Since $\sum_{j=1}^{m}|R_{r_i'}|=2\cdot n_{\sigma}(R)$, we obtain
\begin{equation}\label{equ:20}
\Bigg|\frac{\sum_{j=1}^{m}|S_{r_j'}|}{|S|\cdot |R|}-\frac{2\cdot n_{\sigma}(R)}{|R|^2}\Bigg|< \epsilon.
\end{equation}

Since $\sum_{i=1}^{n}|R_{s_i'}|=\sum_{j=1}^{m}|S_{r_i'}|$, combining \eqref{equ:10} and \eqref{equ:20} leads to $$\Bigg|\frac{n_{\sigma}(S)}{|S|^2}-\frac{n_\sigma(R)}{|R|^2}\Bigg|<\epsilon.$$
This proves the lemma.
\end{proof}

Our approach needs to compute $\epsilon$-approximations and $\epsilon$-nets. To this end, the following lemma shows that Matou\v{s}ek's algorithms~\cite{ref:MatousekAp95} can be applied.

\begin{lemma}\label{lem:epsapprox}
An $\epsilon$-approximation of size $O((1/\epsilon)^2\log(1/\epsilon))$ for $S$ and an $\epsilon$-net of size $O((1/\epsilon)\log(1/\epsilon))$ for $S$ can be computed in $O(n/\epsilon^8\cdot \log^4(1/\epsilon))$ time.
\end{lemma}
\begin{proof}
To avoid the lengthy background explanation, we use concepts from~\cite{ref:MatousekAp95} directly. Consider the range space with ``point set'' $S$ and {\em ranges} of the form $\{s\in S \ |\  s\cap e\neq \emptyset\}$, where $e$ is a canonical arc. The {\em shatter function} of this range space is bounded by $O(n^4)$. To see this, first observe that the distinct ranges defined by the sub-arcs $e$ of the arcs of $S$ is $O(n^4)$. On the other hand, consider the vertical pseudo-trapezoidal decomposition $\vd(S)$ of $S$. For each vertical line segment $e$, if it is in a single cell of $\vd(S)$, the range defined by $e$ is $\emptyset$; otherwise, the two endpoints of $e$ lie in two cells of $\vd(S)$, and the subset of arcs of $S$ {\em vertically between} the two cells (i.e., those arcs of $S$ intersecting $e$) defines a distinct range. As $\vd(S)$ has $O(n^2)$ cells, the size of ranges defined by vertical segments $e$ is $O(n^4)$.
Therefore, the shatter function of the range space is bounded by $O(n^4)$.

Thus, we can apply Theorem~4.1~\cite{ref:MatousekAp95}. Applying Theorem~4.1~\cite{ref:MatousekAp95} also requires a subspace oracle, which can be constructed as follows. Given a subset $S'\subseteq S$, we can construct all ranges of $S'$ in $O(m^5)$ time with $m=|S'|$, as follows.
For each arc $s\in S'$, we compute the intersections between $s$ and all other arcs of $S'$; consequently, all ranges defined by sub-arcs $e\subseteq s$ can be easily obtained in $O(m^3)$ time. In this way, all ranges defined by sub-arcs of the arcs of $S'$ can be constructed in $O(m^4)$ time.
To construct the ranges defined by vertical line segments, we first compute the vertical pseudo-trapezoidal decomposition $\vd(S')$. For every two cells of $\vd(S')$ that are intersected by a vertical line segment, we output the subset of all arcs of $S'$ vertically between them. This computes all ranges of $S'$ defined by vertical line segments. As $\vd(S')$ has $O(m^2)$ cells, the running time is easily bounded by $O(m^5)$.

As such, by applying Theorem~4.1~\cite{ref:MatousekAp95}, an $\epsilon$-approximation of size $O((1/\epsilon)^2\log(1/\epsilon))$ for $S$ can be computed in $O(n/\epsilon^8\cdot \log^4(1/\epsilon))$ time. Similarly, by applying Corollary~4.5~\cite{ref:MatousekAp95}, an $\epsilon$-net of size $O((1/\epsilon)\log(1/\epsilon))$ for $S$ can be computed in $O(n/\epsilon^8\cdot \log^4(1/\epsilon))$ time.
\end{proof}

An $\epsilon$-net $R$ of $S$ is {\em sparse} for a pseudo-trapezoid $\sigma$ if $n_{\sigma}(R)/n_{\sigma}(S)\leq 4|R|^2/|S|^2$. We have the following lemma adapted from Lemma~3.2 in~\cite{ref:ChazelleCu93}.

\begin{lemma}\label{lem:epsnet}
For any pseudo-trapezoid $\sigma$, an $\epsilon$-net of $S$ sparse for $\sigma$ of size $O(1/\epsilon\cdot \log n)$ can be computed in time polynomial in $n$.
\end{lemma}
\begin{proof}
We adapt the algorithm for Lemma~3.2 in~\cite{ref:ChazelleCu93}.
We sketch the main idea and focus on the differences.

We first show that a random sample $R\subseteq S$ of size
$m=\min\{\lceil5\epsilon^{-1}\log n\rceil, n\}$ forms an $\epsilon$-net sparse
for $s$ with probability greater than $1/2$. We assume that $m<n$.
	In the same way as in Lemma~3.2~\cite{ref:ChazelleCu93}, we can
	show that $\text{Prob}[n_s(R)/n_s(S)>4m^2/n^2]<1/4$.

We define a collection $\calS$ of subsets of $S$, as follows. For each arc $s\in S$, the intersections of $s$ with the other arcs of $S$ cut $s$ into sub-arcs; for any pair of points $(p,q)$ on different sub-arcs of $s$, we add the subsets of arcs of $S$ crossing the sub-arc of $s$ between $p$ and $q$ to $\calS$. Clearly, at most $O(n^4)$ subsets of $S$ can be added to $\calS$. On the other hand, for any two cells of the decomposition $\vd(S)$ that intersect the same vertical line segment, we add to $\calS$ the subset of arcs of $S$ vertically between these two cells; at most $O(n^4)$ subsets can be added to $\calS$ in this way as $\vd(\calS)$ has $O(n^2)$ cells. As such, the size of $\calS$ is $O(n^4)$.

According to the definition of $\calS$, to ensure that $R$ is an
$\epsilon$-net of $S$, it suffices to guarantee that
$|S'|>\epsilon n$ implies $S'\cap R\neq\emptyset$ for each
subset $S'$ of $\calS$. In the same way as in Lemma~3.2~\cite{ref:ChazelleCu93}, we can prove that
the probability $p_{S'}$ for $S'\in \calS$ failing that test is less than $1/n^5$.
Since $|\calS|=O(n^4)$, for
large $n$, $\sum_{S'\in \calS}p_{S'}<1/4$ holds. Therefore, we
obtain
$$\text{Prob}\left[\frac{n_s(R)}{n_s(S)}>4\cdot
\frac{m^2}{n^2}\right]+\sum_{S'\in \calS}p_{S'}<1/2.$$

As such, the probability that $R$ is an $\epsilon$-net sparse for $\sigma$ is
larger than $1/2$.

Following the approach of Lemma~3.2~\cite{ref:ChazelleCu93}, the above proof can be converted to a polynomial time algorithm to find such a subset $R$. Refer to Lemma~3.2~\cite{ref:ChazelleCu93} for the details.
\end{proof}

\subsection{The algorithm}
We are now in a position to describe the algorithm for computing a hierarchical
$(1/r)$-cutting for $S$.
We first assume that $r\leq n/8$, and the case $r>n/8$ will be discussed later.

Since a $(1/r)$-cutting is also a $(1/r')$-cutting for any $r'<r$, we can assume that $r\geq \rho$ for some
appropriate constant $\rho$ such that $r=\rho^k$ for some integer $k$. Thus $k=\Theta(\log r)$. The algorithm has $k$ iterations. For each $1\leq i\leq k$, the $i$-th iteration computes a $(1/\rho^i)$-cutting $\Xi_i$ by refining the cutting $\Xi_{i-1}$. Each cell of $\Xi_{i}$ is a pseudo-trapezoid.
Initially, we let $\Xi_0$ be the entire plane. Clearly, $\Xi_0$ is a $(1/\rho^0)$-cutting for $S$. Next, we describe the algorithm for a general iteration to compute $\Xi_i$ based on $\Xi_{i-1}$. We assume that the subsets $S_{\sigma}$ are available for all cells $\sigma\in \Xi_{i-1}$, which is true initially when $i=1$.

Consider a cell $\sigma$ of $\Xi_{i-1}$. We process $\sigma$ as follows. If $|S_{\sigma}|\leq n/\rho^i$, then we do nothing with $\sigma$, i.e., $\sigma$ becomes a
cell of $\Xi_i$. Otherwise, we first compute a $(1/(8\rho_0))$-approximation $A$ of size $O(\rho_0^2\log \rho_0)$ for $S_{\sigma}$ by Lemma~\ref{lem:epsapprox}, and then compute a $(1/(8\rho_0))$-net $R$ of size $O(\rho_0\log \rho_0)$ for $A$ that is sparse for $\sigma$ by Lemma~\ref{lem:epsnet}, with $\rho_0=(\rho^i/n)\cdot |S_{\sigma}|$.

\paragraph{Remark.}
The parameter $8\rho_0$ is $4\rho_0$ in~\cite{ref:ChazelleCu93}. We use a different parameter in order to prove Lemma~\ref{lem:size} because each cell in our cutting is a pseudo-trapezoid. This is also the reason we assume $r\leq n/8$.
\medskip

We compute the vertical pseudo-trapezoidal decomposition of the arcs of $R$, and clip it inside $\sigma$; let  $\vd_{\sigma}(R)$ denote the resulting decomposition inside $\sigma$.
We include $\vd_{\sigma}(R)$ into $\Xi_i$. Finally, for each cell $\sigma_0\in \vd_{\sigma}(R)$, we compute $S_{\sigma_0}$, i.e., the subset of the arcs of $S$ that intersect the interior of $\sigma_0$, by simply checking every arc of $S_{\sigma}$. This finishes the processing of $\sigma$.

The cutting $\Xi_i$ is obtained after we process every cell $\sigma$ of $\Xi_{i-1}$ as above. The next lemma shows that $\Xi_i$ is a $(1/\rho^i)$-cutting of $S$.

\begin{lemma}\label{lem:size}
$\Xi_i$ thus obtained is a $(1/\rho^i)$-cutting of $S$.
\end{lemma}
\begin{proof}
We follow the notation discussed above.
Consider a cell $\sigma_0$ created in $\sigma$ as discussed above.
Our goal is to show that $|S_{\sigma_0}|\leq n/\rho^i$.

We claim that for any canonical arc $e$ in $\sigma_0$, the number of arcs of $S_{\sigma}$
crossing the interior of $e$ is no more than
$|S_{\sigma}|/(4\rho_0)$. Indeed, assume to the contrary that this is
not true. Then, since $A$ is a $1/(8\rho_0)$-approximation of
$S_{\sigma}$, the interior of $e$ would intersect more than
$$\frac{|S_{\sigma}|}{4\rho_0}\cdot
\frac{|A|}{|S_{\sigma}|}-\frac{|A|}{8\rho_0}=\frac{|A|}{8\rho_0}$$
segments of $A$. Thus, the interior of $e$ must be crossed by at least one arc
of $R$, for $R$ is a $(1/(8\rho_0))$-net of $A$. But this is
impossible because $e$ is in $\sigma_0$ (and thus its interior cannot be crossed by any arc of $R$).

Note that each vertex of $\sigma_0$ is on an arc of $S$, which bounds $\sigma_0$. By the general position
assumption, $S_{\sigma_0}$ has at most one arc through a vertex of
$\sigma_0$. Hence, $S_{\sigma_0}$ has at most four arcs through
vertices of $\sigma_0$. Let $S_{\sigma_0}'$ denote the subset of arcs
of $S_{\sigma_0}$ not through any vertex of $\sigma_0$. Let $e_i$,
$i=1,2,3,4$, be the four canonical arcs on the boundary of $\sigma_0$, respectively.
For any arc $s$ of $S_{\sigma_0}'$, since it does not contain any
vertex of $\sigma_0$, $s$ must
cross the interiors of two of $e_i$, for $i=1,2,3,4$. As the interior
of each $e_i$, $1\leq i\leq 4$, can be crossed by at most
$|S_{\sigma}|/(4\rho_0)$ arcs of $S_{\sigma}$, the size of
$S_{\sigma_0}'$ is at most $|S_{\sigma}|/(4\rho_0)\cdot
4/2=|S_{\sigma}|/(2\rho_0)$.
Therefore, $|S_{\sigma_0}|\leq |S_{\sigma_0}'|+4 \leq |S_{\sigma}|/(2\rho_0)+4=n/(2\rho^i)+4$.

Since $r=\rho^k$ and $i\leq k$, $\rho^{i}\leq r$. As $r\leq n/8$, we obtain $\rho^i\leq n/8$, and thus $4\leq n/(2\rho^i)$. Therefore, we obtain $|S_{\sigma_0}|\leq n/\rho^i$.
\end{proof}

\paragraph{The size of $\Xi_i$.}
We next analyze the size of $\Xi_i$. First notice that $\rho_0\leq \rho$.
Since $A$ is a $(1/(8\rho_0))$-approximation of $S_{\sigma}$, by Lemma~\ref{lem:vertexcount}, we have
$$\Bigg|\frac{n_{\sigma}(S_{\sigma})}{|S_{\sigma}|^2}-\frac{n_\sigma(A)}{|A|^2}\Bigg|<\frac{1}{8\rho_0}.$$
Because $R$ is a $(1/(8\rho_0))$-net of $A$ sparse for $\sigma$, we have $n_{\sigma}(R)/n_{\sigma}(A)\leq 4|R|^2/|A|^2$.
This implies that
$$n_{\sigma}(R)\leq 4\cdot \frac{|R|^2}{|S_{\sigma}|^2}\cdot n_{\sigma}(S_{\sigma})+\frac{|R|^2}{2\rho_0}.$$
Since $|R|=O(\rho_0\log \rho_0)$, the number of cells of $\vd_{\sigma}(R)$ is proportional to
$|R|+n_{\sigma}(R)$, which is at most proportional to
\begin{equation}\label{equ:30}
\frac{\rho_0^2\log^2\rho_0}{|S_{\sigma}|^2}\cdot n_{\sigma}(S_{\sigma})+\rho_0\log^2\rho_0.
\end{equation}

Recall that $\rho_0=(\rho^i/n)\cdot |S_{\sigma}|$ and $\rho_0\leq
\rho$. Hence, $\frac{\rho_0\log\rho_0}{|S_{\sigma}|}\leq
\frac{\rho^i}{n}\log \rho$. Recall that $\chi$ is the total number of intersections of all arcs of $S$. Observe that
$\sum_{\sigma\in \Xi_{i-1}}n_{\sigma}(S_{\sigma})\leq \chi$. Let $|\Xi_i|$ denote the number of
cells in $\Xi_i$. Taking the sum of \eqref{equ:30} for all cells $\sigma\in
\Xi_{i-1}$, we obtain that the following holds for a constant $c$:
$$|\Xi_i|\leq c\cdot \left(\frac{\rho^{i}\log \rho}{n}\right)^2\cdot \chi + c\cdot
\rho\log^2\rho\cdot |\Xi_{i-1}|.$$

Since $|\Xi_0|=1$, we can prove by induction that $|\Xi_i|\leq \rho^{2(i+1)}\cdot \chi/n^2+\rho^{(i+1)(1+\delta)}$ for a large enough constant $\rho$, for any $\delta>0$. As such, the size of the last cutting $\Xi_k$ is $O(\chi\cdot r^2/n^2+r^{1+\delta})$ since $r=\rho^k$ and $\rho$ is a constant.

\paragraph{The time analysis.}
Using the above bound of $|\Xi_i|$, we can show that the running time of the algorithm is bounded by $O(nr^{\delta}+\chi\cdot r/n)$. Indeed, for each $\sigma\in \Xi_{i-1}$, since $\rho_0\leq \rho$ and $\rho$ is a constant, by Lemma~\ref{lem:epsapprox}, it takes $O(|S_{\sigma}|)$ time to compute $A$ and constant time to compute $R$ and obtain the decomposition $\vd_{\sigma}(R)$. The subsets $S_{\sigma'}$ for all cells $\sigma'\in \vd_{\sigma}(R)$ can also be obtained in $O(|S_{\sigma}|)$ time since $|R|=O(1)$. Therefore, the total time of the algorithm is at most proportional to
$$\sum_{i=0}^{k}\sum_{\sigma\in \Xi_i}|S_{\sigma}|\leq \sum_{i=0}^{k}|\Xi_i|\cdot \frac{n}{\rho^i}\leq \sum_{i=0}^{k}\left(\rho^{2(i+1)}\cdot \frac{\chi}{n^2}+\rho^{(i+1)(1+\delta)}\right)\cdot  \frac{n}{\rho^i},$$
which is bounded by $O(nr^{\delta}+\chi\cdot r/n)$ since $r=\rho^k$ and $\rho$ is a constant.

This proves Theorem~\ref{theo:cuttingnew} for $r\leq n/8$. Note that each cell in the last cutting $\Xi_k$ is a pseudo-trapezoid.

\paragraph{The case $r>n/8$.}
If $r>n/8$, then we first run the above algorithm with respect to $r'=n/8$ to compute a hierarchical $(1/r')$-cutting $\Xi_0,\Xi_1,\ldots, \Xi_k$. We then perform additional processing as follows. For each cell $\sigma\in \Xi_k$, we know that $|S_{\sigma}|\leq 8$. We compute the vertical pseudo-trapezoidal decomposition of the arcs of $S_{\sigma}$ inside $\sigma$. The resulting pseudo-trapezoids in the decompositions of all cells $\sigma\in \Xi_k$ constitute $\Xi_{k+1}$. It is easy to see that $S_{\sigma'}=\emptyset$ for each $\sigma'\in \Xi_{k+1}$. Hence, $\Xi_{k+1}$ is a $(1/r)$-cutting for $S$ (more precisely, $\Xi_0,\Xi_1,\ldots, \Xi_{k+1}$ form a hierarchical $(1/r)$-cutting). Also, since $|S_{\sigma}|\leq 8$ for all cells $\sigma\in \Xi_k$, following the analysis as above, the complexities of Theorem~\ref{theo:cuttingnew} hold.

\paragraph{Remark.}
Chan and Tsakalidis~\cite{ref:ChanOp16} derived a simpler algorithm than Chazelle's algorithm~\cite{ref:ChazelleCu93} for computing cuttings for a set of lines in the plane (e.g., $\epsilon$-nets and $\epsilon$-approximations are avoided). However, it appears that there are some difficulties to adapt their technique to computing cuttings for curves in the plane. One difficulty, for example, is that their algorithm uses the algorithm of Megiddo~\cite{ref:MegiddoLi83} and
Dyer~\cite{ref:DyerLi84} to construct a (7/8)-cutting of size $4$, which relies on the slopes of lines. For curves, however, it is not clear to use how to define ``slopes'' (e.g., how to define slope for a circle?).

\section{Proof of Theorem~\ref{theo:70}}

The problem is to compute for each upper arc $h$ of $H$ the number of points of $P$
below it (i.e., the number of points of $P$ inside the underlying disk of $h$), where $H= \{h_q \ |\ q\in Q\}$. Recall that $n=|P|$ and $m=|H|$.
We refer to it as the {\em symmetric case} if $n= m$ and {\em asymmetric case} otherwise.
Let $T(m,n)$ denote the time complexity for the problem $(H,P)$ of size $(m,n)$.
In what follows, we first present two algorithms and then combine them.

\paragraph{The first algorithm.}
Using our Cutting Theorem, we compute in $O(mr)$ time a hierarchical $(1/r)$-cutting
$\Xi_0,\ldots,\Xi_k$ for $H$.
For any cell $\sigma$ of $\Xi_i$, $0\leq i\leq k$, let $P(\sigma)=P\cap
\sigma$, i.e., the subset of points of $P$ in $\sigma$.
For each point $p\in P$, we find the cell $\sigma$ of $\Xi_k$ that
contains $p$. For reference purpose later, we call this step {\em the point
location step}, which can be done in $O(n\log r)$ time for all points of $P$.
After this step, the subsets $P(\sigma)$ for all cells $\sigma\in \Xi_k$
are computed. We also need to maintain the cardinalities $|P(\sigma)|$ for cells in other cuttings
$\Xi_i$, $0\leq i\leq k-1$. To this end, we can compute them in a bottom-up
manner following the hierarchical cutting (i.e., process cells of $\Xi_{k-1}$ first and
then $\Xi_{k-2}$, and so on), using the fact that $|P(\sigma)|$ is equal to the
sum of $|P(\sigma')|$ for all children $\sigma'$ of $\sigma$.
As each cell has $O(1)$ children, this step can be
easily done in $O(mr)$ time.

For each arc $h\in H$, our goal is to compute the number of points of $P$ below $h$, denoted by $N_h$, which is initialized to $0$. Starting from $\Xi_0$, suppose $\sigma$ is a cell of $\Xi_i$ crossed by $h$
and $i<k$. For each child cell $\sigma'$ of $\sigma$ in $\Xi_{i+1}$, if $\sigma'$ is below $h$, then we increase $N_h$ by $|P(\sigma')|$ because all points of $P(\sigma')$ are below $h$. Otherwise, if $h$ crosses $\sigma'$, then we proceed on $\sigma'$. In this way, the points of $P$ below $h$ not counted in $N_h$ are those contained in cells $\sigma\in \Xi_k$ that are crossed by $h$. To count those points, we perform further processing as follows.

For each cell $\sigma$ in $\Xi_k$, if $|P_{\sigma}|>n/r^2$, then we arbitrarily partition $P(\sigma)$ into subsets of size between $n/(2r^2)$ and $n/r^2$, called {\em standard subsets} of $P(\sigma)$. As $\Xi_k$ has $O(r^2)$ cells and $|P|=n$, the number of standard subsets of all cells of $\Xi_k$ is $O(r^2)$.
Denote by $H_{\sigma}$ the subset of arcs of $H$ that cross $\sigma$. Our
problem is to compute for each arc $h\in H_{\sigma}$ the number of points of $P(\sigma)$ below $h$, for all cells $\sigma\in \Xi_k$. To this end, for each cell $\sigma$ of $\Xi_k$, for each standard subset $P'(\sigma)$ of $P(\sigma)$, we solve the  subproblem on $(H_{\sigma},P'(\sigma))$ of size $(m/r,n/r^2)$ recursively as above. We thus obtain the following recurrence relation:
\begin{align}\label{equ:70}
  T(m,n)= O(mr + n\log r)+O(r^2)\cdot T(m/r,{n}/{r^2}).
\end{align}
In particular, the factor $O(n\log r)$ is due to the point location step.

\paragraph{The second algorithm.}
Our second algorithm solves the problem using duality. Recall that $Q$ is the set of centers
of the arcs of $H$. Let $P^*$ be the set of lower arcs in $\overline{C}$ defined
by points of $P$. In the dual setting, the problem is equivalent to computing for each point of $Q$
the number of arcs of $P^*$ below it. Using our Cutting Theorem, we compute in
$O(nr)$ time a
hierarchical $(1/r)$-cutting $\Xi_0,\ldots,\Xi_{k}$ for $P^*$. Consider a cell
$\sigma\in \Xi_i$ for $i<k$. For each child cell $\sigma'$ of $\sigma$ in
$\Xi_{i+1}$, let $P^*_{\sigma\setminus\sigma'}$ denote the subset of the arcs of
$P^*$ crossing $\sigma$ but not crossing $\sigma'$ and below $\sigma'$.
We store the cardinality of $P^*_{\sigma\setminus\sigma'}$ at $\sigma'$. For each cell $\sigma$ of $\Xi_k$, we store at $\sigma$ the set $P^*_{\sigma}$ of arcs of $P^*$ crossing $\sigma$. Note that $|P^*_{\sigma}|\leq n/r$. All above can be done in $O(nr)$ time.

For each point $q\in Q$, our goal is to compute the number of arcs of $P^*$ above $q$, denoted by $M_q$, which is initialized to $0$. Starting from $\Xi_0$, suppose $\sigma$ is a cell of $\Xi_i$ containing $q$
and $i<k$. We find the child cell $\sigma'$ of $\sigma$ in $\Xi_{i+1}$ that
contains $q$. We add $|P^*_{\sigma\setminus\sigma'}|$ to $M_q$ and then proceed
on $\sigma'$. In this way, the arcs of $P^*$ below $q$ not counted in $M_q$ are
those contained in the cell $\sigma\in \Xi_k$ that contains $q$. To count those
arcs, we perform further processing as follows.

For each cell $\sigma$ in $\Xi_k$, let $Q(\sigma)$ be the subset of points of $Q$ contained in $\sigma$. If $|Q_{\sigma}|>m/r^2$, then we arbitrarily partition $Q(\sigma)$ into subsets of size between $m/(2r^2)$ and $m/r^2$, called {\em standard subsets} of $Q(\sigma)$. As $\Xi_k$ has $O(r^2)$ cells and $|Q|=m$, the number of standard subsets of all cells of $\Xi_k$ is $O(r^2)$.
Our problem is to compute for each point $q\in Q(\sigma)$ the number of arcs of $P^*_{\sigma}$ below $q$, for all cells $\sigma\in \Xi_k$. To this end, for each cell $\sigma$ of $\Xi_k$, for each standard subset $Q'(\sigma)$ of $Q(\sigma)$, we solve the subproblem on $(Q'(\sigma),P^*_{\sigma})$ of size $(m/r^2,n/r)$ recursively as above. We thus obtain the following recurrence relation:
\begin{align}\label{equ:80}
  T(m,n)= O(nr + m\log r)+O(r^2)\cdot T(m/r^2,{n}/{r}).
\end{align}

\paragraph{Combining the two algorithms.}
By setting $m=n$ and applying \eqref{equ:80} and \eqref{equ:70} in succession (using the same $r$), we obtain the following
\begin{align*}
  T(n,n)= O(nr\log r)+O(r^4)\cdot T({n}/{r^3},{n}/{r^3}).
\end{align*}
Setting $r=n^{1/3}/\log n$ leads to
  $T(n,n)= O(n^{4/3}) + O((n/\log^3 n)^{4/3})\cdot T(\log^3 n,\log^3 n)$.
If we apply the recurrence three times, we can derive the following:
\begin{align}\label{equ:90}
  T(n,n)= O(n^{4/3}) + O((n/b)^{4/3})\cdot T(b,b),
\end{align}
where $b=(\log\log\log n)^3$.

\paragraph{Solving the subproblem $T(b,b)$.}
Due to the property that the value $b$ is tiny, we next show that after $O(2^{\poly(b)})$ time preprocessing, where $\poly(\cdot)$ represents a polynomial function, each $T(b,b)$ can
be solved in $O(b^{4/3})$ time. Since $b=(\log\log\log n)^3$, we have
$2^{\poly(b)}=O(n)$. As such, applying~\eqref{equ:90} solves $T(n,n)$ in
$O(n^{4/3})$ time, which proves Theorem~\ref{theo:70} for the symmetric case; we
will solve the asymmetric case at the end by using the symmetric case
algorithm as a subroutine.

For notational convenience, we still use $n$ to represent $b$. We wish to show that after $O(2^{\poly(n)})$ time preprocessing, $T(n,n)$ can be solved in $O(n^{4/3})$ time. To this end, we show that $T(n,n)$ can
be solved using $O(n^{4/3})$ comparisons, or alternatively, $T(n,n)$ can be
solved by an algebraic decision tree of height $O(n^{4/3})$. Since building the algebraic decision tree can be done in $O(2^{\poly(n)})$ time, each $T(n,n)$ can be solved in $O(n^{4/3})$ time using the decision tree.
As such, in what follows, we focus on proving that $T(n,n)$ can be solved using $O(n^{4/3})$ comparisons.

Observe that it is the point
location step in the recurrence~\eqref{equ:70} in our first algorithm that prevents us from obtaining an
$O(n^{4/3})$ time bound for $T(n,n)$ because each point location introduces an additional logarithmic factor. To overcome the issue, Chan and Zheng~\cite{ref:ChanHo22} propose an {\em $\Gamma$-algorithm
framework} for bounding decision tree complexities (see Section~4.1~\cite{ref:ChanHo22} for the details).
To reduce the complexity for the point location step in their algorithm, Chan and Zheng proposed a {\em basic search lemma} (see Lemma~4.1~\cite{ref:ChanHo22}). We can follow the similar idea as
theirs (see Section~4.3~\cite{ref:ChanHo22}). We only sketch the main idea below and the reader can refer to~\cite{ref:ChanHo22} for details.

Following the notation in \cite{ref:ChanHo22}, for any operation or subroutine of the algorithm,
we use $\Delta\Phi$ to denote the change of the potential $\Phi$. Initially $\Phi=O(n\log n)$. The potential $\Phi$ only decreases during the algorithm. Hence, $\Delta\Phi\leq 0$ always holds and the sum of
$-\Delta\Phi$ during the entire algorithm is $O(n\log n)$.

We modify the point location step of the first algorithm with the following change. To find the cell of $\Xi_k$ containing
each point of $P$, we apply Chan and Zheng's basic search lemma on the $O(r^2)$ cells of
$\Xi_k$, which needs $O(1-r^2\Delta\Phi)$ comparisons (instead of
$O(\log r)$).
As such, excluding the $O(-r^2\Delta\Phi)$ terms, we obtain a new recurrence for the number of comparisons for the first algorithm:
\begin{align}\label{equ:110}
  T(m,n)= O(mr + n)+O(r^2)\cdot T(m/r,n/r^2).
\end{align}

Using the same $r$, we stop the recursion until $m=\Theta(r)$, which is the
base case. In the base case we have $T(m,n)=O(n+r^2)$ (again excluding the term
$O(-r^2\Delta\Phi)$). Indeed, we first construct the arrangement of the
$m=\Theta(r)$ upper arcs in $O(r^2)$ time and then apply the basic search lemma to
find the face of the arrangement containing each point. The above obtains the
total number of points inside each face of the arrangement. Our goal is to
compute for each upper arc the number of points below it. This can be done in
additional $O(r^2)$ time by traversing the arrangement (see
Lemma~5.1~\cite{ref:ChanHo22}).
In this way, the recurrence~\eqref{equ:110} solves to $T(m,n)=(m^2+n)\cdot 2^{O(\log_rm)}$.
By setting $r=m^{\epsilon/2}$, we obtain the following bound
on the number of comparisons excluding the term
$O(-m^{\epsilon}\Delta\Phi)$.
\begin{align}\label{equ:120}
T(m,n)=O(m^2+n).
\end{align}

Now we apply recurrence \eqref{equ:80} with $m=n$ and $r=n^{1/3}$ and obtain the
following
\begin{align*}\label{equ:120}
T(n,n)=O(n^{4/3}) + O(n^{2/3})\cdot T(n^{1/3},n^{2/3}).
\end{align*}
Applying \eqref{equ:120} for $T(n^{1/3},n^{2/3})$ gives $T(n,n)=O(n^{4/3})$ with the excluded terms sum
to $O(n^{\epsilon}\cdot n\log n)$. As such, by setting $\epsilon$ to a small value (e.g., $\epsilon=1/4$), we conclude that $T(n,n)$ can be solved using $O(n^{4/3})$ comparisons.

\paragraph{The asymmetric case.}
The above proves Theorem~\ref{theo:70} for the symmetric case with an $O(n^{4/3})$ time algorithm for solving $T(n,n)$. For the asymmetric case, depending on whether $m\geq n$, there are two cases.

\begin{enumerate}
  \item If $m\geq n$, depending on whether $m< n^2$, there are two subcases.

  \begin{enumerate}
    \item If $m<n^2$, then let $r=m/n$, and thus $m/r^2 = n/r$. Applying \eqref{equ:80} and solving $T(m/r^2,n/r)$ by the symmetric case algorithm give $T(m,n)=O(m\log n+n^{2/3}m^{2/3})$.

    \item If $m\geq n^2$, then we solve the problem in the dual setting for $Q$ and $P^*$ as discussed in the above second algorithm, i.e., compute for each point of $Q$ the number of arcs of $P^*$ below it. We first construct the arrangement $\calA$ of the arcs of $P^*$ in $O(n^2)$ time and then build a point location data structure on $\calA$ in $O(n^2)$ time~\cite{ref:KirkpatrickOp83,ref:EdelsbrunnerOp86}. Next, for each point of $Q$, we find the face of $\calA$ that contains it in $O(\log n)$ time using the point location data structure.
        In addition, it takes $O(n^2)$ time to traverse $\calA$ to compute for each face of $\calA$ the number of arcs below it.
        The total time is $O(n^2+m\log n)$, which is $O(m\log n)$ as $m\geq n^2$.
  \end{enumerate}
   Hence in the case where $m\geq n$, we can solve the problem in $O(m\log n+n^{2/3}m^{2/3})$ time.

  \item If $m< n$, depending on whether $n< m^2$, there are two subcases.
  \begin{enumerate}
	\item If $n<m^2$, then let $r=n/m$, and thus $m/r=n/r^2$. Applying
	\eqref{equ:70} and solving $T(m/r,n/r^2)$ by the symmetric case algorithm give $T(m,n)=O(n\log m+n^{2/3}m^{2/3})$.

	\item If $n\geq m^2$, then we first construct the arrangement $\calA$ of the arcs of $H$ in $O(m^2)$ time and then build a point location data structure on $\calA$ in $O(m^2)$ time~\cite{ref:KirkpatrickOp83,ref:EdelsbrunnerOp86}. Next, for each point of $P$, we find the face of $\calA$ that contains it in $O(\log m)$ time using the point location data structure.
        In addition, as discussed above, it takes $O(m^2)$ time to traverse $\calA$ to compute for each arc of $H$ the number of points of $P$ below it.
        The total time is $O(m^2+n\log m)$, which is $O(n\log m)$ as $n\geq m^2$.
  \end{enumerate}
  Hence in the case where $m<n$, we can solve the problem in $O(n\log
  m+n^{2/3}m^{2/3})$ time.
\end{enumerate}

This proves Theorem~\ref{theo:70}.
\end{document}